\documentclass[12pt]{article}

\usepackage{newtxtext,newtxmath}
\usepackage{graphicx}

\usepackage[letterpaper,margin=1in]{geometry}

\renewenvironment{abstract}
	{\quotation}
	{\endquotation}

\date{}

\makeatletter
\renewcommand{\fnum@figure}{\textbf{Figure \thefigure}}
\renewcommand{\fnum@table}{\textbf{Table \thetable}}
\makeatother

\usepackage{scicite}

\usepackage{url}

\usepackage{mycommands}
\usepackage{anyfontsize}

\def\scititle{A Statistical Test for the Benefits of Personalizing Interventions
}
\title{\bfseries \boldmath \scititle}

\author{
	Zhaoqi Li$^{1\ast\dagger}$, 
	Emma Brunskill$^{1\ast\dagger}$
	\and
	\small$^{1}$Computer Science Dept., Stanford University, Stanford, CA, U.S.A.%\and
\and
	\small$^\ast$Corresponding authors. Email: zli9@stanford.edu, ebrun@stanford.edu\and
	\small$^\dagger$These authors contributed equally to this work.
}

\begin{document} 

% Insert the title and author list
\maketitle

% Abstract, in bold
% There are strict length limits, and not all formats have abstracts.
% Consult the journal instructions to authors for details.
% Do not cite any references in the abstract.
\begin{abstract} \bfseries \boldmath
From medicine to marketing to social sciences, the promise of tailoring interventions to individuals is undeniable. However, practical applications force weighing personalization's potential benefits with its possible increased cost and fragility. We introduce a statistical hypothesis test that evaluates, given historical data, evidence that a personalized intervention policy’s performance will surpass deploying the best single intervention. The test maintains strict type-I error control while achieving asymptotic normality with the minimal possible variance under specified conditions. Results on diverse datasets from job training, depression treatment, education and recommendation systems demonstrate the test’s versatility and its superior performance over alternatives. This test can support decision-makers throughout the intervention sciences by providing a simple and powerful quantification of the potential benefits of personalization. % 121 words
\end{abstract}

% The first paragraph of any Science paper does NOT have a heading
% Nor is it indented
\noindent
An essential component of human inquiry is how to choose interventions to maximize expected outcomes. The applications are vast— from job training to online advertising to medical treatment — and span diverse disciplines — from social sciences to artificial intelligence to statistics.  Because the same intervention can vary in its impact on individuals with different characteristics, a reasonable presumption is that, when possible, intervention decision policies should be personalized, for example by providing different supports to low- and high-performing students. This is a central tenet of precision medicine, automated customer marketing, and research areas such as contextual multi-armed bandits and reinforcement learning. Note that for a personalized decision policy to increase expected outcomes over providing a single intervention to all individuals, there must be different subgroups that benefit most from differing interventions: the presence of heterogeneous treatment effects is necessary but not sufficient  (see Fig.  \ref{fig:personalization_vs_hte}). 

\begin{figure} % Do NOT use \begin{figure*}
	\centering
\includegraphics[width=\textwidth]{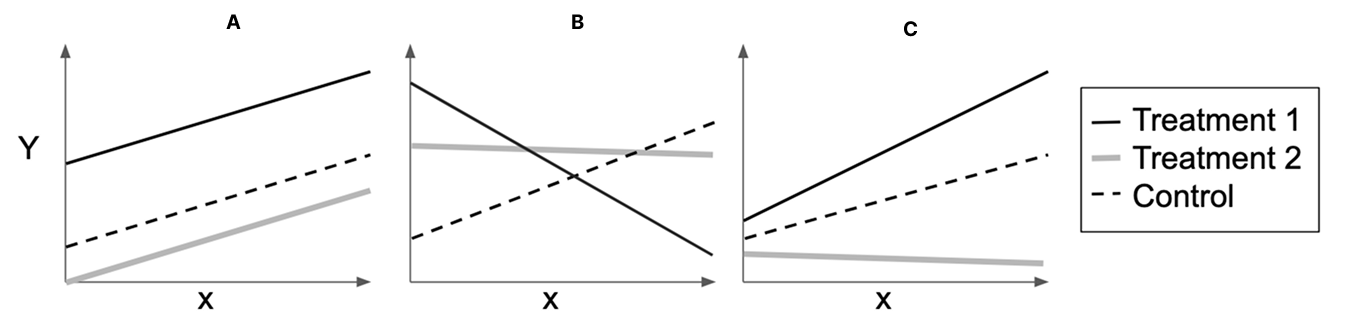} \caption{\textbf{Heterogeneous Treatment Effects are Necessary, But Not Sufficient, For Personalization to Improve Expected Outcomes.} Consider a setting where there are three interventions, $X$ is a single feature describing an individual (e.g., age or blood pressure), and $Y$ is the outcome. In (A) treatment 1 is best for all individuals, and it is equally better than the two alternatives for any individual. In (C) there are heterogeneous treatment effects-- treatment 1 provides a larger benefit over treatment 2 and control for individuals with large $X$ covariate values, but treatment 1 is still the best intervention for all individuals. In (B) there is a personalization effect: treatments 1, 2 and control each maximize outcomes for different subgroups of individuals. }
\label{fig:personalization_vs_hte} 
\end{figure}
Personalization, however, can be both more costly and brittle compared with providing a single intervention for all individuals. It often involves additional logistical challenges, such as resource considerations (different medications can require different cooling infrastructure) or can have lower implementation fidelity, due to more complex software infrastructure or more complicated training (for example, training novice tutors to use different pedagogical strategies with different student groups). In addition, whereas in theory personalized decision policies are never worse than providing all individuals with the same intervention, in practice personalized policies may have worse outcomes than assigning all individuals to a single intervention. This is because, assuming that personalized policies and a single best intervention are both obtained by  analyzing evidence from finite datasets, the amount of data needed to accurately learn the best intervention for different subsets of individuals is generally much more than the sample size needed to identify the best single intervention across all individuals.  

It therefore would often be highly useful, to both researchers and stakeholders, to be able to take an input dataset and statistically test whether there is evidence of an expected outcome benefit to personalizing interventions over providing a single (best) intervention to all individuals. Related questions have been studied in specific disciplines, largely disjointly  -- from health statisticians studying qualitative interaction effects~\cite{peto1982statistical,gail1985testing,silvapulle2001tests,li2006detecting,gunter2011variable,dusseldorp2014qualitative,shi2020sparse} to econometricians advancing hypothesis testing for conditional average treatment effects~\cite{delgado2013conditional,chang2015nonparametric,hsu2017consistent,lee2013testing}. However, such work has been limited because they were (i) restricted to binary decisions, (ii) primarily focused on when a single covariate was available for personalization, and (iii) lacked  guarantees on statistical eﬃciency. To fill these gaps in the literature, here we present a new, simple statistical test for the benefit of personalization called the K-fold Personalization Test (KPT). The KPT is statistically sound and can handle multiple potential interventions (beyond binary treatment and control) and multiple covariates that describe individuals. As we illustrate through experiments on datasets from four domains, education, medicine, job training, and joke recommendations, our test is applicable to a broad range of disciplines and could help to inform practitioners and researchers as they seek to create and deploy intervention policies to maximize desired outcomes.

\begin{figure} % Do NOT use \begin{figure*}
	\centering
	\includegraphics[width=\textwidth]{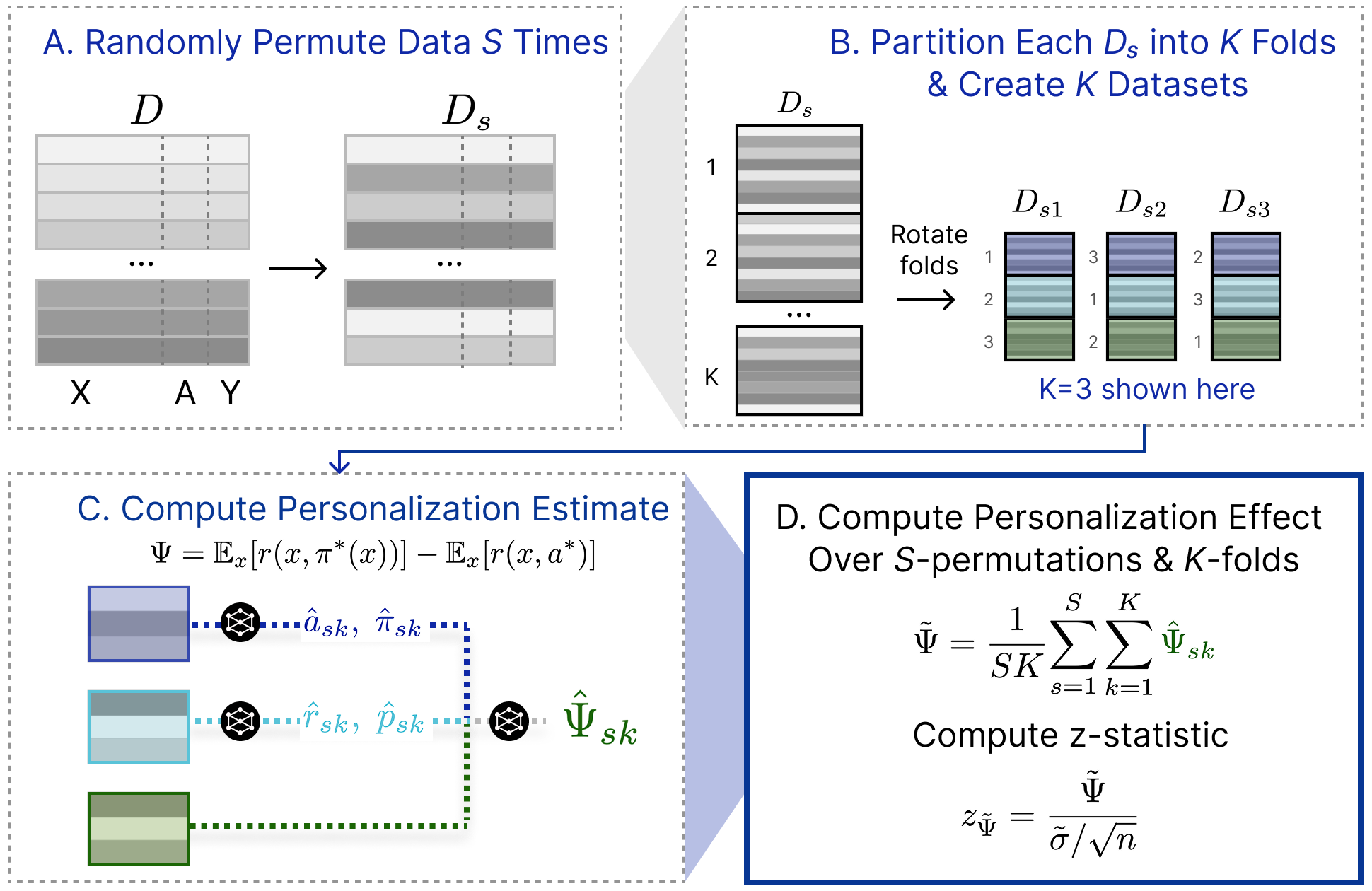} 
	% Captions go below figures
	\caption{\textbf{K-fold Personalization Test.} (A) The input dataset consists of $X$, features; $A$, intervention; and $Y$, outcome and is is randomly permuted $S$ times to create $S$ different $D_s$ datasets. (B) Each permuted dataset $D_s$ is split into $K$-folds. These folds are then used to create $K$ additional datasets per $D_s$: $D_{s1},\ldots,D_{sK}$ (for clarity, $K=3$ is shown).  (C) For a particular $D_{si}$, one or more folds are used to learn a  personalized policy $\hat{\pi}$ and a single best intervention $\hat{a}$, and some of the other folds are used to learn an outcome model $\hat{r}(x,a)$ and a propensity model $\hat{p}(a|s)$. These are then used to compute a personalization effect on another data fold. Since each data point is used for each subtask, this approach allows for data efficiency and simplifies the formal analysis.  (D) The final personalization effect is computed by aggregating across all folds and all data permutations, and the statistical test is computed using this final estimate.}%For visual clarity we show three folds in this figure but our KPT uses six-folds.}
	\label{fig:alg_pk} 
\end{figure}

\section*{Method}
We assume there exists a set of $N$ data points $(X_i,A_i,Y_i)$, where individual $i$ is described by a vector of features $X_i \in \mathcal{X}^d$, who received an intervention\cite{note:propensity}
$A_i \in \mathcal{A}$ and experienced an outcome that is an unknown function of their features, intervention and noise: $Y_i \sim r(X_i,A_i) + \epsilon$ (e.g. a 20-year-old male received job training and later earned \$40k). 
A personalized policy $\pi: \mathcal{X} \to \mathcal{A}$ maps features of an individual to an intervention $a$. For example, it might provide people with high school education job training, while others receive no intervention. We consider policies that lie in a statistical model class, such as a deep neural network, $\pi \in \Pi$. We assume a known distribution over individuals $p(x)$. 

Our objective is to develop a statistical test that takes in the dataset $(X,A,Y)$ and evaluates whether there is evidence that the expected personalization benefit of using a personalized policy over a single intervention for all individuals, $\psi$, is greater than 0: 
\begin{equation}
\psi := \max_{\pi \in \Pi} \sum_x p(x) r(x,\pi(x)) - \max_{a\in\mathcal{A}} \sum_x p(x) r(x,a) > 0. 
\label{eqn:psi}
\end{equation}
Note that the standard equivalence between statistical hypothesis testing and confidence interval construction implies that designing a strong hypothesis test is equivalent to finding an estimator of the mean outcome difference with the tightest confidence interval, which was  our focus. 

Framed in this way (Eqn~\ref{eqn:psi}), our setting appears similar to the ubiquitous task of testing whether an experimental treatment yields the same average outcome as a control condition. However, a key challenge for creating a statistically efficient test for a personalization benefit is that the "conditions" are unknown in advance -- the first term in Equation~\ref{eqn:psi} involves a maximum over personalized policies, and the second involves a maximum over all interventions (see section \ref{sec:notations} of the supplementary materials for further discussion). 

Our core insight is that a simple efficient personalization estimator can be achieved through careful data partitioning and reuse, and a focus on the core estimand of interest-- the personalization effect. In treatment effect estimation~\cite{luedtke2016statistical,chernozhukov2018double}, the treatment effect is the estimand and statistical models of outcomes and the probability of different individuals receiving different interventions, are treated as nuisance parameters that are useful only in enabling better estimates of the desired treatment effect. Our work leverages and advances this seminal idea by additionally treating the personalized policy and best intervention as nuisance parameters.

Our KPT is shown in Figure~\ref{fig:alg_pk} and details of our methods are provided in Section~\ref{sec:kfold_full} of the SM. Our personalization effect estimator per random permutation-fold, $\Psi_{sk}$, is a doubly-robust-style estimator in whichthe personalized policy $\hat{\pi}$ and best single intervention $\hat{a}$ are learned on a subset of folds not including $k$; the outcome models $\hat{r}$ and propensity models $p$ are learned on the remaining folds (also excluding $k$); and the final fold $k$ data $(x_i,a_i,y_i)$ are also used:
\begin{eqnarray}
\hat{\psi}_{sk} &=&\frac{1}{N_{sk}}\sum_{i=1}^{N_{sk}} \hat{r}_{sk}(x_i,\hat{\pi}_{sk}(x_i)) - \hat{r}_{sk}(x_i,\hat{a}_{sk}) + \frac{\mathbf{1}\{\hat{\pi}_{sk}(x_i)=a_i\}}{\hat{p}_{sk}(a_i|x_i)}(y_i-\hat{r}_{sk}(x_i,\hat{\pi}_{sk}(x_i))) \nonumber\\
& & - \frac{\mathbf{1}\{\hat{a}_{sk}=a_i\}}{\hat{p}_{sk}(a_i|x_i)}(y_i-\hat{r}_{sk}(x_i,\hat{a}_{sk})). \label{eqn:aipw}
\end{eqnarray}

We now briefly comment on several aspects of our KPT. First, our method carefully uses data to perform both policy learning and evaluation without sacrificing asymptotic efficiency by extending cross-fitting~\cite{chernozhukov2018double}, a key idea for enabling statistical efficiency when using data partitioning with complex statistical models. Second, KPT  can immediately benefit (Fig.~\ref{fig:alg_pk} step C) from advances in the active research areas of best intervention~\cite{manski2004statistical} and contextual policy learning~\cite{qian2011performance,swaminathan2015batch,athey2021policy,zhou2023offline,kitagawa2018should,luedtke2020performance}. Third, unlike some prior personalization effect estimation methods that extract decision policies by thresholding conditional average treatment effects\cite{shi2020sparse}, our procedure easily allows for different covariates to be used for modeling outcomes and intervention decision policies. For example, it might be of interest to personalize cancer treatment based on someone's initial tumor size but not their address, even if both covariates might help predict their future health outcomes. 
Fourth, though cross fitting helps our method achieve statistical efficiency, it is known that data splitting methods that compute $p$-values or confidence intervals can be sensitive to the particular random partition, which has prompted methods that leverage repeated splits\cite{meinshausen2009p,chernozhukov2018double}. We will show that this substantially improves the stability of our KPT over alternative baselines (including~\cite{imai2023experimental}), while preserving theoretical properties of the test (see supplementary materials Section 4.4 for an additional discussion on split aggregation methods). Finally,  experts in policy learning and evaluation may reasonably be concerned about the potential tension between having sufficient data to learn a good personalized policy and enough additional data to evaluate its impact. Although our theoretical results do not rely on how many folds are used for policy learning (versus outcome modeling learning, etc), when data are limited this is likely to be an important choice. In our experimental results, we use $K=6$ and use half the folds to do personalized policy and best intervention learning in each split-fold estimate. 

\section*{Results}
\subsection*{Properties of the Proposed Personalization Test}
The quality of a proposed statistical test is typically evaluated by its Type I error (how often it rejects the null hypothesis when the null hypothesis is true), its power (how often it correctly rejects a false null hypothesis), its reliability, and its stability.  We will demonstrate empirically the stability of our approach relative to alternatives, and we will use simulations to evaluate its reliability. These results and all formal statements and proofs are in section~\ref{sec:proofs} of the SMsupplementary materials.

We mathematically prove that our proposed test has asymptotically zero Type I error, under standard assumptions and the fairly weak additional assumption that the best single intervention is unique\cite{note:nonzeroutility}. 
This result ensures that, under these assumptions, our approach will not falsely suggest that personalization is beneficial when no personalization effect is present. 

In parametric settings, uniformly most powerful tests are used as the gold standard for statistical tests that guarantee bounded Type I error while maximizing power. However, in semi-parametric statistics (commonly used to estimate treatment effects) the standard alternative objective is to provide statistical estimators-- and tests derived from those estimators--  that obtain the semiparametric efficiency bound. Such  estimators have the minimum asymptotic variance among all regular estimators. An estimator with a minimum variance will maximize the chance of detecting a positive personalization benefit when present, thus  maximizing power. Under stronger assumptions than those that we required to guarantee our Type I results, we proved that our estimator is asymptotically normal and semiparametrically efficient. Specifically, we assumed that personalized policy learning could converge at a fast rate, which has been considered under margin conditions. To our knowledge, our KPT is the first semiparametrically efficient test for personalization. 

\begin{figure} % Do NOT use \begin{figure*}
	\centering
    \includegraphics[width=0.88\textwidth]{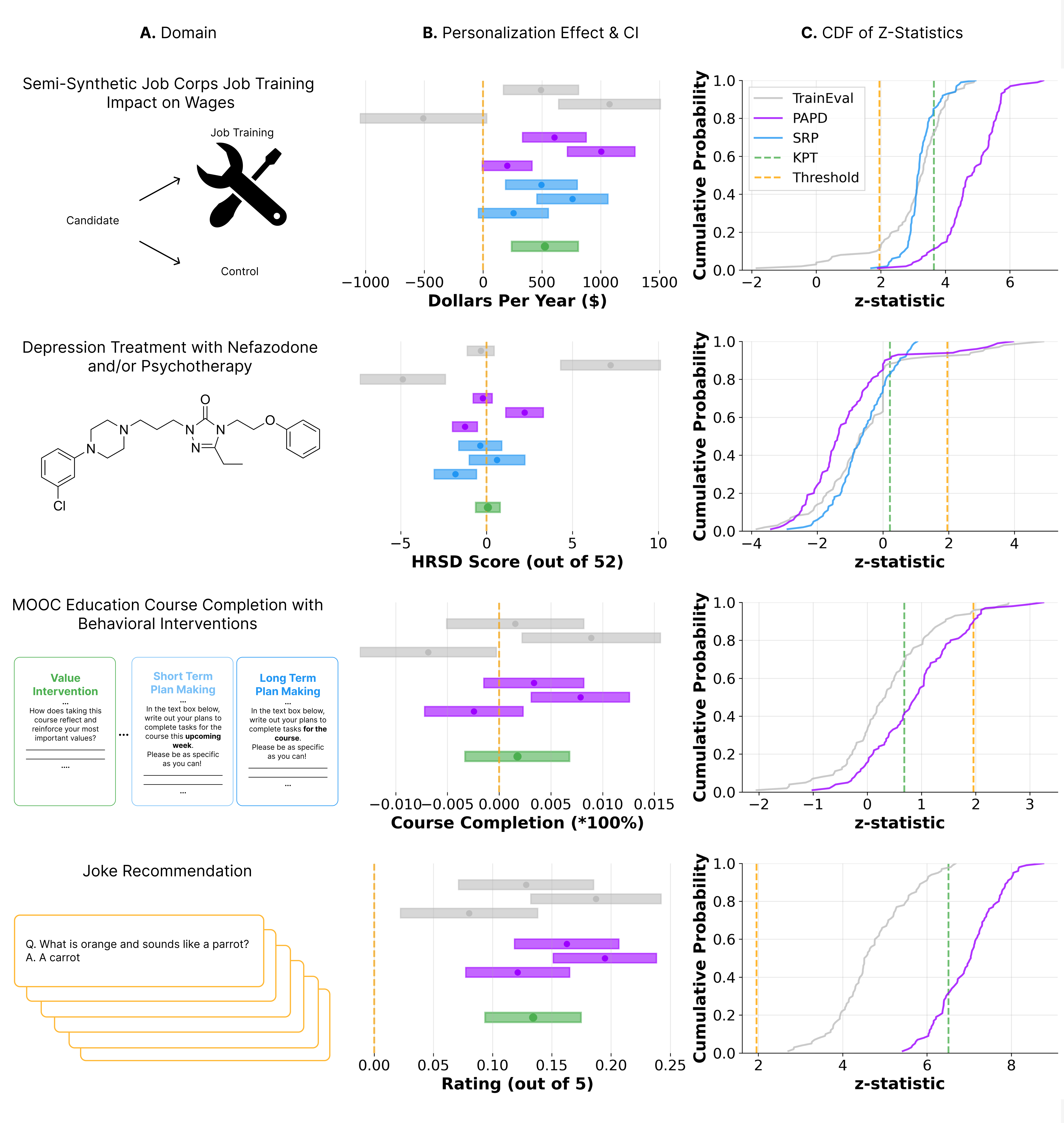} % for an image file named example_figure.*
	% Pick an appropriate width - in print, figures are usually one or two columns wide, which can
	% be approximated by 0.3\textwidth or 0.6\textwidth respectively. Use appropriate label sizes.

	% Captions go below figures
	\caption{Each row shows a dataset setting. The second column shows  
    the confidence interval returned by our K-fold personalization estimate, and three sample confidence intervals output by 
    TrainEval, PAPD and SRP across the 100 random partitions of the dataset: the CI with the lowest lower bound, the CI with the highest upper bound, and one randomly selected CI. The third column shows the 
    cumulative distribution function of $z$-statistics across the 100 random partitions for the baseline methods, as well as the critical value (in orange) for significance  for the $z$-statistic  (1.96, based on $\alpha$ = 0.05) and our KPT $z$-statistic. In contrast to our KPT, prior methods are often highly sensitive to the particular random partition, frequently return wider confidence intervals, and are not always suitable for settings with many interventions (such as the Joke Recommendation and MOOC Education domains).}
	\label{fig:ci_cdf} 
\end{figure}

\subsection*{Estimating the Personalization Benefit in Diverse Datasets}
To illustrate the broad range of potential applications in which assessing the impact of personalization is relevant, we applied our KPT to several real datasets and one semi-synthetic dataset. In three of the four settings the original data came from a randomized controlled trial (RCT), and for the joke recommendation, we mimicked an RCT given the full user ratings. In each case, our KPT produced a confidence interval for the personalization effect and a $p$-value given the computed $z$-statistic  (which we compared with  the critical value of a normal distribution
using $\alpha$ = 0.05: 1.96). We rejected the null if our test statistic was larger than 1.96. We compared with the strongest prior method developed specifically for testing the average positive outcome benefit of personalization~\cite{shi2020sparse}. This is designed for binary treatments and uses a sparse random projection (SRP) to handle high-dimensional covariates. We also compared our method with an influential method for estimating the population average prescriptive effect difference (PAPD) between two different learned personalized policies using an   empirical average and centered outcome, and 
was analyzed under Neyman’s repeated sampling framework~\cite{imai2023experimental}. PAPD can be specialized to our setting, in which one policy does not personalize. We additionally compared our method with a simple natural baseline we denote as TrainEval, in which the data are randomly partitioned into two equal parts, a best single intervention and a personalized policy are fit on half the data, and the other half is used to estimate the expected utility of that personalization policy over the single intervention. In contrast to our KPT approach, these alternative methods use a single random partition of the dataset, so we plot sample confidence intervals for their estimated personalization effects for different random splits and/or partitions, and plotted the cumulative distribution functions (CDFs) of their $z$-statistics across random partitions. This allowed us to compare the widths of the resulting confidence intervals, and the stability of these methods. Note that our KPT method is designed to operate over many random partitions, as shown in Figure 2 (and see Algorithms 1 and 2 in section~\ref{sec:kfold_full} of the supplementary materials).). In all cases, KPT has confidence intervals comparable to or smaller than those of TrainEval, SRP, and PAPD. KPT is also more broadly applicable than SRP which assumes binary treatment and control. Our KPT estimates are also more stable than all baselines, the confidence intervals of which sometimes do not overlap across splits. In addition, we conducted  experiments in fully synthetic environments which demonstrate that: (i) our approach is reliable across random samplings from the same data distribution, and (ii) although our theoretical efficiency results require fast policy learning, KPT can still work well even when data are limited and policy learning is imperfect. See section~\ref{sec:simulation_reliability} in the supplementary materials 
for full details. Code to reproduce the results is available online. 
\\
\\
\noindent{\textbf{Semi-Synthetic Job Corps Job Training.}}\\ 
Job training programs can increase subsequent wages and  education, and reduce criminal activity~\cite{schochet2008does} but they are expensive, and may not benefit all subgroups. Prior research on Job Corps, a national US training program for low-income early adults, found benefits for most participants but participants aged 18 to 19 had slightly lower (but similar) earnings compared with the control group~\cite{schochet2008does}. Here we created a semi-synthetic setting to test and compare personalization estimators\cite{note:semisynth}. 
We used 12 covariates including gender, race, age group, prior arrests, parental status, and education level, and the outcome was monthly self-reported wages 4 years after program enrollment. All covariates and outcome values for the 10214 trial participants were unmodified except for the outcomes for participants aged 18 to 19 at enrollment, for whom we added a synthetic positive wage offset of \$30/week in the control group, and a negative wage offset of \$5/week for those in the Job Corps group. We considered a personalized intervention policy class that can make different decisions for the three age brackets analyzed in prior work~\cite{schochet2008does}, showcasing how KPT can use different features for the personalized policy and outcome modeling. 

Our KPT method successfully identified a significant benefit of personalizing program allocation based on age over providing all participants with the same (program or control) intervention, with 
$p< 10^{-3}$, a $z$-statistic of 3.645, and a confidence interval for the estimated personalization utility of \$ 10.097 $\pm$ 2.77  
on weekly self-reported wages 4 years after program enrollment, or about \$525 additional wages per year, per individual. The confidence interval of the K-fold Personalization estimate was generally similar or narrower than alternate methods SRP, PAPD, and TrainEval, which were also notably sensitive to the specific random partition of the data (Fig.~\ref{fig:ci_cdf} column B, row: "Semi-Synthetic Job Corps Job Training Impact on Wages"). In addition, across 100 random splits, TrainEval failed to report significance in 11\% of runs (see Fig.~\ref{fig:ci_cdf} column C).   
\\
\\
\noindent{\textbf{Depression Treatment with Nefazodone and/or Psychotherapy}}.
\\
Millions of people suffer from depression, so it is of significant interest whether different interventions may benefit different individuals even with the same clinical diagnosis. Prior work ~\cite{zhao2012estimating,shi2020sparse} has considered learning individualized policy interventions given data from a clinical trial in which patients with nonpsychotic chronic major depression were randomized to receive the medication nefazodone, the Cognitive Behavioral-Analysis System of Psychotherapy, or both, for 12 weeks~\cite{keller2000comparison}. The outcome of interest was the Hamilton Rating Scale for Depression at the final follow-up. Prior policy learning work learned an individualized policy (based on 50 baseline covariates, using data from 647 patients) that provided the combination treatment to all individuals~\cite{zhao2012estimating}, suggesting that personalization might not yield additional benefit. Subsequent work~\cite{shi2020sparse} found no personalization benefit using their SRP personalization test; that study assumed the combination was the best single intervention and compared it with a policy that either allocated patients to the combination or Nefazodone, or another policy which allocated patients to the combination or CBASP. We used this same dataset.

We applied our KPT method, which does not require as input the best population-level intervention, and can automatically handle more than two treatments. Results are presented in Figure~\ref{fig:ci_cdf}, row: "Depression Treatment with nefazodone and/or Psychotherapy." Our KPT method returned a non-significant $p$-value ($p > 0.05$), consistent with prior results showing that there is no evidence of personalization benefit \cite{shi2020sparse,zhao2012estimating}.  TrainEval again shows substantial variability, and would (likely falsely, given prior results) return a significant positive estimate of personalization in 8\% of the random partitions (Fig.~\ref{fig:ci_cdf}, Column C).  PAPD is also quite sensitive to the random parition, and in 7\% of random paritions it would estimate a significant benefit of personalization.
\\
\\
\noindent{\textbf{Massive Open Online Classes (MOOCs) Education Course Completion with Behavioral Interventions}.}\\
%Behavioral Science Interventions for MOOC Education Student Retention}.}\\ 
Student completion rates of massive open online classes (MOOCs) are low on average, approximately 10\% ~\cite{kizilcec2020scaling}, and behavioral scientists have sought to design interventions to increase completion~\cite{kizilcec2017eight}. Such interventions could potentially be more effective if personalized.  We use data from a large randomized controlled trial consisting of 247 online courses over 2.5 years  which found limited impact of five behavioral science interventions-- plan-making and value-relevance-- over a control condition on course completion\cite{kizilcec2020scaling} . The authors created a few personalized policies but did not observe significant gains over a single intervention, nor did they explicitly quantify uncertainty or statistical significance. A statistical test can formally assess whether personalized assignment of these interventions increases course completion relative to providing all individuals with a single intervention. The dataset consists of 199517 participants from wave 1 and 2 and we used six covariates. 

Given that there are five potential interventions, we restricted our analysis to our method KPT and the benchmarks PAPD and  TrainEval (see Figure~\ref{fig:ci_cdf} column B, row: ``MOOC Education Course Completion with Behavioral Interventions"). TrainEval is again highly variable, returning a significant impact of personalization in $5\%$ of the random paritions, and its confidence intervals did not always overlap. PAPD had tighter intervals and was less variable, but still showed notable sensitivity to the random partitions, estimating a significant benefit in  11\% of the random partitions. Our KPT method produces a confidence interval overlapping 0 estimated personalization effect, with a $p$-value $>0.1$. Our confidence interval on the personalization effect on course completion was $0.18\%\pm 0.26\%$, suggesting that, for the decision policy class considered, any personalization effect, if present, would be very small. 
\\
\\
\noindent{\textbf{Joke Recommendation}.}\\  To explore personalization benefits in the high-dimensional contexts common to recommendation systems and online marketing, we studied joke recommendations with the widely studied Jester dataset \cite{goldberg2001eigentaste}. Following the preprocessing methods in Kong et al. \cite{kong2020sublinear}, we kept the 10 highest-rated jokes as potential interventions, and transformed $(N=48447)$ user ratings for the remaining 90 jokes into a 100-dimensional covariate feature vector, to mimic a typical recommendation system that leverage prior user ratings to inform new content recommendations. 

Our results, shown in Figure~\ref{fig:ci_cdf} (columb B, row: "Joke Recommendation"), confirm prior work that estimated a higher utility for a personalized policy~\cite{kong2020sublinear}. Our KPT test identifies a significant benefit of personalization in this setting relative to providing all individuals with the estimated highest overall rated joke, with a $p$-value of $<10^{-3}$, a $z$-statistic of $6.51$ and an estimated personalization benefit of $0.134 \pm 0.0206$ out of 5. SRP~\cite{shi2020sparse} is not directly applicable to this setting because their test assumes a binary intervention. TrainEval and PAPD also found a significant effect of personalization across all random partitions, although their specific estimates still varied notably across partitions, especially for TrainEval.

\subsection*{Discussion}
From medicine to marketing to social science, there is significant interest in devising personalized intervention policies that allocate different interventions to different individuals. Understanding whether such personalized policies yield any benefit in expected outcomes over a single intervention provided to all individuals is important, especially because personalization can involve significant additional costs and require larger datasets. %This paper introduced the K-fold Personalization Test, a simple statistical test that, given unconfounded historical data, provides asymptotic guarantees on type-I error, and, under some assumptions, is asymptotically normal with minimum variance. 
Our KPT builds on ideas spanning econometrics, biostatistics and machine learning, integrating easily with existing and future intervention policy-learning algorithms and is the first, to our knowledge, that extends to scenarios beyond binary interventions and limited covariates. As we prove formally, and demonstrate empirically across a diverse range of problem settings, the KPT demonstrates strong performance over alternatives. 

One limitation of our test is that it estimates the benefit of personalization with respect to a particular (stakeholder-specified) intervention policy class (or set of classes). 
A failure to reject the null hypothesis indicates limited evidence in the dataset for personalization within the chosen policy class, rather than proof against any realizable form of personalization. We consider sensitivity to the policy class selection in section~\ref{sec:policy_class_choice} of the supplementary materials. Another limitation is that our test identifies the utility of personalization but does not return a single personalized policy: if personalization is significant, then one approach is to perform policy learning over the full dataset to extract a single policy for future use. 

There are also numerous additional directions for future work. These include formal guarantees for settings in which multiple interventions appear equivalent at the population level, due to personalization effects canceling out at the subgroup level (see discussion in the supplementary materials section 4.8). Such settings are often of substantial scientific interest. For example, although at the population level, a treatment and control intervention could have  equivalent average outcomes, a treatment might provide benefits for women while causing equally large harms for men. Another important question is whether the stability and efficiency of our KPT method can be replicated with a method that requires substantially less computation.  In addition, KPT is suitable for discrete action spaces, and it is of interest to create personalization tests suitable for continuous intervention spaces (e.g. robotics, medication dosage) or extremely large intervention spaces (e.g. generative AI-based individualized ads). Nevertheless, given the ubiquity of policy design decisions, the KPT offers a principled and needed tool for evaluating when and where personalization can deliver value. 		

%There are also numerous additional directions for future work. These include formal guarantees for settings where multiple interventions appear equivalent at the population level, due to  ersonalization effects canceling out at the subgroup level. Still, given the ubiquity of policy design decisions, the K-fold Personalization Test offers a principled and needed tool for evaluating when and where personalization delivers value. 

%[what other research questions remain unanswered. Did your results raise new questions?]

% [Add 2-3 more paragraphs here to develop the discussion further. Compare your results with those reported by other investigators. Report conditions that might limit the extent of legitimate generalizations or otherwise qualify your inferences.]

% [Did the sample characteristics across those 4 datasets differ from other populations to which you might want to generalize. Did any characteristics of your method designing the KPT influence outcomes? Any factors that might have operated to produce atypical results?]

%%%%%%%%%%%%%%%% REFERENCES %%%%%%%%%%%%%%%

\clearpage % Clear all remaining figures and tables then start a new page

% The list of references goes after the main text and before the acknowledgements
% When preparing an initial submission, we recommend you use BibTeX, like this:
%
\bibliography{science_submission/refs_contextualization} % for a file named science_template.bib

\bibliographystyle{sciencemag}

\section*{Acknowledgments}
We wish to thank Daniel Klein and Bruce Arnow for sharing the Nefazodone dataset\cite{keller2000comparison}. 
\paragraph*{Funding:}
This research was supported by unrestricted funds to E.B.
\paragraph*{Author contributions:}
E.B. and Z.L. designed and conducted research (conceptualized, designed estimator, conducted theoretical analysis, wrote code, performed data analysis) and wrote and revised the paper. 

% Conceptualization: 
% Experiment Design: 
% Model Training: 
% Model Hosting: 
% Data Analysis: 
% Visualization: 
% Project Support: 
% Writing - Original Draft: 
% Writing - Review and Editing: 

\paragraph*{Competing interests:}
There are no competing interests to declare.
\paragraph*{Data and materials availability:}

Code to reproduce our analysis is available on Dryad \cite{note:Dryad}. The other three public datasets used in this study are: (i) the Job Corps dataset, available from openICPSR \cite{note:job_data}; (ii) the Massive Open Online Classes (MOOC) education  behavioral interventions dataset, available through the Open Science Framework repository \cite{note:education_data}; (iii) the Joke recommendation dataset, available from the UC Berkeley Jester repository \cite{note:joke_data}. We have documented how to access the source of these datasets, and we have also included on Dryad a snapshot of the data used for the Job corps, MOOC behavioral interventions and Joke recommendation results, to aid in reproducibility.

The Nefazodone-CBASP dataset is not public, and it was obtained through direct request from the authors of the original Nefazodone-CBASP clinical trial~\cite{keller2000comparison}. We have made this accessible to the editor and potential reviewers through an OSF private link to enable the editor and reviewers to reproduce our results. In the future, the Nefazodone-CBASP dataset (only) would need to be requested directly by readers from the dataset owners Keller, Klein, Arnow et al.~\cite{keller2000comparison}: the authors of the current paper contacted the Nefazodone-CBASP dataset owners through daniel.klein@stonybrook.edu and baarnow@stanford.edu.

\subsection*{Supplementary materials}
Materials and Methods\\
Supplementary Text\\
Figs. S1 to S11\\
Table S1 \\
References \textit{(37-\arabic{enumiv})}\\ % automatically fills out the last reference number
% (filling out the other numbers automatically is possible but fiddly and liable to break)
%Movie S1\\
%Data S1

%%%%%%%%%%%%%%%% END OF MAIN TEXT %%%%%%%%%%%%%%%

\newpage
% Figures, tables, equations and pages in the supplement are numbered S1, S2 etc.
\renewcommand{\thefigure}{S\arabic{figure}}
\renewcommand{\thetable}{S\arabic{table}}
\renewcommand{\theequation}{S\arabic{equation}}
\renewcommand{\thepage}{S\arabic{page}}

\setcounter{figure}{0}
\setcounter{table}{0}
\setcounter{equation}{0}
\setcounter{page}{1} % not 0 as \newpage already started a supplementary page
% References continue the numbering from the main text.

%%%%%%%%%%%%%%%% SUPPLEMENT TITLE PAGE %%%%%%%%%%%%%%%

\begin{center}
\section*{Supplementary Materials for\\ \scititle}

% Author list for the supplement
% Indicate the corresponding authors, but do NOT include institutions here
% It would be nice if the template auto-generated this, but doing so is complicated...
Zhaoqi Li$^{\ast\dagger}$,
Emma Brunskill$^{\ast\dagger}$\\
% we're not in a \author{} environment this time, so use \\ for a new line
\small$^\ast$Corresponding authors. Email: zli9@stanford.edu, ebrun@cs.stanford.edu\\
\small$^\dagger$These authors contributed equally to this work.
\end{center}

% Fill out the numbers for each type of supplementary material,
% and delete any lines that aren't applicable.
% These are just example numbers that don't match the rest of this template.
%\section{Full theoretical details}
\section{Notations and objective}\label{sec:notations}

We aim to construct a statistical test to determine whether there is a benefit in using personalized policy from the overall best policy. Let $X\in\X$ be a feature vector with distribution $p$, $A\in\A$ some action, and $Y\in\Y$ an outcome that is observed after the action. We focus on the offline policy learning setting where the sample consists of $n$ independent and identically distributed draws $(x_i,a_i,y_i)_{i=1}^n$ from some distribution $\nu$, where $\E[y_i|x_i,a_i]=r(x_i,a_i)$ with some reward function $r$. 
Let $\Pi$ be a set of policies $\X\to\A$ that map a feature vector to some action.  

For some action $a\in\A$, define its expected utility  as $U(a)=\sum_x p(x) r(x, a)$ and for some policy $\pi\in\Pi$, define its value as $V(\pi)=\sum_{x} p(x)r(x,\pi(x))$. Let the true overall best arm $a^*=\arg\max_a U(a)$ and the true optimal policy $\pi^*=\arg\max_{\pi\in\Pi} V(\pi)$. We define the \textit{personalization effect} $\psi$ as the difference between the optimal policy and the overall best performing arm, i.e. 
\begin{equation}\label{eqn:objective}
    \psi:=\sum_x p(x) r(x, \pi^*(x))- \sum_x p(x) r(x, a^*).
\end{equation}
We are interested in testing whether the personalized best policy gives a larger reward than the overall best arm, that is, constructing a statistical test for $H_0: \psi=0$ against $H_a: \psi>0$. We would like to develop an efficient test that is the most powerful among all tests. 

At a high level, this is similar in spirit to the  standard procedure of testing if an experimental treatment yields the same average outcome as a control condition -- a task that arises ubiquitously across science and industry, including as A/B testing. Here the two (experimental $a_e$ and control $a_c$) interventions are fixed apriori, and, if in a standard randomized control experimental setting, a $t$-test can be used to test if there is evidence given the available data to reject the null hypothesis that the mean outcome of the two interventions is the same: 
\begin{equation}
\sum_x p(x) r(x,a_e) -  \sum_x p(x) r(x,a_c) = 0. 
\end{equation}
There exists a deep literature on estimators and hypothesis testing for such settings, including approaches that can be used with observational data. However, there are a number of additional challenges in estimating the personalization effect, particularly because the "conditions" themselves (the personalized policy and single best intervention) are assumed to also be unknown. 

Before we introduce our approach for estimating the personalization effect, we first consider a few alternative approaches.  

One idea could be to enumerate all possible personalized policies and intervention pairs ($\pi',a'$), and use extensions of policy evaluation techniques (such as doubly robust estimation~\cite{dudik2014doubly}) to compute estimates and confidence intervals over the relative benefit of each $\pi'$ over $a'$, rejecting the null hypothesis if one policy had a personalization effect with a confidence interval greater than zero for all single interventions. A challenge to doing this is that the number of personalized policies is generally intractably large (or infinite), and the confidence intervals would need to hold simultaneously over all pairs to guard against the multiple testing problem, likely leading to vacuously large confidence intervals that would never reject the null hypothesis. 

An alternative idea could be to first learn a personalized policy $\pi''$, and the best single intervention $a''$, using the entire dataset $(X,A,Y)$ and then evaluate the personalization effect and the confidence interval using policy evaluation techniques for $(\pi'',a'')$. Unfortunately, such a method will generally suffer from maximization bias (also known as the optimizer's curse) that arises by first using data to select a decision expected to achieve a maximum, and then using the same data to estimate the value of that decision\cite{van2004rational,smith2006optimizer,hasselt2010double,andrews2024inference}. The maximization bias of using the same data to learn and estimate will apply both to the estimate of the learned personalized policy performance, and the estimate of the learned best-overall action performance. It is unclear precisely how the difference of two such biased estimators will relate to the true personalization effect, but we expect in general that the resulting personalization difference will also be biased upwards-- intuitively there is less data available to estimate the decision policy per context, and so if a policy optimizer seeks to maximize performance on the available data, in general we expect evaluating on the same data will be more optimistically biased than only estimating a single best arm/intervention. Although there do exist bounds on the performance of the output of some algorithms' learned policy relative to the optimal (unknown) policy\cite{swaminathan2015batch,zhou2023offline}, such bounds are typically much too large to be used for estimating if the personalization effect is non-zero.

There also exist personalized policy learning methods (e.g.~\cite{swaminathan2015batch}) which learn a policy that, rather than maximize expected performance, instead maximize a lower bound of the resulting policy performance that accounts for finite sample error (and potentially policy model class complexity). One might wonder if using such techniques would lead to selecting a single arm (or decision policy) whose value, when evaluated on the same dataset on which the arm/policy was learned, would be unbiased. Unfortunately we can show that this is also not true. For simplicity, consider the case where we select the best single action ($a^*$) based on a lower bound on its potential value, using a Hoeffding inequality-style lower bound where $n(a)$ is the number of samples of action $a$ and $\hat{\mu}(a)$ is the empirical average of the outcome $y$ of arm $a$ in the available dataset. 
\begin{eqnarray}
\label{eqn:upperlcb}
\hat{a}^* = \arg\max_a \hat{\mu}(a) - \sqrt{\frac{c}{n(a)}}
\end{eqnarray}
Then Lemma~\ref{lem:lcb_arm_biased} proves that the estimate of its performance will also biased upwards. We expect a similar result can be shown when the contextual bandit policy is selected based on a Hoeffding-style lower bound. As we mentioned above, it is unclear how the difference of two biased estimators will relate to the true personalization effect. However, we discuss this instance to highlight that using lower confidence bounds to select a best arm or contextual policy, will not, in general, mean that we can use the same dataset (as used for learning) for estimating the performance of said policies.

\begin{lemma}\label{lem:lcb_arm_biased}
The estimate of the performance of an arm chosen to maximize its empirical lower confidence bound will be biased upwards: 
\begin{equation}
E[\hat{\mu}(\hat{a}^*)] \geq E[\mu(\hat{a}^*)].
\end{equation}
\end{lemma}
\begin{proof}
First define the event $A_i$ to be the event that arm $a_i$ has the maximum lower confidence bound (as in Equation~\ref{eqn:upperlcb}). Assume that $\hat{\mu}_j$ is fixed for all $j \neq i$. 

Then $A_i$ is increasing in $\hat{\mu}_i$: if $A_i$ is true at a particular $\hat{\mu}_i$, it will continue to be true if $\hat{\mu}_i$ increases, and if $A_i$ is not true, then increasing $\hat{\mu}_i$ may potentially make arm $a_i$ have the highest lower confidence bound and therefore make $A_i$ true. 

Then applying Lemma~\ref{lem:expectX} with $X=\hat{\mu}(a_i)$ and $A=A_i$ we get 
$E[\hat{\mu}(a_i)|\hat{a}^*=a_i] \geq E[\hat{\mu}(a_i)] = \mu(a_i).$
\end{proof}

\begin{lemma}
\label{lem:expectX}
Let $A$ be an event that is increasing in $X$. Then $E[X|A] \geq E[A].$
\end{lemma}
\begin{proof}
Since $A$ is increasing in X, the indicator function $I_A$ is also an increasing function of X. 
Therefore
\begin{eqnarray}
Cov(X,I_A) &\geq& 0 \\
E[X I_A] &\geq& E[X] E[I_A] \\
& = & E[X] P(A), \label{eqn:ex1}
\end{eqnarray}
where the second line follows from the definition of covariance, and the third equality follows since $E[I_A] = P(A).$

We now consider an alternate expression for $E[X I_A]$:
\begin{eqnarray}
E[X I_A] &= & \int_{w} I(w \in A) X(w) dP(w)\\
&=& \int_{w \in A} X(w) dP(w)\\
&=& \int_{w \in A} X(w) dP(w) \frac{P(A)}{P(A)}\\
&=& P(A) \int_{w \in A} X(w) \frac{dP(w)}{P(A)}\\
&=& P(A) \int_{w \in A} X(w) \frac{dP(w \cap A)}{P(A)}\\
&=& P(A) \int_{w \in A} X(w) dP(w | A)\\
&=& P(A) E[X|A] \label{eqn:ex2}
\end{eqnarray}
where line (5) holds because $dP(w) = dP(w \cap A)$ on $w \in A$. 

We now combine Equations~\ref{eqn:ex1} and~\ref{eqn:ex2} to yield 
\begin{eqnarray}
P(A) E[X|A]  & \geq & E[X] P(A)   \\
E[X|A]  & \geq & E[X]
\end{eqnarray}
\end{proof}

We also highlight that heterogeneous treatment effects are necessary (though not sufficient) for a personalization effect to exist. Therefore, in analyzing if a personalization effect exists, one could execute a HTE test (see recent methods such as ~\cite{imai2025statistical} and~\cite{yadlowsky2025evaluating}) and if there is no presence of HTE, then no evidence of personalization is present. However, there do exist settings where HTE is present but no personalization effect is present (see Figure 1). We discuss these issues further in Supplement Section~\ref{sec:hte_and_personalization}. 

\subsection{A K-Fold Personalization Estimator}\label{sec:kfold_full}
We now present details of our K-fold Personalization Estimator. Pseudocode is shown in Algorithms~\ref{alg:kpt} and~\ref{alg:single_kpt}. A  core insight is that we can treat the specific personalized policy and best single intervention as  nuisance parameters for the main estimand of interest, which is the personalization effect. This is analogous to how the outcome model and propensity  model are treated as nuisance parameters for treatment effect estimation. Due to its strong properties in treatment effect estimation, we use an Augmented Inverse Propensity-Weighted (AIPW) style estimator for the personalization effect, defined as 
\begin{eqnarray}
\hat{\psi}^{k} &=&\frac{1}{n_k}\sum_{i=1}^{n_k} \hat{r}^{k}(x_i,\hat{\pi}^k(x_i)) - \hat{r}^{k}(x_i,\hat{a}^k) + \frac{\mathbf{1}\{\hat{\pi}^k(x_i)=a_i\}}{\hat{p}^k(a_i|x_i)}(y_i-\hat{r}^{k}(x_i,\hat{\pi}^k(x_i))) \nonumber\\
& & - \frac{\mathbf{1}\{\hat{a}^{k}=a_i\}}{\hat{p}^k(a_i|x_i)}(y_i-\hat{r}^{k}(x_i,\hat{a}^{k})) \label{eqn:aipw_},
\end{eqnarray}
using a personalized policy $\hat{\pi}^k$, single best intervention $\hat{a}^k$ learned on part of the data, an outcome model $\hat{r}^k$ and propensity model $\hat{p}^k$ learned on another part, and $(x,a,y)$ tuples from a final subpart of the data, fold $k$  (Alg.~\ref{alg:single_kpt}). 

This is repeated through cross-fitting~\cite{chernozhukov2018double}, permuting the data used to do policy and best intervention modeling, outcome and propensity modeling and the remaining held out data used in the estimator, ensuring that each data tuple $(x,a,y)$ is used once in the final role. Cross-fitting has become a popular tool used in deriving efficient estimators, likely in part because it often allows less restrictive assumptions on the model classes used for representing nuisance parameters. 

The particular estimate may depend on the specific data included in each fold, and so we repeatedly partition the data randomly and compute a personalization effect for each random partition. The final estimate and variance is computed across all random partitions and folds, as shown in Algorithm~\ref{alg:kpt}.

\begin{algorithm}[!htb]
\caption{K-Fold Personalization Estimate}
\begin{algorithmic}[1]\label{alg:kpt}
\REQUIRE number of folds $K$, Dataset $D$, number of repeated splits $S$
\FOR{$s = 1,...,S$}
    \STATE Re-ordering the data randomly to create $D_s$
    \STATE $\hat{\psi}_s,\;\hat{\sigma}_s^2 \leftarrow$ SingleSplitKPT($D_s$)
\ENDFOR
\STATE Compute $\tilde{\psi}=\frac{1}{S}\sum_{s=1}^S\hat{\psi}_s$, $\; \tilde{\sigma}^{2}=\frac{1}{S}\sum_{s=1}^S\left(\hat{\sigma}_s^2+(\hat{\psi}_s-\tilde{\psi})^2\right)$
\STATE Compute $z$-statistic $z_{\tilde{\psi}}=\frac{\tilde{\psi}}{\tilde{\sigma}/\sqrt{n}}$ and $p$-value $p$
\ENSURE $\tilde{\psi},\tilde{\sigma},z_{\tilde{\psi}},p$
\end{algorithmic}
\end{algorithm}
\begin{algorithm}[!htb]
\caption{SingleSplitKPT}
\begin{algorithmic}[1]\label{alg:single_kpt}
\REQUIRE number of folds $K$, Dataset $D$, and optional proportion of data to learn policy/best intervention $p_{\pi}$
\STATE Partition $D$ into $D_1,D_2,\cdots,D_K$ folds.
\FOR{$k = 1,2,\cdots,K$}
    \STATE $\hat{\pi}^k \leftarrow$ PersonalizedPolicyLearning$(D_{mod(k+1,K)},\ldots, D_{mod(k+p_{\pi}K,K})$ 
    \STATE $\hat{a}^k \leftarrow $ BestInterventionLearning$(D_{mod(k+1,K)},\ldots, D_{mod(k+p_{\pi}K,K})$    
%    $\hat{\pi}^k$ and single best intervention $\hat{a}^k$ using $D_{mod(k+1,K)},\ldots, D_{mod(k+p_{\pi}K,K}$
    \STATE $\hat{r}^{k} \leftarrow $OutcomeModelLearning$(D_{mod(k+p_{\pi}K,K},\ldots,D_{mod(k+K-1,K)})$
    \STATE $\hat{p}^{k} \leftarrow$ PropensityModelLearning$(D_{mod(k+p_{\pi}K,K},\ldots,D_{mod(k+K-1,K)})$
%    \STATE Compute $\hat{r}^{k}$ and $\hat{p}^k$ from $D_{mod(k+p_{\pi}K,K},\ldots,D_{mod(k+K-1,K)}$
    \STATE Compute $\hat{\psi}^{k}(\hat{\pi}^k,\hat{a}^k,\hat{r}^{k},\hat{p}^{k})$ using equation~\eqref{eqn:aipw_} with  $(x_i,a_i,y_i) \in D_{k}$
\ENDFOR
\STATE Compute $\hat{\psi}=\frac{1}{K}\sum_{k=1}^K \hat{\psi}^{k}$, variance $\hat{\sigma}^2:=\Var(\hat{\psi})$
\ENSURE $\hat{\psi}$, $\hat{\sigma}^2$
\end{algorithmic}
\end{algorithm}

Note that unlike standard treatment effect estimation, in our setting, $\pi^*(x)$ may be equal to $a^*$. 

We require at least $K=3$ but many choices are possible. Note that the proportion of data used to learn the best personalized policy/intervention does not need to be the same proportion as that used to fit outcome and propensity models, and the data held out for use in the final AIPW-style personalization estimation. In general policy learning will benefit from more data.  One possibility is to use a large proportion of the data for policy learning -- such as 80\%, 19\% for outcome modeling and 1\% as the held out set. In general as $K$ gets close to the number of data points $N$, this method will look more like a leave-one-out style cross validation method historically popular in predictive machine learning. The challenge to this is that the computational cost is linear in the number of folds $K$ and so this can start to become prohibitively expensive. In our main experiments we use $K=6$ and use half the folds for policy and best intervention learning, two for outcome modeling and one for the held out evaluation. 

We note that a wide variety of machine learning and statistical models can be used to perform the personalized policy, best intervention, outcome model and propensity model learning. The final two modeling tasks are prediction problems, and there exist an enormous number of predictive modeling tools (linear regression, supervised machine learning, etc) for this task. Personalized policy and best intervention learning are optimization tasks: the goal is to learn from the relevant historical data, a personalized decision policy or single best intervention that is expected to maximize the expected utility if used on individuals with a similar distribution as that sampled in the available dataset. There is also a rich literature on this task of policy learning and best intervention (often called "best arm" in the multi-armed bandit community) learning, in econometrics, machine learning and other fields (e.g.\cite{kitagawa2018should,dudik2014doubly}). Our algorithm is agnostic to the particular approaches used to perform this learning, though we do require some assumptions on the setting and methods used for our theoretical results, as we will shortly describe. 

An alternate approach to our repeated partitions as done in Algorithm \ref{alg:kpt} could be to do repeated bootstraps, where a set of data is sampled with replacement from the original dataset, and then we run Algorithm \ref{alg:single_kpt} on that sampled with replacement dataset, and repeat the whole procedure many times with different bootstraps. While bootstrapping is a popular method for computing sample statistics with minimal parametric assumptions, our focus here was on deriving an efficient estimator of the personalization effect with strong type I error guarantees. Our theoretical analyses leverage  the fact that each observation appears exactly once in the test portion of each split in our approach, which allows us to control the asymptotic distribution of the estimator. Under bootstrap resampling, observations may appear multiple times or not at all in a given sample, which would change the dependence structure of the estimator and require a separate theoretical treatment.  Our Algorithm \ref{alg:kpt} repeated re-ordering and random partitioning over the full dataset is analogous to the approach proposed and used in \cite{chernozhukov2018double}. Studying bootstrap-based variants of estimating personalization effects is an interesting direction for future work.

\section{Related Work}
Though understanding if customizing interventions to different covariates leads to improved outcomes is relevant across a wide range of disciplines, the literature tackling this is relatively sparse and largely disjoint across fields.  In medicine such an effect is termed a qualitative interaction\cite{peto1982statistical,li2006detecting,silvapulle2001tests}. Prior work includes a likelihood ratio test that assumes there is a small  prespecified set of patient subgroups~\cite{gail1985testing},  patient covariate selection methods to maximize qualitative interactions\cite{gunter2011variable}, and decision tree-like methods to identify interpretable qualitative treatment-subgroup interactions\cite{dusseldorp2014qualitative}. To our knowledge this prior work is restricted primarily to binary interventions (treatment or control) and and does not provide tests which are provably statistically efficient, nor focus on settings with a number of covariates,  where machine learning methods are commonly employed. 
% Most closely related to our approach is more recent work~\cite{shi2020sparse} which developed a related method for testing for personalization with binary interventions with high-dimensional input features. 
% It performs personalized policy learning by estimating a conditional average treatment effect (CATE), and then thresholding to provide the intervention if the estimated CATE is greater than zero, else the control intervention is provided. 
% Like our method, it uses a cross-fitting approach~\cite{chernozhukov2018double}.  
% This method has type-I guarantees but does not have efficiency guarantees. We compare to this method in our experiments. 

A related approach to estimating the benefit of personalization is the SRP method~\cite{shi2020sparse}, which aims to detect overall qualitative treatment effects (OQTE) under binary interventions with potentially high dimensional input covariates. They show that testing for OQTE is equivalent to testing the benefit of personalization under binary interventions, so our objective is the same as theirs in the binary intervention case. Nevertheless, there are several important methodological differences. First, SRP uses a Bonferroni correction in constructing the test, whereas our test obtains the test statistic by averaging over $K$ folds. Bonferroni provides valid Type-I error control without assumptions, but it can be conservative, especially when extending their test to multiple actions, which results in a much larger multiple testing problem. In contrast, taking the average yields a higher power. The tradeoff is that controlling Type I error for the averaged statistic requires a mild uniqueness assumption that the optimal action can be learned quickly enough. This assumption is not needed for the Bonferroni correction but is common, especially in bandit literature \cite{lattimore2020bandit}, and can be satisfied when, for example, there is a minimum gap between the best and the second best intervention. Second, SRP is tailored for high-dimensional covariate spaces and explicitly incorporates random projections to mitigate the curse of dimensionality. This makes SRP particularly suitable when $d\gg n$ or when sparsity plays a central role in identifying the personalized effect. Our proposed test targets a more generic semiparametric setting without requiring random projections, sparsity, or high-dimensional regularization. As a result, KPT applies naturally in moderate-dimensional, non-sparse, or structured-covariate regimes whereas SRP is better suited when covariates are high-dimensional. It remains an interesting direction whether one can incorporate the sparse projection approach to KPT and get a better performance in high-dimensional settings. Third, SRP is designed specifically for binary treatment spaces. In such settings policy learning can be used to extracting a threshold policy class given an estimate of the conditional average treatment effect (CATE). It is not possible to directly extend this to when there are more than 2 interventions without further assumptions (as done in the SRP authors' own work, by assuming that the best overall action was known and comparing separately to two different other actions in pairwise fashion~\cite{shi2020sparse}), or without incurring potentially a combinatorial number of pairwise comparisons, which will generally inflate the Bonferroni connection required and reduce power. Fourth and finally, our test aggregates over repeated random permutations of the splits, whereas SRP only performs a single split. This reduces variability of a single split in finite samples. In practice, we find that aggregation substantially stabilizes the test statistic and the $p$-value, while leaving the asymptotic analysis unchanged. We have showed in various real data scenarios that our test produces a more stable estimate than SRP in Figure 3.

Another paper included an estimator for the difference in value of two different learned personalized policies, given data from a randomized controlled trial~\cite{imai2023experimental}. This estimator would include our personalization effect if one policy is restricted to learn a single best intervention for all individuals. The proposed estimator leverages cross-folds, and in each cross-fold estimates the population average prescriptive effect difference (PAPD) between the two policies as
\begin{equation}
\Delta(\pi_1,\pi_2,\mathcal{D}_k) = \frac{1}{n_1} \sum_{i=1}^{N_k} Y_i T_i (I(\pi_1(x_i) - \pi_2(x_i))) - 
\frac{1}{n_0} \sum_{i=1}^{N_k} Y_i (1-T_i) (I(\pi_1(x_i) - \pi_2(x_i))), 
\end{equation}
where $\mathcal{D}_k$ is the $k$-th fold of the data, the policies $\pi_1$ and $\pi_2$ are learned on the data in the other folds, the $I(\cdot)$ is the indicator function, and $x_i,T_i,Y_i$ are the covariates, binary decision (treatment or not), and outcome for the $i$-th datapoint in the $\mathcal{D}_k$ fold. 
The authors evaluate this and their other proposed estimators under Neyman’s repeated sampling framework~\cite{Neyman1923}, and provide a bound on the bias of the estimator and its exact finite-sample variance. In their empirical simulations, they center the outcome variable $Y$ to reduce variance. While there are number of similarities to our proposed approach, including cross fitting, there are several important distinctions: (1) our proposed estimator is based on AIPW, which is semi-parametrically efficient for average treatment effect estimation, whereas PAPD can be viewed as an IPW estimator, which is not; (2) our analysis focused on providing asymptotic normality and semiparametric efficiency, rather than finite sample bounds on bias and variance; and (3) our estimator explicitly accounts and addresses the variability due to data splitting, which PAPD does not. While their original results are presented for binary treatments, we can easily extend their PAPD estimator to finite intervention spaces. We do so and compare to the PAPD estimator with centered outcome variable $Y$, and find that while their confidence intervals are often of similar width to our KPT, their estimates are sensitive to the data split (see Figure 3). We also find that using a centered outcome variable $Y$ was important. In early experiments with PAPD without centering, the PAPD confidence interval width was sometimes double our KPT confidence interval width. In general, we would expect that if the learned outcome model is reasonable, an AIPW-style estimator should match or improve over an IPW-style estimator.

In a separate line of work, there exist more general techniques developed in economics-- such as techniques for inference based on conditional moment inequalities\cite{andrews2013inference} or on intersection bounds\cite{chernozhukov2013intersection}-- which could be specialized to enabling tests of conditional stochastic dominance of one intervention over another, and applied to our setting. However, the aim of such work is quite different, and they do not consider developing efficient statistical tests for the utility of personalized decision policies. Several econometrics papers\cite{delgado2013conditional,lee2013testing,chang2015nonparametric,hsu2017consistent} have provided hypothesis testing on quantities related to conditional average treatment effects, including work that uses structural properties (e.g. concavity) of the joint difference in treatment outcomes\cite{delgado2013conditional}, or reduces the problem to unconditional moment inequalities\cite{hsu2017consistent}. Two of these papers\cite{delgado2013conditional,hsu2017consistent} tested for personalization effects using the National Supported Work Demonstration program dataset\cite{lalonde1986evaluating} (with only~\cite{delgado2013conditional} rejecting the null that the intervention is beneficial at all ages) and one examined if there were differential impacts of 
single-sex schooling on academic achievements\cite{chang2015nonparametric}. 
These papers assume binary interventions, focus on when a single feature is available for personalization, and typically use important sampling based estimators. 
While they provide consistency and asymptotic type 1 error guarantees under some  assumptions\cite{delgado2013conditional,chang2015nonparametric,hsu2017consistent,lee2013testing}, they do not provide efficiency results. 

More broadly, testing if there is a utility benefit to personalization is related to conditional average treatment effect estimation (CATE)\cite{athey2016recursive,wager2018estimation,kunzel2019metalearners,kennedy2023towards}. The presence of conditional average treatment effects that vary (aka heterogeneous treatment effects) is necessary but not sufficient for a personalization effect to exist (as illustrated in Figure 1).  In addition, research on CATE and heterogeneous treatment effect estimation focus on binary interventions. There is also an extensive literature on selecting a best single intervention from a set of interventions given data from a randomized control trial or observational data~\cite{manski2004statistical}, performing inference on the selected intervention~\cite{andrews2024inference}, as well as related work on personalized policy learning~\cite{swaminathan2015batch,athey2021policy,zhou2023offline}. Some of this work is complimentary: for example, our K-Fold Personalization estimate can utilize many different algorithms for policy learning or extracting a single intervention to maximize expected population outcomes, and our theory will directly apply to such cases, assuming the used procedures still satisfy our required assumptions.

In addition, off-policy evaluation for contextual bandits and Markov decision processes is a rich area of interest. There are a number of OPE estimators constructed for the same setting we consider, of evaluating the performance of a contextual bandit policy where data are drawn from a static, unconfounded, behavior policy (e.g.~\cite{dudik2014doubly,wang2017optimal,su2019cab}). However  such estimators often are focused on different objectives (like finite sample mean squared error, or bounds on the minimax error), are not typically proved to be statistically efficient, and, while often offering good empirical performance, sometimes rely on an additional hyperparameter which is non-trivial to select (e.g.~\cite{su2019cab}).

Even pragmatically, if we were to use two different OPE estimates, we ultimately aim to produce a single confidence interval over the personalization effect. Some of the alternate OPE contextual bandit off-policy estimators~\cite{wang2017optimal,su2019cab} have low mean squared error but can be biased, and rely on an estimate of the variance and an upper bound on the bias term, without a central limit theorem on the resulting estimator distribution. In such cases, one might leverage Empirical Bernstein finite sample bounds to aim to compute a confidence interval in the difference between the personalized policy performance and best arm, but then one would also need to widen the confidence interval to account for potential bias. We expect such an approach to yield wide confidence intervals and to have limited power to detect personalization effects.

In contrast, the augmented inverse propensity weighted (AIPW) estimator is known to be statistically efficient for both evaluating the performance of a single intervention (as is commonly studied in causal inference), and for evaluating the performance of a particular contextual bandit policy, given both the propensity and outcome models are consistent, and sufficient conditions on the rate of convergence of the propensity and outcome model nuisance parameter estimates (equivalent to our Assumptions 5 and 7) (see review article~\cite{uehara2025review}). 

Therefore for our TrainEval baseline, we estimate the policy and best overall intervention on one part of the data, and then, on the other part of the data, we estimate the personalization effect using an AIPW estimator and confidence interval. %The Shi et al. SRP estimator~\cite{shi2020sparse} is designed specifically to estimate a personalization effect, and is also expected to perform strongly,  as we observe in our empirical results.

%In this work we do not consider policy estimation from adaptive online datasets, such as data gathered from an adaptive contextual bandit algorithms. The considerations for such a setting are quite different, since the decisions made for later contexts are not independent of the decisions and outcomes observed from prior observed contexts. Performing personalization effect estimation from  online adaptive data is an interesting direction for future work. 

Finally, we note that research in the AI community on personalized policy learning, within the framework of contextual multi-armed bandit learning, typically focuses on online settings, where the measure of algorithm performance is often cumulative regret (which is the difference between the optimal expected outcome of the best intervention versus the intervention executed by the algorithm) or sample complexity bounds. In contrast, in this paper we are focused in the offline setting, where data are already available from a static behavior policy, such as a randomized controlled trial, or unconfounded observational historical data. The offline policy learning and evaluation setting has also received substantial interest in the econometrics, biostatistics and AI research community, in part because there are sometimes safety or logistic concerns that may restrict performing adaptive exploration online. For example, if the outcome of an intervention is delayed (such as a job training program, and its resulting impact on employment 1 year later), it may be more practical to randomize a batch of individuals to different conditions, rather than perform a sequential experiment in which the outcome from the current individual and intervention must be observed before making a decision of how to intervene on the next individual. As noted by several papers (e.g.~\cite{hadad2021confidence,zhang2020inference,bibaut2021post}) leveraging data from online settings offers additional technical complexities, because later interventions selected for individuals depend on the outcome of the interventions for earlier individuals. Recent work~\cite{jia2025crammingcontextualbanditsonpolicy} proposed a new method that provides, under certain assumptions, a strong estimate of the value of the final policy produced at the end of an online contextual bandit algorithm. Their method is designed specifically for evaluating the performance of the final policy computed by the online contextual bandit algorithm-- if the data were gathered using a randomized controlled trial, as we consider in our experiments-- such an algorithm would satisfy their stability requirements but the resulting policy they estimate would simply be the randomized policy. We leave leveraging online data from adaptive data for estimating personalization effects as an interesting area for future work.

\section{The K-Fold Personalization Test Has Valid Type-I Error}\label{sec:proofs}
We first prove our K-Fold Personalization Test has valid Type-I error under fairly weak assumptions.  

We now lay out our assumptions. 
\begin{assumption}[SUTVA]\label{assump:SUTVA}
$Y_i=Y_i(A_i)$ for all $i$. 
\end{assumption}
\begin{assumption}[Unconfoundedness]\label{assump:unconfoundedness}
    $Y_i \perp A_i | X_i$ for all $i$. 
\end{assumption}
\begin{assumption}[Strong overlap]\label{assump:strong_overlap}
$\exists 0<\eta<1$ such that for each $a$ and each $x$, $p(a|x)\in [\eta, 1-\eta]$.
\end{assumption}
\begin{assumption}[Bounded reward]\label{assump:bounded_r}
$\exists C_1<\infty$ such that $\max_{x,a}r(x,a)\leq C_1$.
\end{assumption}
\begin{assumption}[Consistency of estimators]\label{assump:consistent}
For each $k$, we can learn consistent estimators
% $\hat{\pi}^k$ of $\pi^*$ in the sense that \[P(\hat{\pi}^k(x)\neq \pi^*(x))=o(1).\] We can learn a consistent estimator   
$\hat{r}^k$ of $r$, in the sense that for any $a$, 
\[\frac{1}{n_k} \sum_{i=1}^{n_k}\left(\hat{r}^{k}\left(x_i,a\right)-r\left(x_i,a\right)\right)^2 \rightarrow_p 0.\] We can also learn $\hat{p}^k$ of the propensity score $p$ such that for each $a$, \[\frac{1}{n_k} \sum_{i=1}^{n_k}\bigsmile{\frac{1}{\hat{p}^k(a|x_i)}-\frac{1}{p(a|x_i)}}^2\rightarrow_p 0.\]
\end{assumption}
\begin{assumption}[Fast best-arm learner]\label{assump:fast_best_arm_learner} For each $k$, $\P(\hat{a}^k\neq a^*)=o(n_k^{-1/2})$. 
\end{assumption}
\begin{assumption}[Convergence of nuisance parameters]\label{assump:conv_nuisance_params} For each $k$, 
\begin{align}
    \max_{a}\frac{1}{|\II_k|}\sum_{i\in\II_k}\bigsmile{\frac{1}{\hat{p}^k(a|x_i)}-\frac{1}{p(a|x_i)}}\bigsmile{\hat{r}^k(x_i,a)-r(x_i,a)}=o_P(n^{-1/2}).
\end{align}
\end{assumption}
\begin{assumption}[Finite variance]\label{assump:finite_var}
$\exists C<\infty$ such that $\Var(\epsilon_i)\leq C$ for each $i$.
\end{assumption}
\begin{assumption}[Imperfect observation] There exists $c_0>0$ such that $\Var(\epsilon)\geq c_0$. \label{assump:imp_obs}
\end{assumption}
\begin{assumption}[Consistency of learned personalized policy]
\label{ass:consistency}
The oracle optimal policy $\pi^*(x)$ is unique for almost every $x$. The learned policies satisfy, for each fold $k$:
\begin{eqnarray}
P\bigl(\hat{\pi}_k(X) \neq \pi^*(X)\bigr) &\to& 0 \quad \text{as } n \to \infty.
%\label{eq:classification_consistency}
\end{eqnarray}
\end{assumption}

Assumptions 1-5 are standard assumptions in causal inference. Assumption \ref{assump:conv_nuisance_params} is a common assumption in treatment-effect style estimators, where a good rate is only required of the product of errors in the model and propensity nuisance parameters, instead of good rates on both. Assumptions~\ref{assump:conv_nuisance_params} and  \ref{assump:finite_var} are mild assumptions on the outcome noise. The key assumption that is different than treatment effect estimation (or doubly-robust estimators of a personalized policy's outcomes) is Assumption~\ref{assump:fast_best_arm_learner}, which has implications on the hardness of the problem setting in terms of learning an optimal arm. However, this assumption is satisfied if we assume a minimum gap (MG) in average outcomes across the interventions/actions: 
\begin{enumerate}
    \item[(MG)] There exists $\delta>0$ such that for all $a\ne a_{all}^*$, $\E_{x}[r(x,a_{all}^*)-r(x,a)]\geq \delta$. 
\end{enumerate}
Such a gap assumption is common in the multi-armed bandit literature and implies that there is a single intervention that has a better outcome on average than other interventions-- in other words, there are not two or more best interventions with identical average outcomes. We view this as quite reasonable in many settings.

Assumption~\ref{ass:consistency} ensures the optimal personalized policy in the policy class is unique and the learned policies will asymptotically converge to this optimal policy. This is a quite mild assumption, though there are settings that violate this: for example, if there are many individuals with covariates where multiple actions have the some outcome (e.g. consider in the binary intervention space, when many individuals have zero treatment effect).

We first consider the type I error if $S$=1, before extending our analysis to repeated splits/shuffles (for the full Algorithm 1). Recall 
we observe $N$ i.i.d.\ copies $w_i = (X_i, A_i, Y_i)$ drawn from a distribution $p(x)p(a|x)p(y|a,x)$.
The data are randomly partitioned into $K$ folds $\mathcal{I}_1, \dots, \mathcal{I}_K$ of equal size $n_k = N/K$.
Let $D_k = \{w_i : i \in \mathcal{I}_k\}$ denote the data in fold $k$, and $D_{-k} = \{w_i : i \notin \mathcal{I}_k\}$ the data outside fold $k$.

For each fold $k$, note that the decision policy $\hat{\pi}_k : \mathcal{X} \to \mathcal{A}$ is trained on (a subset) of the data in the other folds $D_{-k}$.

% define h
For $w=(x,a,y)$ and for some $\pi\in\Pi$, $b\in\A$, $r'$ and $p'$, define
\begin{equation}\label{eqn:h}
    h_{\pi,b,r',p'}(w):=r'(x,\pi(x))+\frac{\1\{\pi(x)=a\}}{p'(a|x)}(y-r'(x,\pi(x)))-r'(x,b)-\frac{\1\{b=a\}}{p'(a|x)}(y-r'(x,b)).
\end{equation}

Note Algorithm 2 computes 
\begin{equation}
\hat{\psi}  = \frac{1}{K} \sum_{k=1}^K \hat{\psi}_k,
\end{equation}
where
\begin{eqnarray}
\hat{\psi}_k &=& \frac{1}{n_k}\sum_{i \in \mathcal{I}_k} h_{\hat{\pi}_k,\hat{a}_k,\hat{r}_k,\hat{p}_k}(w_i)
\end{eqnarray}

The AIPW score evaluated at policy $\pi$ is:
\begin{eqnarray}
h_{\pi,a^*,r,p}(w) &=& \frac{\1(A = \pi(X))}{e(X, A)}\bigl(Y - r(X, A)\bigr) + \mu\bigl(X, \pi(X)\bigr),
\label{eq:score}
\end{eqnarray}
where $p(a|x)$ is the (oracle-known) propensity score and $r(x,a)$ is the (oracle-known) outcome model.
For observation $i \in \mathcal{I}_k$, write $h_i = h(w_i, \hat{\pi}_k)$.

\begin{theorem}
Under Assumptions \ref{assump:SUTVA}--\ref{ass:consistency}, under $H_0$, for any fixed $\alpha>0$, the $z$-statistic $z_n \;=\;
\frac{\sqrt{n}\,\hat\psi}{\hat\sigma}$ from one run of Algorithm~\ref{alg:kpt} satisfies \(\Pr(z_n>z_{1-\alpha/2})\to 0\) as $n\to\infty$. 
\label{thm:valid_type1}
\end{theorem}
\begin{proof}
Recall that $\hat{\psi}=\frac{1}{K}\sum_{k=1}^K \hat{\psi}^{AIPW,k}$, where 
\begin{align}
    \hat{\psi}^{AIPW,k} &= \frac{1}{n_k}\sum_{i=1}^{n_k} \hat{r}^{k}(x_i,\hat{\pi}^k(x_i)) + \frac{\mathbf{1}\{\hat{\pi}^k(x_i)=a_i , \hat{a}^k \neq \hat{\pi}^k(x_i)\}}{\hat{p}^k(a_i|x_i)}(y_i-\hat{r}^{k}(x_i,\hat{\pi}^k(x_i)))\\
    &\qquad- \hat{r}^{k}(x_i,\hat{a}^k)- \frac{\mathbf{1}\{\hat{a}^k=a_i , \hat{a}^k \neq \hat{\pi}^k(x_i)\}}{\hat{p}^k(a_i|x_i)}(y_i-\hat{r}^{k}(x_i,\hat{a}^k))\\
    &= \frac{1}{n_k}\sum_{i=1}^{n_k} \hat{r}^{k}(x_i,\hat{\pi}^k(x_i)) + \frac{\mathbf{1}\{\hat{\pi}^k(x_i)=a_i\}}{\hat{p}^k(a_i|x_i)}(y_i-\hat{r}^{k}(x_i,\hat{\pi}^k(x_i)))\\
    &\qquad- \hat{r}^{k}(x_i,\hat{a}^k)- \frac{\mathbf{1}\{\hat{a}^k=a_i \}}{\hat{p}^k(a_i|x_i)}(y_i-\hat{r}^{k}(x_i,\hat{a}^k)).
\end{align}
By Lemma~\ref{lem:Tnk_conv}, we have for each $k$,  $T_{n,k}:=\frac{\sqrt{n_k} \hat{\psi}^{AIPW,k}}{\hat{\sigma}_k}$ has valid Type-I error in the sense that for any fixed $\alpha$ and any $k$, \(\Pr(T_{n,k}>z_{1-\alpha/2})\to 0\) as $n\to\infty$. Then, note that $n_k=\frac{n}{K}$ for each $k$, we have
\begin{align}
    z_n &= \frac{\sqrt{n}\hat{\psi}}{\hat{\sigma}} = \frac{\sqrt{n}\frac{1}{K}\sum_{k=1}^K \hat{\psi}^{AIPW,k}}{\hat{\sigma}} = \frac{\frac{1}{\sqrt{K}}\sum_{k=1}^K \sqrt{n_k} \hat{\psi}^{AIPW,k}}{\hat{\sigma}}.
\end{align}
By Proposition~\ref{prop:bound_hat_psi}, we can lower bound the variance of the full dataset by the average of the variance of folds, so 
\begin{align}
    z_n &\leq \frac{\frac{1}{\sqrt{K}}\sum_{k=1}^K \sqrt{n_k} \hat{\psi}^{AIPW,k}}{\frac{1}{K}\sum_{k=1}^K\hat{\sigma}_k} \label{eqn:var}\\
    &= \sqrt{K}\sum_{i=1}^K \frac{\hat{\sigma}_i}{\sum_{j=1}^K \hat{\sigma}_j} \frac{\sqrt{n_i}\hat{\psi}^{AIPW,i}}{\hat{\sigma}_i}\\
    &= \sqrt{K}\sum_{i=1}^K \frac{\hat{\sigma}_i}{\sum_{j=1}^K \hat{\sigma}_j} T_{n,i}\\
    &\leq \sqrt{K}\max_{i\in [K]} T_{n,i}\tag{linear combination}
\end{align}
Since for each $i$, $T_{n,i}=o_P(1)$ and $K$ is finite, $\max_{i\in [K]} T_{n,i}=o_P(1)$ and $\sqrt{K}\max_{i\in [K]} T_{n,i}=o_P(1)$. Therefore, we have the statement above. 
\end{proof}

\begin{corollary}
Under Assumptions~\ref{assump:SUTVA}--\ref{ass:consistency}, under $H_0$, for any fixed $S$, the $z$-statistic $z_{\tilde{\psi}}$ in Algorithm~\ref{alg:kpt} satisfies for any $\alpha$, \(\Pr(z_{\tilde{\psi}}>z_{1-\alpha/2})\to 0\) as $n\to\infty$. 
\end{corollary}
\begin{proof}
Recall that in Algorithm~\ref{alg:kpt}, the $z$-statistic is defined as
\begin{align*}
    z_{\tilde{\psi}} &= \frac{\tilde{\psi}}{\tilde{\sigma}/\sqrt{n}} = \frac{\frac{1}{S}\sum_{i=1}^S \hat{\psi}_i}{\frac{1}{\sqrt{n}} \sqrt{\frac{1}{S}\sum_{i=1}^S\left(\hat{\sigma}_i^2+(\hat{\psi}_i-\hat{\psi})^2\right)}}.
\end{align*}
Define the $z$-statistic for one repeat as $z_i:=\sqrt{n}\frac{\hat{\psi}_i}{\hat{\sigma}_i}$. Note that
\[
\hat\psi_i = \frac{\hat\sigma_i}{\sqrt{n}}\,z_i,
\]
so
\[
\tilde\psi
= \frac{1}{S}\sum_{i=1}^S \frac{\hat\sigma_i}{\sqrt{n}}\,z_i
= \frac{1}{S\sqrt{n}}\sum_{i=1}^S \hat\sigma_i\,z_i.
\]
Therefore
\begin{align*}
z_{\tilde{\psi}}
&= \frac{\sqrt{n}\,\tilde\psi}{\tilde\sigma}
\leq \frac{\displaystyle \sqrt{n}\,\frac{1}{S\sqrt{n}}\sum_i \hat\sigma_i\,z_i}
       {\displaystyle \sqrt{\frac{1}{S}\sum_j \hat\sigma_j^2}} \\[6pt]
&= \frac{\displaystyle \frac{1}{S}\sum_{i=1}^S \hat\sigma_i\,z_i}
       {\displaystyle \frac{1}{\sqrt{S}}\sqrt{\sum_{j=1}^S\hat\sigma_j^2}}
= \frac{\displaystyle \sum_{i=1}^S \hat\sigma_i\,z_i}
       {\displaystyle  \sqrt{S \sum_{j=1}^S\hat\sigma_j^2}}
       \leq\frac{\displaystyle \sum_{i=1}^S \hat\sigma_i\,z_i}
       { \displaystyle \sum_{j=1}^S\hat\sigma_j}.
\end{align*}
where the last inequality follows from Cauchy–Schwarz. Therefore, 
\begin{align}
    z_{\tilde{\psi}} &\leq  \sum_{i=1}^S \frac{\hat{\sigma}_i}{\sum_{j=1}^S \hat{\sigma}_j} z_i. 
\end{align}
Since each $z_i \xrightarrow{p}0$,  $z_{\tilde{\psi}}\xrightarrow{p}0$. 

\end{proof}

\begin{lemma}\label{lem:Tnk_conv}
Under Assumptions~\ref{assump:SUTVA}--\ref{ass:consistency}, for each $k$, we have for any $\epsilon>0$, \(\Pr(T_{n,k}>\epsilon)\to 0\) as $n\to\infty$. 
\end{lemma}
\begin{proof}
Let $w=(x,a,y)$ and for some $\pi\in\Pi$, $b\in\A$, $r'$ and $p'$, define \[h_{\pi,b,r',p'}(w)=r'(x,\pi(x))+\frac{\1\{\pi(x)=a\}}{p'(a|x)}(y-r'(x,\pi(x)))-r'(x,b)-\frac{\1\{b=a\}}{p'(a|x)}(y-r'(x,b)).\] Note that 
\begin{align*}
    \frac{\sqrt{n_k} \hat{\psi}^{AIPW,k}}{\hat{\sigma}_k} &= \frac{\frac{1}{\sqrt{n_k}}\sum_{i=1}^{n_k}h_{\hat{\pi}^k,\hat{a}^k,\hat{r}^k,\hat{p}^k}(w_i)}{\sqrt{\frac{1}{n_k}\sum_{i=1}^{n_k}(h_{\hat{\pi}^k,\hat{a}^k,\hat{r}^k,\hat{p}^k}(w_i)-\bar{h}_{\hat{\pi}^k,\hat{a}^k,\hat{r}^k,\hat{p}^k})^2}}.
\end{align*}
where $\bar{h}_{\hat{\pi}^k,\hat{a}^k,\hat{r}^k,\hat{p}^k}:=\frac{1}{n_k}\sum_{i=1}^{n_k}h_{\hat{\pi}^k,\hat{a}^k,\hat{r}^k,\hat{p}^k}(w_i)$ is the sample average. Since $\hat{\pi}^k$, $\hat{a}^k$, $\hat{r}^k$, and $\hat{p}^k$ are fixed for fold $k$ as they are learned using data from other folds, by Law of Large Numbers we have $\hat{\sigma}_k\overset{p}{\to} \sigma_k$, where $\sigma_k^2:=\Var(h_{\hat{\pi}^k,\hat{a}^k,\hat{r}^k,\hat{p}^k}(w))$. If for all $i\in[n_k]$, $\hat{\pi}^k(w_i)=\hat{a}^k$, then by our definition that $0/0=0$, the ratio is zero so it satisfies the statement. Then we focus on the case that there exists some $i\in[n_k]$ such that $\hat{\pi}^k(w_i)\neq \hat{a}^k$. Then, $\hat{\sigma}_k > 0$. We also focus on $p'$ such that $\max_{a,x}\frac{1}{p'(a|x)}\leq c$ for some $c<\infty$. Now, by definition of $y$, 
\begin{align*}
    &h_{\pi,b,r',p'}(w)\\
    &=r'(x,\pi(x))+\frac{\1\{\pi(x)=a\}}{p'(a|x)}(r(x,a)-r'(x,a)+\epsilon)-r'(x,b)-\frac{\1\{b=a\}}{p'(a|x)}(r(x,a)-r'(x,a)+\epsilon)\\
    &= r'(x,\pi(x))-r'(x,b)+\left(\frac{\1\{\pi(x)=a\}}{p'(a|x)}-\frac{\1\{b=a\}}{p'(a|x)}\right)(r(x,a)-r'(x,a))+\left(\frac{\1\{\pi(x)=a\}}{p'(a|x)}-\frac{\1\{b=a\}}{p'(a|x)}\right)\epsilon
\end{align*}
Then 
\begin{align*}
    \E[h_{\pi,b,r',p'}(w)^2]&= \E\left[\left(r'(x,\pi(x))-r'(x,b)+\left(\frac{\1\{\pi(x)=a\}}{p'(a|x)}-\frac{\1\{b=a\}}{p'(a|x)}\right)(r(x,a)-r'(x,a))\right)^2\right]\\
    &\quad+\E\left[\left(\frac{\1\{\pi(x)=a\}}{p'(a|x)}-\frac{\1\{b=a\}}{p'(a|x)}\right)^2\epsilon^2\right]\tag{the cross-term has expectation zero}
\end{align*}
whereas 
\begin{align*}
    \E[h_{\pi,b,r',p'}(w)]^2 &= \E\left[\left(r'(x,\pi(x))-r'(x,b)+\left(\frac{\1\{\pi(x)=a\}}{p'(a|x)}-\frac{\1\{b=a\}}{p'(a|x)}\right)(r(x,a)-r'(x,a))\right)\right]^2.
\end{align*}
Therefore, 
\begin{align*}
    \Var(h_{\pi,b,r',p'}(w)) &= \E\left[\left(r'(x,\pi(x))-r'(x,b)+\left(\frac{\1\{\pi(x)=a\}}{p'(a|x)}-\frac{\1\{b=a\}}{p'(a|x)}\right)(r(x,a)-r'(x,a))\right)^2\right]\\
    &\quad-\E\left[\left(r'(x,\pi(x))-r'(x,b)+\left(\frac{\1\{\pi(x)=a\}}{p'(a|x)}-\frac{\1\{b=a\}}{p'(a|x)}\right)(r(x,a)-r'(x,a))\right)\right]^2\\
    &\quad+\E\left[\left(\frac{\1\{\pi(x)=a\}}{p'(a|x)}-\frac{\1\{b=a\}}{p'(a|x)}\right)^2\epsilon^2\right]\\
    &= \Var\bigsmile{r'(x,\pi(x))-r'(x,b)+\left(\frac{\1\{\pi(x)=a\}}{p'(a|x)}-\frac{\1\{b=a\}}{p'(a|x)}\right)(r(x,a)-r'(x,a))}\\
    &\quad+ \E\left[\left(\frac{\1\{\pi(x)=a\}}{p'(a|x)}-\frac{\1\{b=a\}}{p'(a|x)}\right)^2\epsilon^2\right]\\
    &\geq \E\left[\left(\frac{\1\{\pi(x)=a\}}{p'(a|x)}-\frac{\1\{b=a\}}{p'(a|x)}\right)^2\epsilon^2\right]\tag{variance is nonnegative}\\
    &\geq c\E\left[\left(\1\{\pi(x)=a\}-\1\{b=a\}\right)^2\epsilon^2\right]\tag{by our assumption on $p'$}\\
    &= c\E_x \left[\E_a \left[\left(\1\{\pi(x)=a\}-\1\{b=a\}\right)^2\epsilon^2\,\big| x\right]\right]\\
    &= c\E_x \left[ \sum_a p(a|x) \left(\1\{\pi(x)=a\}-\1\{b=a\}\right)^2\E[\epsilon^2]\,\right]\\
    &= c\E_x \left[ \left(p(\pi(x)|x)+p(b|x)\right)\1\{\pi(x)\neq b\}\E[\epsilon^2]\,\right]\\
    &\geq c\E_x \left[ 2c_0\1\{\pi(x)\neq b\}\right]\tag{by Assumption~\ref{assump:imp_obs}}\\
    &= 2c'\P(\pi(x)\neq b),
\end{align*}
for $c'=c_0$. 

We next focus on the numerator. Note that
\begin{eqnarray}
    \frac{1}{\sqrt{n_k}}\sum_{i=1}^{n_k}h_{\hat{\pi}^k,\hat{a}^k,\hat{r}^k, \hat{p}^k}(w_i)&=& \sqrt{n_k}\frac{1}{n_k}\sum_{i=1}^{n_k}h_{\hat{\pi}^k,\hat{a}^k,\hat{r}^k, \hat{p}^k}(w_i) \\
    &=& \sqrt{n_k} o_p(n^{-1/2}) + \sqrt{n_k}\frac{1}{n_k}\sum_{i=1}^{n_k}h_{\hat{\pi}^k,\hat{a}^k,r,p}(w_i), 
\end{eqnarray}
by Proposition~\ref{prop:avg_h_conv}.

Note that for any $w$, and any $\tilde{r},\tilde{p}$ if $\hat{\pi}^k(w)=\hat{a}^k$, $h_{\hat{\pi}^k,\hat{a}^k,\tilde{r},\tilde{p}}(w)=0$, so 
\begin{eqnarray}
 \frac{1}{n_k} \sum_i [h_{\hat{\pi}^k,\hat{a}^k,\tilde{r},\tilde{p}}(w)]=
 \frac{1}{n_k} \sum_i h_{\hat{\pi}^k,\hat{a}^k,\tilde{r},\tilde{p}}(w)\1\{\hat{\pi}^k(w)\neq \hat{a}^k\}.
 \end{eqnarray}
 Therefore
 \begin{eqnarray}
\sqrt{n_k}\frac{1}{n_k}\sum_{i=1}^{n_k}h_{\hat{\pi}^k,\hat{a}^k,r,p}(w_i) &=& \underbrace{\sqrt{n_k}\frac{1}{n_k}\sum_{i=1}^{n_k}h_{\hat{\pi}^k,\hat{a}^k,r,p}(w_i) \1\{\hat{\pi}^k(w)\neq \hat{a}^k\}]  \1\{\hat{a}^k \neq a^*\}}_{Term \; 331} + \nonumber \\ & & 
\underbrace{\sqrt{n_k}\frac{1}{n_k}\sum_{i=1}^{n_k}h_{\hat{\pi}^k,\hat{a}^k,r,p}(w_i) \1\{\hat{\pi}^k(w)\neq \hat{a}^k\}]  \1\{\hat{a}^k = a^*\}}_{Term \; 332}  \nonumber
\end{eqnarray}
By Lemma~\ref{lem:h_hat_pi_oracle_neq_astart_oh_p_1}, Term 332 is $o_p(1)$. We now focus on Term 331.

\begin{eqnarray}
 \sqrt{n_k}\frac{1}{n_k}\sum_{i=1}^{n_k}h_{\hat{\pi}^k,\hat{a}^k,r,p}(w_i) \1\{\hat{\pi}^k(w)\neq \hat{a}^k\}]  \1\{\hat{a}^k \neq a^*\} \\
= \sqrt{n_k} \1\{\hat{a}^k \neq a^*\}  \frac{1}{n_k}\sum_{i=1}^{n_k}h_{\hat{\pi}^k,\hat{a}^k,r,p}(w_i) \1\{\hat{\pi}^k(w)\neq \hat{a}^k\}]  \\
\leq  \sqrt{n_k} \1\{\hat{a}^k \neq a^*\}  
 \sqrt{
 \frac{1}{n_k}\sum_{i=1}^{n_k}h_{\hat{\pi}^k,\hat{a}^k,r,p}(w_i)^2}\sqrt{\frac{1}{n_k}\sum_{i=1}^{n_k} \1\{\hat{\pi}^k(w)\neq \hat{a}^k\}} \\
 \leq  \sqrt{n_k} \1\{\hat{a}^k \neq a^*\}  
  \sqrt{\frac{1}{n_k}\sum_{i=1}^{n_k} \1\{\hat{\pi}^k(w)\neq \hat{a}^k\}} O_p(1),
 \end{eqnarray}
where the first inequality uses Cauchy-Schwarz, and the second inequality follows by Lemma~\ref{lem:avg_h_and_avg_h2_with_oracle_arp_Op1}, $\frac{1}{n_k}\sum_{i=1}^{n_k}h_{\hat{\pi}^k,\hat{a}^k,r,p}(w_i)^2 = O_p(1)$.  

Note that $\1\{\hat{\pi}^k(w_i)\neq \hat{a}^k\}$ are i.i.d. because $\hat{\pi}^k$ and $\hat{a}$ are fixed (and learned from the other fold). By the central limit theorem $\frac{1}{n_k}\sum_{i=1}^{n_k} \1\{\hat{\pi}^k(w)\neq \hat{a}^k\} =  P(\hat{\pi}^k(W) \neq \hat{a}) + O_p(n_k^{-1/2})$. Therefore

%{\color{blue} 

%\textbf{ZL attempt for fix:} 
\begin{eqnarray}
\sqrt{n_k} \1\{\hat{a}^k \neq a^*\}\sqrt{\frac{1}{n_k}\sum_{i=1}^{n_k} \1\{\hat{\pi}^k(w)\neq \hat{a}^k\}} O_p(1) \\
= \sqrt{n_k} \1\{\hat{a}^k \neq a^*\}\sqrt{P(\hat{\pi}^k(W)\neq \hat{a}^k)+O_P(n_k^{-1/2})} O_p(1)\\
= \sqrt{n_k} \1\{\hat{a}^k \neq a^*\}\sqrt{P(\hat{\pi}^k(W)\neq \hat{a}^k)} O_p(1) + \1\{\hat{a}^k \neq a^*\} O_P(n_k^{1/4}).
\end{eqnarray}
Then, putting it together, conditional on $\mathcal{D}_{-k}$, 
\begin{eqnarray}
    T_{n,k}&\leq& \frac{\sqrt{n_k} \1\{\hat{a}^k \neq a^*\}\sqrt{P(\hat{\pi}^k(W)\neq \hat{a}^k)} O_p(1) + \1\{\hat{a}^k \neq a^*\} O_P(n_k^{1/4})}{\sqrt{\Var(h_{\hat{\pi}^k,\hat{a}^k,\hat{r}^k,\hat{p}^k}(W))}+o_P(1)}\\
    &=& \frac{\sqrt{n_k} \1\{\hat{a}^k \neq a^*\}\sqrt{P(\hat{\pi}^k(W)\neq \hat{a}^k)} O_p(1) }{\sqrt{2c_0 \P(\hat{\pi}^k(W)\neq \hat{a}^k)}+o_P(1)} +\frac{\1\{\hat{a}^k \neq a^*\} O_P(n^{1/4})}{\sqrt{2c_0 \P(\hat{\pi}^k(W)\neq \hat{a}^k) }+o_P(1)} \label{eqn:tnk_2}
\end{eqnarray}
Consider the second term. First, if $\hat{a} = a^*$ this term is 0. 
Now consider if $\hat{a} \neq a^*$. Note that under $H_0$
 and Assumption~\ref{ass:consistency}, $\hat{\pi}^k(W) \to a^*$ for all $W$. Since $\hat{a} \neq a^*$, every $W$ where $\hat{\pi}^k(W) = \hat{a}^k$ also must satisfy $\hat{\pi}^k(W) \neq a^*$. Therefore
\begin{eqnarray}
P(\hat{\pi}^k(W) = \hat{a}^k) \leq P\hat{\pi}^k(W) \neq a^*) \to 0. 
\end{eqnarray}
This implies $P\hat{\pi}^k(W) \neq  \hat{a}^k) \to 1$. So the denominator of the second term, in this second case.  $\sqrt{2c_0 \P(\hat{\pi}^k(W)\neq \hat{a}^k)} \to \sqrt{2c_0} $. Therefore the second term is at most
\begin{eqnarray}
\frac{\1\{\hat{a}^k \neq a^*\} O_P( n^{1/4})}{\sqrt{2c_0 } + o_p(1)} = \1\{\hat{a}^k \neq a^*\} O_P( n^{1/4}).
\end{eqnarray}
Then note that unconditionally
\begin{equation}
P(\1\{\hat{a}^k \neq a^*\} n^{1/4} > \epsilon) \leq P\1\{\hat{a}^k \neq a^*\}) = o_p(n^{-1/2}) \to 0,
\end{equation}
since for any large enough $n_k$, $\1\{\hat{a}^k \neq a^*\} n^{1/4} > \epsilon$ if $\1\{\hat{a}^k \neq a^*\}$, and by Assumption~\ref{assump:fast_best_arm_learner} the latter result holds. 

This shows that $\1\{\hat{a}^k \neq a^*\} n^{1/4} = o_P(1)$ and therefore $\frac{\1\{\hat{a}^k \neq a^*\} O_P(n^{1/4})}{\sqrt{2c_0 \P(\hat{\pi}^k(W)\neq \hat{a}^k)}} = o_p(1).$

% The rest of the proof for the first term goes through as before.

Let $C_1$ be some constant such that $\frac{1}{n_k}\sum_{i=1}^{n_k}h_{\hat{\pi}^k,\hat{a}^k,r,p}(w_i)^2 = O_p(1) \leq C$.
 
Therefore, substituting back into Equation~\ref{eqn:tnk_2}, 
\begin{eqnarray}
    T_{n,k}&\leq& 
    \frac{\sqrt{n_k} \1\{\hat{a}^k \neq a^*\}\sqrt{P(\hat{\pi}^k(W)\neq \hat{a}^k)} O_p(1) }{\sqrt{2c_0 \P(\hat{\pi}^k(W)\neq \hat{a}^k)}+o_P(1)} + o_P(1) \\
    &\leq& \frac{C_1 \sqrt{n_k} \1\{\hat{a}^k \neq a^*\}  \sqrt{P(\hat{\pi}^k(W) \neq \hat{a}^k)}}{\sqrt{2c_0 \P(\hat{\pi}^k(W)\neq \hat{a}^k)}}\\
    &\leq& \frac{C}{2c_0}\sqrt{n_k}\1\{\hat{a}^k\neq a^*\} \label{eqn:conv_per_n}
%    &\overset{p}{\to}& \frac{C}{2c_0}\sqrt{n_k}\1\{\hat{a}^k\neq a^*\}.
\end{eqnarray}
 Let $m_{\mathcal{D}_{-k}}:=\frac{C}{2c_0}\sqrt{n_k}\1\{\hat{a}^k\neq a^*\}$. \\Now we show that unconditionally $T_{n,k}=o_P(1)$. Fix an arbitrary \(\epsilon' > 0\). We will show \(\Pr(|T_{n,k}| > 2\epsilon') \to 0\). First, note that
\[
\Pr\bigl(|T_{n,k}|>2\epsilon'\bigr)
\;=\;
\Pr\bigl(|T_{n,k} - m_{\mathcal{D}_{-k}} + m_{\mathcal{D}_{-k}}|>2\epsilon'\bigr)
\;\le\;
\Pr\bigl(|T_{n,k} - m_{\mathcal{D}_{-k}}|>\epsilon'\bigr)
\;+\;
\Pr\bigl(|m_{\mathcal{D}_{-k}}|>\epsilon'\bigr).
\]
Define
\[
g(\mathcal{D}_{-k})
\;=\;
\Pr\bigl(|T_{n,k} - m_{\mathcal{D}_{-k}}| > \epsilon' \,\big|\;\mathcal{D}_{-k}\bigr).
\]
By the law of total probability,
\[
\Pr\bigl(|T_{n,k} - m_{\mathcal{D}_{-k}}|>\epsilon'\bigr)
=
\mathbb{E}\bigl[g(\mathcal{D}_{-k})\bigr].
\]
Since \(T_{n,k} \mid \mathcal{D}_{-k}\to_p m_{\mathcal{D}_{-k}}\) by equation \eqref{eqn:conv_per_n}, for each fixed \(\mathcal{D}_{-k}\) we have
\[
g(\mathcal{D}_{-k}) \;=\; \Pr\bigl(|T_{n,k} - m_{\mathcal{D}_{-k}}| > \epsilon' \mid \mathcal{D}_{-k}\bigr)
\;\longrightarrow\;0.
\]
Moreover, \(0\le g(\mathcal{D}_{-k})\le1\), so by the Dominated Convergence Theorem,
\[
\mathbb{E}\bigl[g(\mathcal{D}_{-k})\bigr]
\;\longrightarrow\;
0.
\]
As for the second term, by Markov’s inequality,
\[
\Pr\bigl(|m_{\mathcal{D}_{-k}}|>\epsilon'\bigr)
\;\le\;
\frac{\mathbb{E}\bigl[\,|m_{\mathcal{D}_{-k}}|\,\bigr]}{\epsilon'}
\;\longrightarrow\;
0,
\]
since $\E[m_{\mathcal{D}_{-k}}]=\E[\sqrt{n_k}\1\{\hat{a}^k\neq a^*\}]=o(1)$ by Assumption~\ref{assump:fast_best_arm_learner}. Combining all the above gives
\[
\Pr\bigl(|T_{n,k}|>2\epsilon'\bigr)
\;\le\;
\Pr\bigl(|T_{n,k} - m_{\mathcal{D}_{-k}}|>\epsilon'\bigr)
\;+\;
\Pr\bigl(|m_{\mathcal{D}_{-k}}|>\epsilon'\bigr)
\;\longrightarrow\;0.
\]
Hence \(T_{n,k} \xrightarrow{p} 0\) unconditionally.

\end{proof}

\begin{proposition}\label{prop:avg_h_conv}
Under Assumptions~\ref{assump:SUTVA}--\ref{ass:consistency}, for each $k$, 
\begin{equation}
\frac{1}{n_k} \sum_{i\in\II_k} h_{\hat{\pi}^k,\hat{a}^k,\hat{r}^k,\hat{p}^k}(w_i)-\frac{1}{n_k} \sum_{i\in\II_k}  h_{\hat{\pi}^k,\hat{a}^k,r,p}(w_i) = o_P(n_k^{-1/2}).
\end{equation}
\end{proposition}
\begin{proof}
We have 
\begin{align}
    &h_{\hat{\pi}^k,\hat{a}^k,\hat{r}^k,\hat{p}^k}(w) - h_{\hat{\pi}^k,\hat{a}^k,r,p}(w) \\
    &=\hat{r}^k(x,\hat{\pi}^k(x)) + \frac{\mathbf{1}\{\hat{\pi}^k(x)=a\}}{\hat{p}^k(a|x)}\left(y-\hat{r}^k(x,\hat{\pi}^k(x))\right)-r(x,\hat{\pi}^k(x)) - \frac{\mathbf{1}\{\hat{\pi}^k(x)=a\}}{p(a|x)}\left(y-r(x,\hat{\pi}^k(x))\right)\\
    &\qquad-\left(\hat{r}^k(x,\hat{a}^k) + \frac{\mathbf{1}\{\hat{a}^k=a\}}{\hat{p}^k(a|x)}\left(y-\hat{r}^k(x,\hat{a}^k)\right)-r(x,\hat{a}^k) - \frac{\mathbf{1}\{\hat{a}^k=a\}}{p(a|x)}\left(y-r(x,\hat{a}^k)\right)\right).
\end{align}

Then, define 
\begin{equation}\label{eqn:m0_k}
    \hat{m}_{(0)}^k(w)=\hat{r}(x,\hat{a}^k(x)) + \frac{\mathbf{1}\{\hat{a}^k(x)=a\}}{\hat{p}^k(a|x)}(y-\hat{r}(x,\hat{a}^k(x))).
\end{equation}
\begin{equation}\label{eqn:m0_hata}
    \hat{m}_{(0)}^{k,\hat{a}}(w)=r(x,\hat{a}^k) + \frac{\mathbf{1}\{\hat{a}^k=a\}}{p(a|x)}(y-\hat{r}(x,\hat{a}^k)).
\end{equation}

We have 
\begin{align}
    &h_{\hat{\pi}^k,\hat{a}^k,\hat{r}^k,\hat{p}^k}(w) - h_{\hat{\pi}^k,\hat{a}^k,r,p}(w) \\
    &=\hat{r}^k(x,\hat{\pi}^k(x)) + \frac{\mathbf{1}\{\hat{\pi}^k(x)=a\}}{\hat{p}^k(a|x)}\left(y-\hat{r}^k(x,\hat{\pi}^k(x))\right)-r(x,\hat{\pi}^k(x)) - \frac{\mathbf{1}\{\hat{\pi}^k(x)=a\}}{p(a|x)}\left(y-r(x,\hat{\pi}^k(x))\right)\label{eqn:first_term}\\
    &\qquad-(\widehat{m}_{(0)}^k(w)-\widehat{m}_{(0)}^{k, \hat{a}}(w)) \label{eqn:prop_second_term}.
\end{align}

The first term \eqref{eqn:first_term} is equal to
\begin{align}
    &\hat{r}^k(x,\hat{\pi}^k(x)) + \frac{\mathbf{1}\{\hat{\pi}^k(x)=a\}}{\hat{p}^k(a|x)}y-\frac{\mathbf{1}\{\hat{\pi}^k(x)=a\}}{\hat{p}^k(a|x)}\hat{r}^k(x,\hat{\pi}^k(x))\\
    &\qquad-r(x,\hat{\pi}^k(x)) - \frac{\mathbf{1}\{\hat{\pi}^k(x)=a\}}{p(a|x)}y+\frac{\mathbf{1}\{\hat{\pi}^k(x)=a\}}{p(a|x)}r(x,\hat{\pi}^k(x))\\
    &\qquad+ \frac{\mathbf{1}\{\hat{\pi}^k(x)=a\}}{p(a|x)}\hat{r}^k(x,\hat{\pi}^k(x))-\frac{\mathbf{1}\{\hat{\pi}^k(x)=a\}}{p(a|x)}\hat{r}^k(x,\hat{\pi}^k(x))\\
    &\qquad+\bigsmile{1-\frac{\mathbf{1}\{\hat{\pi}^k(x)=a\}}{p(a|x)}}r(x,\hat{\pi}^k(x)) - \bigsmile{1-\frac{\mathbf{1}\{\hat{\pi}^k(x)=a\}}{p(a|x)}}r(x,\hat{\pi}^k(x))\\
    &=\bigsmile{1-\frac{\mathbf{1}\{\hat{\pi}^k(x)=a\}}{p(a|x)}}\left(\hat{r}^k(x,\hat{\pi}^k(x))-r(x,\hat{\pi}^k(x))\right)+\bigsmile{1-\frac{\mathbf{1}\{\hat{\pi}^k(x)=a\}}{p(a|x)}}r(x,\hat{\pi}^k(x))\\
    &\qquad+\frac{\mathbf{1}\{\hat{\pi}^k(x)=a\}}{p(a|x)}\hat{r}^k(x,\hat{\pi}^k(x))-\frac{\mathbf{1}\{\hat{\pi}^k(x)=a\}}{\hat{p}^k(a|x)}\hat{r}^k(x,\hat{\pi}^k(x))+\frac{\mathbf{1}\{\hat{\pi}^k(x)=a\}}{\hat{p}^k(a|x)}y\\
    &\qquad-r(x,\hat\pi^k(x)) - \frac{\mathbf{1}\{\hat\pi^k(x)=a\}}{p(a|x)}(y-r(x,\hat\pi^k(x)))\\
    &= \underbrace{\bigsmile{1-\frac{\mathbf{1}\{\hat{\pi}^k(x)=a\}}{p(a|x)}}(\hat{r}^k(x,\hat{\pi}^k(x))-r(x,\hat{\pi}^k(x))}_{\textbf{S1}} \\
    &\qquad+\bigsmile{\frac{\mathbf{1}\{\hat{\pi}^k(x)=a\}}{p(a|x)}-\frac{\mathbf{1}\{\hat{\pi}^k(x)=a\}}{\hat{p}^k(a|x)}}\bigsmile{\hat{r}(x,\hat{\pi}^k(x))-r(x,\hat{\pi}^k(x))}\\
    &\qquad-\frac{\mathbf{1}\{\hat{\pi}^k(x)=a\}}{\hat{p}^k(a|x)}r(x,\hat{\pi}^k(x)))+r(x,\hat{\pi}^k(x))+\frac{\mathbf{1}\{\hat{\pi}^k(x)=a\}}{\hat{p}^k(a|x)}y\\
    &\qquad-r(x,\hat{\pi}^k(x)) - \frac{\mathbf{1}\{\hat{\pi}^k(x)=a\}}{p(a|x)}(y-r(x,\hat{\pi}^k(x)))
\end{align}
\begin{align}
    &= \textbf{S1}+\underbrace{\bigsmile{\frac{\mathbf{1}\{\hat{\pi}^k(x)=a\}}{p(a|x)}-\frac{\mathbf{1}\{\hat{\pi}^k(x)=a\}}{\hat{p}^k(a|x)}}\bigsmile{\hat{r}(x,\hat{\pi}^k(x))-r(x,\hat{\pi}^k(x))}}_{\textbf{S4}} \\
    &\qquad+\frac{\mathbf{1}\{\hat{\pi}^k(x)=a\}}{p(a|x)}r(x,\hat{\pi}^k(x))+\frac{\mathbf{1}\{\hat{\pi}^k(x)=a\}}{\hat{p}^k(a|x)}y\\
    &\qquad-\frac{\mathbf{1}\{\hat\pi^k(x)=a\}}{p(a|x)}y-\frac{\mathbf{1}\{\hat{\pi}^k(x)=a\}}{\hat{p}^k(a|x)}r(x,\hat{\pi}^k(x))\\
    &=\textbf{S1}+\textbf{S4}+\underbrace{\bigsmile{\frac{\mathbf{1}\{\hat{\pi}^k(x)=a\}}{\hat{p}^k(a|x)}-\frac{\mathbf{1}\{\hat{\pi}^k(x)=a\}}{p(a|x)}}(y-r(x,\hat{\pi}^k(x)))}_{\textbf{S3'}} \label{eqn:prop_decomp_s1_s3_s4}
\end{align}
By Lemma~\ref{lem:S1}, we have $\sqrt{n_k}\textbf{S1}\overset{p}{\to} 0$. By Assumption~\ref{assump:conv_nuisance_params}, $\sqrt{n_k}\textbf{S4}\overset{p}{\to} 0$. By Lemma~\ref{lem:S3'}, $\sqrt{n_k}\textbf{S3'}\overset{p}{\to} 0$. %and by the exact same decomposition as the first term with $\hat{\pi}^k(x)$ replaced by $\hat{a}^k$, we can similarly show that the second term is also $o_P(n_k^{-1/2})$. 
The second term (Equation~\ref{eqn:prop_second_term}) is $o_P(n_k^{-1/2})$ by Lemma~\ref{lem:m0_hata}. 
Combining the two terms gives us the final statement. 
% **EB: have to add something like old lemma 3.5 back in, but only has to handle nusiance parameters for r/p. 
\end{proof}

\begin{lemma}
\label{lem:avg_h_and_avg_h2_with_oracle_arp_Op1}
Under Assumptions~\ref{assump:SUTVA}--\ref{ass:consistency}, for each fold $k$, the empirical average of $h$ and $h^2$ for any $\pi'$ and $a'$ given the oracle nuisance parameters $r$,$p$, is 
\begin{equation}
\frac{1}{n_k} \sum_i h_{\pi',a',r,p}(w_i) = O_p(1),
\end{equation}
and
\begin{equation}
\frac{1}{n_k} \sum_i h_{\pi',a',r,p}(w_i)^2 = O_p(1),
\end{equation}
\end{lemma}
\begin{proof}
%As part of this, we will first show that the average of $h^2_{\pi',a',r,p}$ is $O_p(1)$.

To show the first equation, 
by the triangle inequality, 
for any policy $\pi$ and action $a$:

%\frac{1}{n_k} \sum_i |h_{\pi',a',r,p}(w)| &\leq& \frac{2 |\epsilon_i|}{\eta} + 2 B.
\begin{eqnarray}
\frac{1}{n_k} \sum_i |h_{\pi',a',r,p}(w_i)| &\leq & 2 B  + \frac{1}{n_k} \sum_i \frac{2 |\epsilon_i|}{\eta}, \label{eq:h_bounded_with_oracle_rp}
\end{eqnarray}
where the inequality holds because $|\1(A = \pi(X))| \leq 1$, $p(a|x) \geq \eta$ (Assumption~\ref{assump:strong_overlap}), $Y(x_i,a_i)-r(x_i,a_i)= \epsilon_i$ and $|r(x, a)| \leq B$ $\forall$ $x,a$(Assumption~\ref{assump:bounded_r}).

Then using Markov's inequality, 
\begin{eqnarray}
P ( \frac{1}{n_k} \sum_i |h_{\pi',a',r,p}(w_i)| > M) &\leq& \frac{E [ \frac{1}{n_k} \sum_i |h_{\pi',a',r,p}(w_i)| ] }{M} \\
& \leq & \frac{2 B + \frac{2 E[\epsilon_i]}{\eta}}{M} \\
& = & 2 B / M ,
\end{eqnarray}
where the equality follows because $E[\epsilon_i]=0.$ Therefore $\frac{1}{n_k} \sum_i |h_{\pi',a',r,p}(w_i)| = O_p(1)$.

To prove the second equation in the Lemma, we again start from Equation~\ref{eq:h_bounded_with_oracle_rp} and square each term:
\begin{eqnarray}
\frac{1}{n_k} \sum_i h_{\pi',a',r,p}(w_i)^2 &\leq & \frac{1}{n_k} \sum_i  4 B^2 + \frac{ \epsilon_i^2}{\eta^2} + \frac{ 2 B \epsilon_i }{\eta} 
\end{eqnarray}
To bound the second term, we note that from Assumption~\ref{assump:finite_var}, $E[\epsilon^2_i] \leq C$, and therefore from Markov’s inequality
\begin{eqnarray}
P( \frac{1}{n_k} \sum  \frac{\epsilon_i^2}{\eta^2} > M) \leq  \frac{1}{\eta^2} E [  \frac{1}{n_k} \sum  \epsilon_i^2 ] {M}  \leq \frac{C}{M \eta^2} \to  \frac{1}{n_k} \sum  \frac{\epsilon_i^2}{\eta^2} = O_p(1).
\end{eqnarray}
Finally, the cross term is also $O_p(1)$:
\begin{eqnarray}
E [ \frac{B}{n_k} \sum_i | \epsilon_i| ] \leq  \frac{B}{n_k} \sum_i \sqrt{ E[ \epsilon^2_i] } \leq B \sqrt{C},
\end{eqnarray}
where the inequality uses Jensen’s and the last step uses the Assumption~\ref{assump:finite_var}. Combining this with Markov yields the cross term is $O_p(1)$, and therefore   
\begin{equation}
\frac{1}{n_k} \sum_i h_{\pi',a',r,p}(w_i)^2 = O_p(1)
\end{equation}
\end{proof}

\begin{lemma}
\label{lem:h_hat_pi_oracle_neq_astart_oh_p_1}
Under Assumptions~\ref{assump:SUTVA}--\ref{ass:consistency}, under the null hypothesis $H_0$, for each fold $k$,  given the oracle nuisance parameters $a^* $,$r$,$p$,  
\begin{equation}
\sqrt{n_k} \frac{1}{n_k} \sum_i h_{\pi,a^*,r,p}(w_i)I(\pi(x) \neq a^*) = o_p(1)
\end{equation}
\end{lemma}
\begin{proof}
\begin{eqnarray}
\sqrt{n_k} \frac{1}{n_k} \sum_i h_{\pi,a^*,r,p}(w_i)I(\hat{\pi}(x) \neq a^*) &\leq& 
\sqrt{n_k} \frac{1}{n_k} \sum_i \frac{1}{p(a|x)} (\I(a_i = \hat{\pi}(x)) - \I(a_i = a^*) I(\hat{\pi}(x) \neq a^*) \epsilon_i  \nonumber \\
&\leq& 
\sqrt{n_k} \frac{1}{\eta} \frac{1}{n_k} \sum_i  I(\hat{\pi}(x) \neq a^*) \epsilon_i,
\end{eqnarray}
where the first inequality holds because $r(x_i,\hat{\pi}(x_i)) - r(x_i,a^*) \leq 0$ for any $\hat{\pi}$ under the null hypothesis $H_0$.  
Note that since $E[\epsilon]=0$, the expected value of the above is 0.
Then by the Chebyshev inequality
\begin{eqnarray}
P( |\sqrt{n_k} \frac{1}{\eta} \frac{1}{n_k} \sum_i  I(\hat{\pi}(x) \neq a^*) \epsilon_i| \geq c_1) &\leq& \frac{E[ (\sqrt{n_k} \frac{1}{\eta} \frac{1}{n_k} \sum_i  I(\hat{\pi}(x) \neq a^*) \epsilon_i)^2] }{c_1^2} \\
&=& \frac{\frac{1}{n_k} \sum_i  E[ I(\hat{\pi}(x_i) \neq a^*)\epsilon_i^2] }{\eta^2 c_1^2} \\
&=& \frac{\frac{1}{n_k} \sum_i  E[ I(\hat{\pi}(x_i) \neq a^*)] C }{\eta^2 c_1^2},\\
&=& \frac{ E[ I(\hat{\pi}(x_i) \neq a^*)] C }{\eta^2 c_1^2},\\
&=& \frac{C}{\eta^2 c_1^2} \cdot o_p(1) \\
&=& o_p(1)
\end{eqnarray}
where the first equality holds because all $\epsilon_i$ are i.i.d. and so the cross terms vanish, and the second inequality holds by Assumption~\ref{assump:finite_var} (finite variance), and the second to last equality holds by consistency of the policy learning (Assumption~\ref{ass:consistency}), which, under $H_0$, implies $\pi(x) \to a^*$ $\forall x$.

\end{proof}

\begin{lemma}\label{lem:m0_k}
Let 

\begin{equation}\label{eqn:m0_star}
    \hat{m}_{(0)}^{k,*}(w)=r(x,a^*(x)) + \frac{\mathbf{1}\{a^*(x)=a\}}{p(a|x)}(y-\hat{r}(x,a^*(x))).
\end{equation}
and consider $\hat{m}_{(0)}^k$ (as defined in  Equation~\eqref{eqn:m0_k}). Then under Assumptions \ref{assump:strong_overlap}--\ref{assump:finite_var}, we have for each $k$,  
\[\sqrt{n_k}\frac{1}{n_k}\sum_{i\in\II_{n_k}}\hat{m}_{(0)}^k(w_i)-\hat{m}_{(0)}^{k,*}(w_i)\overset{p}{\to} 0. \]
\end{lemma}
\begin{proof}
Note that with Assumption~\ref{assump:bounded_r}, we have for each $k$,
\[\E_x[r(x,\hat{a}^k(x))-r(x,a^*(x))] \leq C_1 \E_x[\1\{\hat{a}^k(x)\neq a^*(x)\}]=o(n_k^{-1/2})\]
by Assumption \ref{assump:fast_best_arm_learner}. The proof then follows exactly from the proof of Lemma~\ref{lem:m1_k} with $\hat{\pi}^k(x_i)$ replaced by $\hat{a}^k$. 
\end{proof}

\begin{lemma}\label{lem:m0_hata}
Let $\hat{m}_{(0)}^k$ be defined in Equation~\eqref{eqn:m0_k} and $\hat{m}_{(0)}^{k,\hat{a}}$ be defined in \eqref{eqn:m0_hata}. Then under Assumptions \ref{assump:strong_overlap}--\ref{assump:consistent} and \ref{assump:conv_nuisance_params}--\ref{assump:finite_var}, we have for each $k$,  
\[\sqrt{n_k}\frac{1}{n_k}\sum_{i\in\II_{n_k}}\hat{m}_{(0)}^k(w_i)-\hat{m}_{(0)}^{k,\hat{a}}(w_i)\overset{p}{\to} 0. \]
\end{lemma}
\begin{proof}
We consider some fixed $k$. Note that 
\begin{align*}
    &\hat{m}_{(0)}^k(w_i)-\hat{m}_{(0)}^{k,\hat{a}}(w_i)\\
    &= \hat{r}^k(x_i,\hat{a}^k) + \frac{\mathbf{1}\{\hat{a}^k=a_i\}}{\hat{p}^k(a_i|x_i)}(y_i-\hat{r}^k(x_i,\hat{a}^k))-r(x_i,\hat{a}^k) - \frac{\mathbf{1}\{\hat{a}^k=a_i\}}{p(a_i|x_i)}(y_i-r(x_i,\hat{a}^k))\\
    &=  \bigsmile{1-\frac{\mathbf{1}\{\hat{a}^k=a_i\}}{p(a_i|x_i)}}(\hat{r}^k(x_i,\hat{a}^k)-r(x_i,\hat{a}^k)+\bigsmile{1-\frac{\mathbf{1}\{\hat{a}^k=a_i\}}{p(a_i|x_i)}}r(x_i,\hat{a}^k)\\
    &\qquad+\frac{\mathbf{1}\{\hat{a}^k=a_i\}}{p(a_i|x_i)}\hat{r}^k(x_i,\hat{a}^k)-\frac{\mathbf{1}\{\hat{a}^k=a_i\}}{\hat{p}^k(a_i|x_i)}\hat{r}^k(x_i,\hat{a}^k)+\frac{\mathbf{1}\{\hat{a}^k=a_i\}}{\hat{p}^k(a_i|x_i)}y_i\\
    &\qquad-r(x_i,\hat{a}^k) - \frac{\mathbf{1}\{\hat{a}^k=a_i\}}{p(a_i|x_i)}(y_i-r(x_i,\hat{a}^k))\\
    &= \underbrace{ \bigsmile{1-\frac{\mathbf{1}\{\hat{a}^k=a_i\}}{p(a_i|x_i)}}(\hat{r}^k(x_i,\hat{a}^k)-r(x_i,\hat{a}^k))}_{\mathbf{S1_i}}+\bigsmile{\frac{\mathbf{1}\{\hat{a}^k=a_i\}}{p(a_i|x_i)}-\frac{\mathbf{1}\{\hat{a}^k=a_i\}}{\hat{p}^k(a_i|x_i)}}\bigsmile{\hat{r}(x_i,\hat{a}^k)-r(x_i,\hat{a}^k)}\\
    &\qquad-\frac{\mathbf{1}\{\hat{a}^k=a_i\}}{\hat{p}^k(a_i|x_i)}r(x_i,\hat{a}^k))+r(x_i,\hat{a}^k)+\frac{\mathbf{1}\{\hat{a}^k=a_i\}}{\hat{p}^k(a_i|x_i)}y_i-r(x_i,\hat{a}^k) - \frac{\mathbf{1}\{\hat{a}^k=a_i\}}{p(a_i|x_i)}(y_i-r(x_i,\hat{a}^k))\\
    &= \mathbf{S1_i}+\underbrace{\bigsmile{\frac{\mathbf{1}\{\hat{a}^k=a_i\}}{p(a_i|x_i)}-\frac{\mathbf{1}\{\hat{a}^k=a_i\}}{\hat{p}^k(a_i|x_i)}}\bigsmile{\hat{r}(x_i,\hat{a}^k)-r(x_i,\hat{a}^k)}}_{\mathbf{S4_i}} \\
    &\qquad+\frac{\mathbf{1}\{\hat{a}^k=a_i\}}{p(a_i|x_i)}r(x_i,\hat{a}^k)+\frac{\mathbf{1}\{\hat{a}^k=a_i\}}{\hat{p}^k(a_i|x_i)}y_i-\frac{\mathbf{1}\{\hat{a}^k=a_i\}}{p(a_i|x_i)}y_i-\frac{\mathbf{1}\{\hat{a}^k=a_i\}}{\hat{p}^k(a_i|x_i)}r(x_i,\hat{a}^k)\\
    &=\mathbf{S1_i}+\mathbf{S4_i}+ \underbrace{\bigsmile{\frac{\mathbf{1}\{\hat{a}^k=a_i\}}{\hat{p}^k(a_i|x_i)}-\frac{\mathbf{1}\{\hat{a}^k=a_i\}}{p(a_i|x_i)}}(y_i-r(x_i,\hat{a}^k))}_{\mathbf{S3_i}}.
\end{align*}
By Assumption~\ref{assump:consistent} and Lemma~\ref{lem:S1}, we show that $\sqrt{n_k}\frac{1}{n_k}\sum_{i\in\II_k}\mathbf{S1_i}\overset{p}{\to}0$. For $\mathbf{S4_i}$, we have $\sqrt{n_k}\frac{1}{n_k}\sum_{i\in\II_k}\mathbf{S4_i}\overset{p}{\to}0$ by Assumption~\ref{assump:conv_nuisance_params}. For $\mathbf{S3_i}$, note that $y_i=r(x_i,a_i)+\epsilon_i$, 
\begin{align*}
    &\bigsmile{\frac{\mathbf{1}\{\hat{a}^k=a_i\}}{\hat{p}^k(a_i|x_i)}-\frac{\mathbf{1}\{\hat{a}^k=a_i\}}{p(a_i|x_i)}}(y_i-r(x_i,\hat{a}^k))\\
    &= \bigsmile{\frac{\mathbf{1}\{\hat{a}^k=a_i\}}{\hat{p}^k(a_i|x_i)}-\frac{\mathbf{1}\{\hat{a}^k=a_i\}}{p(a_i|x_i)}}(r(x_i,a_i)-r(x_i,\hat{a}^k)+\epsilon_i) \\
    &= \bigsmile{\frac{\mathbf{1}\{\hat{a}^k=a_i\}}{\hat{p}^k(a_i|x_i)}-\frac{\mathbf{1}\{\hat{a}^k=a_i\}}{p(a_i|x_i)}}\epsilon_i.
\end{align*}
where the last line follows since the expression is only nonzero if $\hat{a}^k=a_i$. First, the expectation 
\begin{align*}
    \E\bigbrak{\bigsmile{\frac{\mathbf{1}\{\hat{a}^k=a_i\}}{\hat{p}^k(a_i|x_i)}-\frac{\mathbf{1}\{\hat{a}^k=a_i\}}{p(a_i|x_i)}}\epsilon_i}=0
\end{align*}
since all nuisance parameters $\hat{p}^k$ and $\hat{a}^k$ are learned on data excluding the $k$th fold. Also, 
\begin{align*}
    &\Var\bigsmile{\sqrt{n_k}\frac{1}{n_k}\sum_{i\in\II_k}\bigsmile{\frac{\mathbf{1}\{\hat{a}^k=a_i\}}{\hat{p}^k(a_i|x_i)}-\frac{\mathbf{1}\{\hat{a}^k=a_i\}}{p(a_i|x_i)}}\epsilon_i}\\
    &= \frac{1}{n_k}\sum_{i\in\II_k}\Var\bigsmile{\bigsmile{\frac{\mathbf{1}\{\hat{a}^k=a_i\}}{\hat{p}^k(a_i|x_i)}-\frac{\mathbf{1}\{\hat{a}^k=a_i\}}{p(a_i|x_i)}}\epsilon_i}\\
    &= \frac{1}{n_k}\sum_{i\in\II_k}\E\bigbrak{\bigsmile{\bigsmile{\frac{\mathbf{1}\{\hat{a}^k=a_i\}}{\hat{p}^k(a_i|x_i)}-\frac{\mathbf{1}\{\hat{a}^k=a_i\}}{p(a_i|x_i)}}\epsilon_i}^2}\\
    &= \frac{1}{n_k}\sum_{i\in\II_k}\E\bigbrak{\mathbf{1}\{\hat{a}^k=a_i\}\bigsmile{\bigsmile{\frac{1}{\hat{p}^k(a_i|x_i)}-\frac{1}{p(a_i|x_i)}}^2\epsilon_i^2}}\\
    &\leq \frac{C}{n_k}\sum_{i\in\II_k}\E\bigbrak{\bigsmile{\frac{1}{\hat{p}^k(a_i|x_i)}-\frac{1}{p(a_i|x_i)}}^2}.
\end{align*}
By Assumption~\ref{assump:consistent}, $\E\bigbrak{\bigsmile{\frac{1}{\hat{p}^k(a_i|x_i)}-\frac{1}{p(a_i|x_i)}}^2}=o(1)$. By Chebyshev inequality we have that $\sqrt{n_k}\frac{1}{n_k}\sum_{i\in\II_k}\mathbf{S3_i}\overset{p}{\to}0$. Combining all three terms gives the statement. 
\end{proof}

\subsection{Efficiency with known nuisance parameters}
\label{sec:efficiency_known}
Now we show this procedure is efficient under conditions. We start with some common assumptions made by usual causal inference literatures. 

Also, for notational convenience we define the following terms. 
{
\small
\begin{align}
    \Gamma_i^* &:=r(x_i,\pi^*(x_i)) - r(x_i,a^*_{all}) + \frac{\mathbf{1}\{\pi^*(x_i)=a_i,a^*_{all}\ne \pi^*(x_i)\}}{p(a_i|x_i)}(y_i-r(x_i,\pi^*(x_i)))\nonumber\\
    &\qquad- \frac{\mathbf{1}\{a^*_{all}=a_i,a^*_{all}\ne \pi^*(x_i)\}}{p(a_i|x_i)}(y_i-r(x_i,a^*_{all}))\label{eqn:Gamma_i_star_1}\\
    &= r(x_i,\pi^*(x_i)) - r(x_i,a^*_{all}) + \frac{\mathbf{1}\{\pi^*(x_i)=a_i\}}{p(a_i|x_i)}(y_i-r(x_i,\pi^*(x_i))) - \frac{\mathbf{1}\{a^*_{all}=a_i\}}{p(a_i|x_i)}(y_i-r(x_i,a^*_{all}))\label{eqn:Gamma_i_star_2}\\
    \hat{\Gamma}_i^k &:= \hat{r}^{k}(x_i,\hat{\pi}^k(x_i)) - \hat{r}^{k}(x_i,\hat{a}^k) + \frac{\mathbf{1}\{\hat{\pi}^k(x_i)=a_i,\hat{a}^k\ne \hat{\pi}^k(x_i)\}}{p(a_i|x_i)}(y_i-\hat{r}^{k}(x_i,\hat{\pi}^k(x_i)))\nonumber\\
    &\qquad- \frac{\mathbf{1}\{\hat{a}^k=a_i,\hat{a}^k\ne \hat{\pi}^k(x_i)\}}{p(a_i|x_i)}(y_i-\hat{r}^{k}(x_i,\hat{a}^k_{all})).
\end{align}
}
Also, we define the oracle estimator $\hat{\psi}^{AIPW}_*=\frac{1}{n}\sum_{i=1}^n \Gamma_i^*$. 

We know that if $\hat{\pi}(x_i)$ is never equal to $\hat{a}_{all}$, then our AIPW estimator reduces to the same AIPW estimator for estimating treatment effect, which by Theorem 3.4 of \cite{wager2024causal} would be efficient. Now we consider the case where $\hat{\pi}(x_i)=\hat{a}_{all}$ for some $i$. We first state the following proposition that shows our oracled version of the estimator is still unbiased even if $\pi^*(x_i)=a^*_{all}$ for some $i$.
\begin{proposition}\label{prop:IPW_bias}
Under Assumptions \ref{assump:unconfoundedness}, $\hat{\psi}_*^{AIPW}$ is unbiased, i.e. $\E[\hat{\psi}_*^{AIPW}]=\psi$.
\end{proposition}
\begin{proof}
Note that 
\begin{align*}
    \hat{\psi}^{AIPW}_*&=\frac{1}{n}\sum_{i=1}^n \Gamma_i^* \\
    &= \frac{1}{n}\sum_{i=1}^n r(x_i,\pi^*(x_i)) - r(x_i,a^*_{all}) + \frac{\mathbf{1}\{\pi^*(x_i)=a_i,a^*_{all}\ne \pi^*(x_i)\}}{p(a_i|x_i)}(y_i-r(x_i,\pi^*(x_i)))\\
    &\qquad- \frac{\mathbf{1}\{a^*_{all}=a_i,a^*_{all}\ne \pi^*(x_i)\}}{p(a_i|x_i)}(y_i-r(x_i,a^*_{all}))\\
    &=\frac{1}{n_k}\sum_{i=1}^{n_k} r(x_i,\pi^*(x_i)) - r(x_i,a_{all}^*) + \frac{\mathbf{1}\{\pi^*(x_i)=a_i\}}{p(a_i|x_i)}(y_i-r(x_i,\pi^*(x_i))) \\
    &\qquad - \frac{\mathbf{1}\{a_{all}^*=a_i\}}{p(a_i|x_i)}(y_i-r(x_i,a_{all}^*)),
\end{align*}
therefore in what follows we focus on the latter estimator. Note that for each $i$,  
\begin{align*}
    \E\bigbrak{r(x_i,\pi^*(x_i))\bigsmile{1-\frac{\mathbf{1}\{\pi^*(x_i)=a_i \}}{p(a_i|x_i)}}} &= \E_{x_i}\left[r(x_i,\pi^*(x_i))\E_{a_i}\bigbrak{1-\frac{\mathbf{1}\{\pi^*(x_i)=a_i\}}{p(a_i|x_i)}\Big|x_i}\right]
\end{align*}
where the second term 
\begin{align*}
    \E_{a_i}\bigbrak{1-\frac{\mathbf{1}\{\pi^*(x_i)=a_i\}}{p(a_i|x_i)}\Big|x_i} &= 1 - \sum_{a'\in\A} p(a'|x_i) \frac{\mathbf{1}\{\pi^*(x_i)=a'\}}{p(a'|x_i)}=0.
\end{align*}
Similarly, we have for each $i$, 
\begin{align*}
    \E\bigbrak{r(x_i,a_{all}^*)\bigsmile{1-\frac{\mathbf{1}\{a_{all}^*=a_i \}}{p(a_i|x_i)}}} &= \E_{x_i}\left[r(x_i,a_{all}^*)\E_{a_i}\bigbrak{1-\frac{\mathbf{1}\{a_{all}^*=a_i \}}{p(a_i|x_i)}\Big|x_i}\right]=0.
\end{align*}
Then we have 
\begin{align*}
    \E[\hat{\psi}^{AIPW}_*] &= \frac{1}{n_k}\sum_{i=1}^{n_k}\E\bigbrak{\frac{\mathbf{1}\{\pi^*(x_i)=a_i\}}{p(a_i|x_i)}y_i- \frac{\mathbf{1}\{a_{all}^*=a_i\}}{p(a_i|x_i)} y_i} \\
    &= \E_X[r(X,\pi^*(X)) - r(X, a_{all}^*)] = V(\pi^*) - \max_a \sum_{x}p(x) r(x,a) = \psi.
\end{align*}
\end{proof}
We can also compute the variance $\Var(\hat{\psi}_*^{AIPW})$, the variance of the oracle estimator. We have by independence
\begin{align*}
    \Var(\hat{\psi}_*^{AIPW}) &= \Var\bigsmile{\frac{1}{n}\sum_{i=1}^n \Gamma_i^*} = \frac{1}{n}\sum_{i=1}^n \Var(\Gamma_i^*).
\end{align*}
The following proposition shows the asymptotic normality of the oracle estimator and characterizes the variance for each $\Gamma_i^*$. 
\begin{proposition}\label{prop:asymp_normal_oracle}
Under Assumption~\ref{assump:unconfoundedness}, we have 
\[\sqrt{n}(\hat{\psi}_*^{AIPW}-\psi)\Rightarrow N(0, \Var(\Gamma_1^*)) \] where for each $i$, 
\begin{align*}
    \Var(\Gamma_i^*)
    &= \Var\Big(r(x_i,\pi^*(x_i)) - r(x_i,a_{all}^*) \\
    &\qquad+ \frac{\mathbf{1}\{\pi^*(x_i)=a_i,a_{all}^*\ne \pi^*(x_i)\}}{p(a_i|x_i)}(y_i-r(x_i,\pi^*(x_i)))\\
    &\qquad- \frac{\mathbf{1}\{a_{all}^*=a_i,a_{all}^*\ne \pi^*(x_i)\}}{p(a_i|x_i)}(y_i-r(x_i,a_{all}^*)) \Big) \\
    &= \Var(r(x_i,\pi^*(x_i) - r(x_i,a_{all}^*))) \\
    &\qquad+ \E\bigbrak{\bigsmile{\frac{\mathbf{1}\{\pi^*(x_i)=a_i,a_{all}^*\ne \pi^*(x_i)\}}{p(a_i|x_i)}(y_i-r(x_i,\pi^*(x_i)))}^2}\\
    &\qquad+ \E\bigbrak{\bigsmile{\frac{\mathbf{1}\{a_{all}^*=a_i,a_{all}^*\ne \pi^*(x_i)\}}{p(a_i|x_i)}(y_i-r(x_i,a_{all}^*))}^2}. 
\end{align*}
\end{proposition}
\begin{proof}
Since $\hat{\psi}_*^{AIPW}$ is a sum of $n$ i.i.d. terms, so by CLT we have the first part of the statement above. As for the variance, by law of total variance, 
\begin{align*}
    \Var(\Gamma_i^*) &= \E[\Var(\Gamma_i^* | x_i)] + \Var(\E[\Gamma_i^* | x_i]).
\end{align*}
Since $\Gamma_i^*$ is an unbiased estimator, we have $\E[\Gamma_i^* | x_i]=r(x_i,\pi^*(x_i)) - r(x_i,a_{all}^*)$. Then the second term $\Var(\E[\Gamma_i^* | x_i])=\Var(r(x_i,\pi^*(x_i)) - r(x_i,a_{all}^*))$. The first term becomes
{
\small
\begin{align*}
    &\E[\Var(\Gamma_i^* | x_i)]\\
    &= \E_{x_i}\left[\Var_{a_i,y_i}\left(r(x_i,\pi^*(x_i)) - r(x_i,a_{all}^*)+\frac{\mathbf{1}\{\pi^*(x_i)=a_i,a_{all}^*\ne \pi^*(x_i)\}}{p(a_i|x_i)}(y_i-r(x_i,a_i))\right.\right. \\
    &\qquad- \left.\left.\frac{\mathbf{1}\{a_{all}^*=a_i,a_{all}^*\ne \pi^*(x_i)\}}{p(a_i|x_i)}(y_i-r(x_i,a_i))\Big | x_i\right)\right]\\
    &= \E_{x_i}\left[\Var_{a_i,y_i}\left(\frac{\mathbf{1}\{\pi^*(x_i)=a_i,a_{all}^*\ne \pi^*(x_i)\}}{p(a_i|x_i)}(y_i-r(x_i,a_i))-\frac{\mathbf{1}\{a_{all}^*=a_i,a_{all}^*\ne \pi^*(x_i)\}}{p(a_i|x_i)}(y_i-r(x_i,a_i))\Big | x_i\right)\right]\tag{the first two terms are constants given $x_i$}\\
    &= \E_{x_i}\bigbrak{\E_{a_i,y_i}\left[\bigsmile{\frac{\mathbf{1}\{\pi^*(x_i)=a_i,a_{all}^*\ne \pi^*(x_i)\}}{p(a_i|x_i)}(y_i-r(x_i,a_i)) - \frac{\mathbf{1}\{a_{all}^*=a_i,a_{all}^*\ne \pi^*(x_i)\}}{p(a_i|x_i)}(y_i-r(x_i,a_i))\Big | x_i}^2\right]}\tag{the term has mean zero due to $\E[y_i|x_i,a_i]=r(x_i,a_i)$}\\
    &= \E_{x_i}\Big[\E_{a_i,y_i}\Big[\Big(\frac{\mathbf{1}\{\pi^*(x_i)=a_i,a_{all}^*\ne \pi^*(x_i)\}}{p(a_i|x_i)}(y_i-r(x_i,a_i))\Big)^2+ \Big(\frac{\mathbf{1}\{a_{all}^*=a_i,a_{all}^*\ne \pi^*(x_i)\}}{p(a_i|x_i)}(y_i-r(x_i,a_i))\Big)^2 \\
    &\qquad- 2\Big(\frac{\mathbf{1}\{\pi^*(x_i)=a_i,a_{all}^*\ne \pi^*(x_i)\}}{p(a_i|x_i)}(y_i-r(x_i,a_i))\Big)\Big(\frac{\mathbf{1}\{a_{all}^*=a_i,a_{all}^*\ne \pi^*(x_i)\}}{p(a_i|x_i)}(y_i-r(x_i,a_i))\Big) \Big | x_i\Big]\Big]\tag{algebra}\\
    &= \E\bigbrak{\bigsmile{\frac{\mathbf{1}\{\pi^*(x_i)=a_i,a_{all}^*\ne \pi^*(x_i)\}}{p(a_i|x_i)}(y_i-r(x_i,a_i))}^2}+ \E\bigbrak{\bigsmile{\frac{\mathbf{1}\{a_{all}^*=a_i,a_{all}^*\ne \pi^*(x_i)\}}{p(a_i|x_i)}(y_i-r(x_i,a_i))}^2}\\
    &\quad- 2\E_{x_i}\bigbrak{\E_{a_i,y_i}\bigbrak{\Big(\frac{\mathbf{1}\{\pi^*(x_i)=a_i,a_{all}^*\ne \pi^*(x_i)\}}{p(a_i|x_i)}(y_i-r(x_i,a_i))\Big)\Big(\frac{\mathbf{1}\{a_{all}^*=a_i,a_{all}^*\ne \pi^*(x_i)\}}{p(a_i|x_i)}(y_i-r(x_i,a_i))\Big) \Big| x_i}}.
\end{align*}
}
Note that the third term is equal to 
{
\small
\[- 2\E_{x_i}\bigbrak{\E_{a_i}\left[\frac{\mathbf{1}\{\pi^*(x_i)=a_i,a_{all}^*\ne \pi^*(x_i)\}\mathbf{1}\{a_{all}^*=a_i,a_{all}^*\ne \pi^*(x_i)\}}{p(a_i|x_i)^2}\Big| x_i\right]\E_{y_i}\bigbrak{\Big(y_i-r(x_i,\pi^*(x_i))\Big)\Big(y_i-r(x_i,a_{all}^*)\Big) \Big| x_i}}\]
}
by unconfoundedness (Assumption~\ref{assump:unconfoundedness}). Then, note that $\E_{a_i}[\mathbf{1}\{\pi^*(x_i)=a_i,a_{all}^*\ne \pi^*(x_i)\}\mathbf{1}\{a_{all}^*=a_i,a_{all}^*\ne \pi^*(x_i)\} | x_i]=0$, the third term is zero, so we have proved our statement. 
\end{proof}

\subsection{Asymptotic Normality for $\hat{\psi}^{AIPW}$ with estimated propensity score}
\label{sec:efficiency_estimated}
In what follows we show that 
\[\sqrt{n}(\hat{\psi}^{AIPW}-\hat{\psi}_*^{AIPW})\overset{p}{\to} 0\]
with estimated propensity score. In order to show this convergence, we need an additional assumption.  

\begin{assumption}[Low-regret policy learner]
For each $k$, $\E_x[r(x,\hat{\pi}^k(x))-r(x,\pi^*(x))]=o(n^{-1/2})$.\label{assump:consistent_regret}
\end{assumption}
%Assumption~\ref{ass:consistency} requires the policy learner is consistent, which can be satisfied with reasonable learners. Whereas,
Assumption~\ref{assump:consistent_regret} requires that the estimated policy and the estimated best arm converges in a super-parametric rate. Fortunately, these fast-rate conditions have been well-studied in a variety of literatures \cite{luedtke2020performance,zhao2012estimating,kitagawa2018should,qian2011performance,farahmand2011action,chambaz2017targeted}. In particular, 
% \cite{luedtke2020performance} gives a fast rate given that we use ERM learner for $\hat{\pi}$ and our data is sampled from a fixed distribution with certain conditions. Also, 
Assumption~\ref{assump:consistent_regret} can be achieved via a plug-in estimator under certain margin conditions \cite{luedtke2020performance,kitagawa2018should,qian2011performance}, for example, from \cite{qian2011performance},
\begin{enumerate}
    \item[(MA)] There exists $\alpha>0$ and $C<\infty$ such that for all $\epsilon>0$,  
    \[\mathbf{P}\left(\max _{a \in \mathcal{A}} r(X, a)-\max _{a \in \mathcal{A} \backslash a^*} r(X, a) \leq \epsilon\right) \leq C \epsilon^\alpha.\]
\end{enumerate}
Or, when action is binary, we assume the mass of the treatment effect at zero is controlled \cite{luedtke2020performance,kitagawa2018should}: 
\begin{enumerate}
    \item[(MA')] Define $\tau(x):=r(x,1)-r(x,0)$. There exists $\alpha>0$ and $C<\infty$ such that for all $\epsilon>0$,
    \[\mathbf{P}\left(|\tau(X)| \leq \epsilon\right) \leq C \epsilon^\alpha.\]
\end{enumerate}

\begin{lemma}\label{lem:oracle_convergence}
Under Assumptions \ref{assump:SUTVA}--\ref{assump:finite_var} and \ref{ass:consistency}--\ref{assump:consistent_regret}, we have \[\sqrt{n}(\hat{\psi}^{AIPW}-\hat{\psi}_*^{AIPW})\overset{p}{\to} 0.\]
\end{lemma}

\begin{proof}
Note that we can write $\hat{\psi}_*^{AIPW}=\sum_{k=1}^6 \frac{|\mathcal{I}_k|}{n}\hat{\psi}^{AIPW,k}_*$ where 
\[\hat{\psi}^{AIPW,k}_*=\frac{1}{|\mathcal{I}_k|}\sum_{i\in \mathcal{I}_k} \Lambda_i.\]
Further, for each $k$, we can decompose $\hat{\psi}^{AIPW,k}$ as 
\begin{align}
    \hat{\psi}^{AIPW,k} &= \frac{1}{n_k}\sum_{i\in\II_k} \hat{r}^{k}(x_i,\hat{\pi}^k(x_i)) + \frac{\mathbf{1}\{\hat{\pi}^k(x_i)=a_i\}}{\hat{p}^k(a_i|x_i)}(y_i-\hat{r}^{k}(x_i,\hat{\pi}^k(x_i)))\\
    &\qquad- \hat{r}^{k}(x_i,\hat{a}^k)- \frac{\mathbf{1}\{\hat{a}^k=a_i \}}{\hat{p}^k(a_i|x_i)}(y_i-\hat{r}^{k}(x_i,\hat{a}^k))\\
    &=\frac{1}{n_k}\sum_{i\in\II_k}\hat{m}_{(1)}^{k}(w_i)-\hat{m}_{(0)}^{k}(w_i)
\end{align}
where for $w_i=(x_i,a_i,y_i)$, 
\begin{equation}\label{eqn:m1_k}
\hat{m}_{(1)}^k(w_i)=\hat{r}^k(x_i,\hat{\pi}^k(x_i)) + \frac{\mathbf{1}\{\hat{\pi}^k(x_i)=a_i\}}{\hat{p}^k(a_i|x_i)}(y_i-\hat{r}^k(x_i,\hat{\pi}^k(x_i))).
\end{equation}
Analogously we define 
\begin{equation}\label{eqn:m1_star}
\hat{m}_{(1)}^{k,*}(w_i)=r(x_i,\pi^*(x_i)) + \frac{\mathbf{1}\{\pi^*(x_i)=a_i\}}{p(a_i|x_i)}(y_i-r(x_i,\pi^*(x_i)))
\end{equation}
Recall the definition of $\hat{m}_{(0)}^k$ in equation~\eqref{eqn:m0_k} and the definition of $\hat{m}_{(0)}^{k,*}$ in equation~\eqref{eqn:m0_star}, it is sufficient to show for each $k$, $\sqrt{n_k}\frac{1}{n_k}\sum_{i\in\II_k}\left(\hat{m}_{(1)}^k(w_i)-\hat{m}_{(1)}^{k,*}(w_i)\right)\overset{p}{\to} 0$ and $\sqrt{n_k}\frac{1}{n_k}\sum_{i\in\II_k}\left(\hat{m}_{(0)}^k(w_i)-\hat{m}_{(0)}^{k,*}(w_i)\right)\overset{p}{\to} 0$, which we show in Lemmas~\ref{lem:m1_k} and \ref{lem:m0_k} respectively. Repeat the same argument for each $k$ gives the final result. 
\end{proof}
\begin{lemma}\label{lem:m1_k}
Define $\hat{m}_{(1)}^{k}$ as in equation~\eqref{eqn:m1_k} and $\hat{m}_{(1)}^{k,*}$ as in equation~\eqref{eqn:m1_star}, under Assumptions~\ref{assump:conv_nuisance_params},~\ref{assump:consistent},~\ref{assump:consistent_regret},~\ref{ass:consistency}, \ref{assump:bounded_r}, and \ref{assump:finite_var},  for each $k$, \[\sqrt{n_k}\frac{1}{n_k}\sum_{i\in\II_k}\left(\hat{m}_{(1)}^k(w_i)-\hat{m}_{(1)}^{k,*}(w_i)\right)\overset{p}{\to} 0.\]
\end{lemma}
\begin{proof}
We consider some fixed $k$. Note that 
\begin{align*}
    &\hat{m}_{(1)}^k(w_i)-\hat{m}_{(1)}^{k,*}(w_i)\\
    &= \hat{r}^k(x_i,\hat{\pi}^k(x_i)) + \frac{\mathbf{1}\{\hat{\pi}^k(x_i)=a_i\}}{\hat{p}^k(a_i|x_i)}(y_i-\hat{r}^k(x_i,\hat{\pi}^k(x_i)))\\
    &\qquad-r(x_i,\pi^*(x_i)) - \frac{\mathbf{1}\{\pi^*(x_i)=a_i\}}{p(a_i|x_i)}(y_i-r(x_i,\pi^*(x_i)))\\
    &=  \bigsmile{1-\frac{\mathbf{1}\{\hat{\pi}^k(x_i)=a_i\}}{p(a_i|x_i)}}(\hat{r}^k(x_i,\hat{\pi}^k(x_i))-r(x_i,\hat{\pi}^k(x_i))+\bigsmile{1-\frac{\mathbf{1}\{\hat{\pi}^k(x_i)=a_i\}}{p(a_i|x_i)}}r(x_i,\hat{\pi}^k(x_i))\\
    &\qquad+\frac{\mathbf{1}\{\hat{\pi}^k(x_i)=a_i\}}{p(a_i|x_i)}\hat{r}^k(x_i,\hat{\pi}^k(x_i))-\frac{\mathbf{1}\{\hat{\pi}^k(x_i)=a_i\}}{\hat{p}^k(a_i|x_i)}\hat{r}^k(x_i,\hat{\pi}^k(x_i))+\frac{\mathbf{1}\{\hat{\pi}^k(x_i)=a_i\}}{\hat{p}^k(a_i|x_i)}y_i\\
    &\qquad-r(x_i,\pi^*(x_i)) - \frac{\mathbf{1}\{\pi^*(x_i)=a_i\}}{p(a_i|x_i)}(y_i-r(x_i,\pi^*(x_i)))\\
    &= \underbrace{ \bigsmile{1-\frac{\mathbf{1}\{\hat{\pi}^k(x_i)=a_i\}}{p(a_i|x_i)}}(\hat{r}^k(x_i,\hat{\pi}^k(x_i))-r(x_i,\hat{\pi}^k(x_i))}_{\mathbf{S1_i}} \\
    &\qquad+\bigsmile{\frac{\mathbf{1}\{\hat{\pi}^k(x_i)=a_i\}}{p(a_i|x_i)}-\frac{\mathbf{1}\{\hat{\pi}^k(x_i)=a_i\}}{\hat{p}^k(a_i|x_i)}}\bigsmile{\hat{r}(x_i,\hat{\pi}^k(x_i))-r(x_i,\hat{\pi}^k(x_i))}\\
    &\qquad-\frac{\mathbf{1}\{\hat{\pi}^k(x_i)=a_i\}}{\hat{p}^k(a_i|x_i)}r(x_i,\hat{\pi}^k(x_i)))+r(x_i,\hat{\pi}^k(x_i))+\frac{\mathbf{1}\{\hat{\pi}^k(x_i)=a_i\}}{\hat{p}^k(a_i|x_i)}y_i\\
    &\qquad-r(x_i,\pi^*(x_i)) - \frac{\mathbf{1}\{\pi^*(x_i)=a_i\}}{p(a_i|x_i)}(y_i-r(x_i,\pi^*(x_i)))\\
    &= \mathbf{S1_i}+\underbrace{\bigsmile{\frac{\mathbf{1}\{\hat{\pi}^k(x_i)=a_i\}}{p(a_i|x_i)}-\frac{\mathbf{1}\{\hat{\pi}^k(x_i)=a_i\}}{\hat{p}^k(a_i|x_i)}}\bigsmile{\hat{r}(x_i,\hat{\pi}^k(x_i))-r(x_i,\hat{\pi}^k(x_i))}}_{\mathbf{S4_i}} \\
    &\qquad+ \underbrace{\bigsmile{1-\frac{\mathbf{1}\{\pi^*(x_i)=a_i\}}{p(a_i|x_i)}}(r(x_i,\hat{\pi}^k(x_i))-r(x_i,\pi^*(x_i)))}_{\mathbf{S2_i}}\\
    &\qquad+\frac{\mathbf{1}\{\pi^*(x_i)=a_i\}}{p(a_i|x_i)}r(x_i,\hat{\pi}^k(x_i))+\frac{\mathbf{1}\{\hat{\pi}^k(x_i)=a_i\}}{\hat{p}^k(a_i|x_i)}y_i\\
    &\qquad-\frac{\mathbf{1}\{\pi^*(x_i)=a_i\}}{p(a_i|x_i)}y_i-\frac{\mathbf{1}\{\hat{\pi}^k(x_i)=a_i\}}{\hat{p}^k(a_i|x_i)}r(x_i,\hat{\pi}^k(x_i))\\
    &=\mathbf{S1_i}+\mathbf{S4_i}+\mathbf{S2_i}+ \underbrace{\bigsmile{\frac{\mathbf{1}\{\hat{\pi}^k(x_i)=a_i\}}{\hat{p}^k(a_i|x_i)}-\frac{\mathbf{1}\{\pi^*(x_i)=a_i\}}{p(a_i|x_i)}}(y_i-r(x_i,\hat{\pi}^k(x_i)))}_{\mathbf{S3_i}}.
\end{align*}
We then show $\mathbf{S1_i}$, $\mathbf{S2_i}$, and $\mathbf{S3_i}$ separately in Lemmas~\ref{lem:S1}, \ref{lem:S2}, \ref{lem:S3}, respectively. For $\mathbf{S4_i}$, we have $\sqrt{n_k}\frac{1}{n_k}\sum_{i\in\II_k}\mathbf{S4_i}\overset{p}{\to}0$ by Assumption~\ref{assump:conv_nuisance_params}. 
Combining all four terms gives $\sqrt{n_k}(\hat{m}_{(1)}^k-\hat{m}_{(1)}^{k,*})\overset{p}{\to} 0$. 
\end{proof}

\begin{lemma}[Convergence of \textbf{S1}]\label{lem:S1}
For each $k$, we have by Assumption~\ref{assump:consistent},
\[\sqrt{n_k}\frac{1}{n_k}\sum_{i\in\II_k} \bigsmile{1-\frac{\mathbf{1}\{\hat{\pi}^k(x_i)=a_i\}}{p(a_i|x_i)}}(\hat{r}^k(x_i,\hat{\pi}^k(x_i))-r(x_i,\hat{\pi}^k(x_i))\overset{p}{\to} 0.\]
\end{lemma}
\begin{proof}
Note that for each $i$, 
\begin{align}
    &\E\bigbrak{\bigsmile{1-\frac{\mathbf{1}\{\hat{\pi}^k(x_i)=a_i\}}{p(a_i|x_i)}}(\hat{r}^k(x_i,\hat{\pi}^k(x_i))-r(x_i,\hat{\pi}^k(x_i))} \\
    &= \E_{\{x_i\}_{i\in\II_k},\II_{-k}}\bigbrak{\E_{\{a_i\}_{i\in\II_k}}\bigbrak{\bigsmile{1-\frac{\mathbf{1}\{\hat{\pi}^k(x_i)=a_i\}}{p(a_i|x_i)}}(\hat{r}^k(x_i,\hat{\pi}^k(x_i))-r(x_i,\hat{\pi}^k(x_i))\Big| \II_{-k}, \{x_i\}}}\\
    &= \E_{\{x_i\}_{i\in\II_k},\II_{-k}}\bigbrak{(\hat{r}^k(x_i,\hat{\pi}^k(x_i))-r(x_i,\hat{\pi}^k(x_i))\E_{\{a_i\}_{i\in\II_k}}\bigbrak{\bigsmile{1-\frac{\mathbf{1}\{\hat{\pi}^k(x_i)=a_i\}}{p(a_i|x_i)}}\Big| \II_{-k}, \{x_i\}}}\\
    &= \E\bigbrak{(\hat{r}^k(x_i,\hat{\pi}^k(x_i))-r(x_i,\hat{\pi}^k(x_i)) \bigsmile{\sum_{a} p(a|x_i)\bigsmile{1-\frac{\mathbf{1}\{\hat{\pi}^k(x_i)=a\}}{p(a|x_i)}}}}=0
\end{align}
regardless of the reward model estimate. Then, the variance 
\begin{align}
    &\Var\bigbrak{\frac{1}{n_k}\sum_{i\in \II_k}\bigsmile{1-\frac{\mathbf{1}\{\hat{\pi}^k(x_i)=a_i\}}{p(a_i|x_i)}}(\hat{r}^k(x_i,\hat{\pi}^k(x_i))-r(x_i,\hat{\pi}^k(x_i))} \\
    &= \frac{1}{n_k^2}\sum_{i\in \II_k} \Var\bigbrak{\bigsmile{1-\frac{\mathbf{1}\{\hat{\pi}^k(x_i)=a_i\}}{p(a_i|x_i)}}(\hat{r}^k(x_i,\hat{\pi}^k(x_i))-r(x_i,\hat{\pi}^k(x_i))}\tag{since $\hat{r}^k$ and $\hat{\pi}^k$ are learned independently from $\II_k$}\\
     &= \frac{1}{n_k^2}\sum_{i\in \II_k} \E\bigbrak{\bigsmile{\bigsmile{1-\frac{\mathbf{1}\{\hat{\pi}^k(x_i)=a_i\}}{p(a_i|x_i)}}(\hat{r}^k(x_i,\hat{\pi}^k(x_i))-r(x_i,\hat{\pi}^k(x_i))}^2}\tag{because the expectation is zero as shown above}\\
    &= \frac{1}{n_k^2}\sum_{i\in \II_k}\E_{\{x_i\}_{i\in\II_k},\II_{-k}}\bigbrak{(\hat{r}^k(x_i,\hat{\pi}^k(x_i))-r(x_i,\hat{\pi}^k(x_i)))^2\E_{\{a_i\}_{i\in\II_k}}\bigbrak{\bigsmile{1-\frac{\mathbf{1}\{\hat{\pi}^k(x_i)=a_i\}}{p(a_i|x_i)}}^2\Big| \II_{-k}, \{x_i\}}}\\
    &\leq \frac{1}{n_k^2\eta^2}\sum_{i\in \II_k}\E\bigbrak{(\hat{r}^k(x_i,\hat{\pi}^k(x_i))-r(x_i,\hat{\pi}^k(x_i)))^2}=o(1/n_k)
\end{align}
by Assumption~\ref{assump:consistent}. Then, by Chebyshev inequality, i.e. Lemma~\ref{lem:sqrt_conv}, we get the statement above. 
% for any $\epsilon>0$,
% \begin{align*}
%     &\P\left(\frac{1}{\sqrt{n_k}}\sum_{i\in\II_k} \bigsmile{1-\frac{\mathbf{1}\{\hat{\pi}^k(x_i)=a_i\}}{p(a_i|x_i)}}(\hat{r}^k(x_i,\hat{\pi}^k(x_i))-r(x_i,\hat{\pi}^k(x_i))>\epsilon\right)\\
%     &\leq \frac{\Var\bigbrak{\frac{1}{\sqrt{n_k}}\sum_{i\in \II_k}\bigsmile{1-\frac{\mathbf{1}\{\hat{\pi}^k(x_i)=a_i\}}{p(a_i|x_i)}}(\hat{r}^k(x_i,\hat{\pi}^k(x_i))-r(x_i,\hat{\pi}^k(x_i))} }{\epsilon^2} \overset{p}{\to} 0
% \end{align*}
% since the variance is $o(1)$. 
\end{proof}
\begin{lemma}[Convergence of \textbf{S2}]\label{lem:S2}
For each $k$, we have by Assumption~\ref{ass:consistency},
\[\sqrt{n_k}\frac{1}{n_k}\sum_{i\in\II_k}\bigsmile{1-\frac{\mathbf{1}\{\pi^*(x_i)=a_i\}}{p(a_i|x_i)}}(r(x_i,\hat{\pi}^k(x_i))-r(x_i,\pi^*(x_i)))\overset{p}{\to} 0.\]
\end{lemma}
\begin{proof}
Similar to the proof above, we also have 
\[\E\bigbrak{\frac{1}{\sqrt{n_k}}\sum_{i\in\II_k} \bigsmile{1-\frac{\mathbf{1}\{\pi^*(x_i)=a_i\}}{p(a_i|x_i)}}(r(x_i,\hat{\pi}^k(x_i))-r(x_i,\pi^*(x_i)))}=0\] and 
\begin{align}
    &\Var\bigbrak{\frac{1}{n_k}\sum_{i\in\II_k} \bigsmile{1-\frac{\mathbf{1}\{\pi^*(x_i)=a_i\}}{p(a_i|x_i)}}(r(x_i,\hat{\pi}^k(x_i))-r(x_i,\pi^*(x_i)))} \\
    &= \frac{1}{n_k^2}\sum_{i\in \II_k} \E\bigbrak{\bigsmile{\bigsmile{1-\frac{\mathbf{1}\{\pi^*(x_i)=a_i\}}{p(a_i|x_i)}}(r(x_i,\hat{\pi}^k(x_i))-r(x_i,\pi^*(x_i)))}^2}\\
    &\leq \frac{1}{n_k^2\eta^2}\sum_{i\in \II_k}\E\bigbrak{(r(x_i,\hat{\pi}^k(x_i))-r(x_i,\pi^*(x_i)))^2}=o(1/n_k)
\end{align}
by Assumption~\ref{ass:consistency}. Then by Lemma~\ref{lem:sqrt_conv} we get the statement. 
\end{proof}
\begin{lemma}[Convergence of \textbf{S3}]\label{lem:S3}
For each $k$, we have by Assumption~\ref{assump:consistent_regret},~\ref{ass:consistency}, \ref{assump:bounded_r},~\ref{assump:conv_nuisance_params} and \ref{assump:finite_var}, 
\[\sqrt{n_k}\frac{1}{n_k}\sum_{i\in \II_k}\bigsmile{\frac{\mathbf{1}\{\hat{\pi}^k(x_i)=a_i\}}{\hat{p}^k(a_i|x_i)}-\frac{\mathbf{1}\{\pi^*(x_i)=a_i\}}{p(a_i|x_i)}}(y_i-r(x_i,\hat{\pi}^k(x_i)))\overset{p}{\to}0.\]
\end{lemma}
\begin{proof}
Note that for each $i$, 
\begin{align*}
    &\bigsmile{\frac{\mathbf{1}\{\hat{\pi}^k(x_i)=a_i\}}{\hat{p}^k(a_i|x_i)}-\frac{\mathbf{1}\{\pi^*(x_i)=a_i\}}{p(a_i|x_i)}}(y_i-r(x_i,\hat{\pi}^k(x_i)))\\
    &=\underbrace{\bigsmile{\frac{\mathbf{1}\{\hat{\pi}^k(x_i)=a_i\}}{\hat{p}^k(a_i|x_i)}-\frac{\mathbf{1}\{\hat{\pi}^k(x_i)=a_i\}}{p(a_i|x_i)}}(y_i-r(x_i,\hat{\pi}^k(x_i)))}_{\mathbf{S3'_i}}\\
    &\qquad+\underbrace{\bigsmile{\frac{\mathbf{1}\{\hat{\pi}^k(x_i)=a_i\}}{p(a_i|x_i)}-\frac{\mathbf{1}\{\pi^*(x_i)=a_i\}}{p(a_i|x_i)}}(y_i-r(x_i,\hat{\pi}^k(x_i)))}_{\mathbf{S3''_i}}.
\end{align*}
We prove these two terms separately in Lemmas~\ref{lem:S3'} and \ref{lem:S3''} below. Combining both lemmas gives the statement. 
\end{proof}

\begin{lemma}[Convergence of $\mathbf{S3'}$]\label{lem:S3'}
For each $k$, we have by Assumption~\ref{assump:conv_nuisance_params} and~\ref{assump:finite_var}
\[\sqrt{n_k}\frac{1}{n_k}\sum_{i\in\II_k}\bigsmile{\frac{\mathbf{1}\{\hat{\pi}^k(x_i)=a_i\}}{\hat{p}^k(a_i|x_i)}-\frac{\mathbf{1}\{\hat{\pi}^k(x_i)=a_i\}}{p(a_i|x_i)}}(y_i-r(x_i,\hat{\pi}^k(x_i)))\overset{p}{\to}0.\]
\end{lemma}
\begin{proof}
Since $y_i=r(x_i,a_i)+\epsilon_i$ for each $i$, 
\begin{align*}
    &\bigsmile{\frac{\mathbf{1}\{\hat{\pi}^k(x_i)=a_i\}}{\hat{p}^k(a_i|x_i)}-\frac{\mathbf{1}\{\hat{\pi}^k(x_i)=a_i\}}{p(a_i|x_i)}}(y_i-r(x_i,\hat{\pi}^k(x_i)))\\
    &= \bigsmile{\frac{\mathbf{1}\{\hat{\pi}^k(x_i)=a_i\}}{\hat{p}^k(a_i|x_i)}-\frac{\mathbf{1}\{\hat{\pi}^k(x_i)=a_i\}}{p(a_i|x_i)}}(r(x_i,a_i)+\epsilon_i-r(x_i,\hat{\pi}^k(x_i)))\\
    &= \mathbf{1}\{\hat{\pi}^k(x_i)=a_i\}\bigsmile{\frac{1}{\hat{p}^k(\hat{\pi}^k(x_i)|x_i)}-\frac{1}{p(\hat{\pi}^k(x_i)|x_i)}}(r(x_i,\hat{\pi}^k(x_i))+\epsilon_i-r(x_i,\hat{\pi}^k(x_i)))\tag{it is only nonzero when $\hat{\pi}^k(x_i)=a_i$}\\
    &= \mathbf{1}\{\hat{\pi}^k(x_i)=a_i\}\bigsmile{\frac{1}{\hat{p}^k(\hat{\pi}^k(x_i)|x_i)}-\frac{1}{p(\hat{\pi}^k(x_i)|x_i)}}\epsilon_i.
\end{align*}
Then $\E[\mathbf{S3'_i}]=0$ since $\epsilon_i$ is mean-zero and independent of $x_i,a_i$. Also, 
\begin{align*}
    \Var\left(\frac{1}{n_k}\sum_{i\in\II_k}\mathbf{S3'_i}\right) &= \frac{1}{n_k}\E\bigbrak{\bigsmile{\mathbf{1}\{\hat{\pi}^k(x_i)=a_i\}\bigsmile{\frac{1}{\hat{p}^k(\hat{\pi}^k(x_i)|x_i)}-\frac{1}{p(\hat{\pi}^k(x_i)|x_i)}}\epsilon_i}^2}\\
    &= \frac{1}{n_k}\E\bigbrak{\mathbf{1}\{\hat{\pi}^k(x_i)=a_i\}\bigsmile{\frac{1}{\hat{p}^k(\hat{\pi}^k(x_i)|x_i)}-\frac{1}{p(\hat{\pi}^k(x_i)|x_i)}}^2\epsilon_i^2}\\
    &\leq \frac{C}{n_k}\E_{x_i}\bigbrak{\bigsmile{\frac{1}{\hat{p}^k(\hat{\pi}^k(x_i)|x_i)}-\frac{1}{p(\hat{\pi}^k(x_i)|x_i)}}^2}\tag{by Assumption~\ref{assump:finite_var}}\\
    &\leq\frac{C}{n_k} \max_a\E_{x_i}\bigbrak{\bigsmile{\frac{1}{\hat{p}^k(a|x_i)}-\frac{1}{p(a|x_i)}}^2}
\end{align*}
which is $o(1/n_k)$ by Assumption~\ref{assump:conv_nuisance_params}. Then, by Chebyshev inequality, i.e. Lemma~\ref{lem:sqrt_conv}, we get the statement above. 
\end{proof}

\begin{lemma}[Convergence of $\mathbf{S3''}$]\label{lem:S3''}
For each $k$, we have by Assumption~\ref{assump:consistent_regret},~\ref{ass:consistency}, \ref{assump:bounded_r}, and \ref{assump:finite_var}, 
\[\sqrt{n_k}\frac{1}{n_k}\sum_{i\in\II_k}\bigsmile{\frac{\mathbf{1}\{\hat{\pi}^k(x_i)=a_i\}}{p(a_i|x_i)}-\frac{\mathbf{1}\{\pi^*(x_i)=a_i\}}{p(a_i|x_i)}}(y_i-r(x_i,\hat{\pi}^k(x_i)))\overset{p}{\to}0.\]
\end{lemma}
\begin{proof}
Similar to the proof above, for each $i$, 
\begin{align*}
    &\bigsmile{\frac{\mathbf{1}\{\hat{\pi}^k(x_i)=a_i\}}{p(a_i|x_i)}-\frac{\mathbf{1}\{\pi^*(x_i)=a_i\}}{p(a_i|x_i)}}(y_i-r(x_i,\hat{\pi}^k(x_i)))\\
    &= \bigsmile{\frac{\mathbf{1}\{\hat{\pi}^k(x_i)=a_i\}}{p(a_i|x_i)}-\frac{\mathbf{1}\{\pi^*(x_i)=a_i\}}{p(a_i|x_i)}}(r(x_i,a_i)+\epsilon_i-r(x_i,\hat{\pi}^k(x_i))).
\end{align*}
Therefore, 
\begin{align}
    &\E\left[\frac{1}{n_k}\sum_{i\in\II_k}\bigsmile{\frac{\mathbf{1}\{\hat{\pi}^k(x_i)=a_i\}}{p(a_i|x_i)}-\frac{\mathbf{1}\{\pi^*(x_i)=a_i\}}{p(a_i|x_i)}}(r(x_i,a_i)+\epsilon_i-r(x_i,\hat{\pi}^k(x_i)))\right]\\
    &= \frac{1}{n_k}\sum_{i\in\II_k}\E_{\{(x_i,a_i)\}_{i\in\II_k}}\bigbrak{\bigsmile{\frac{\mathbf{1}\{\hat{\pi}^k(x_i)=a_i\}}{p(a_i|x_i)}-\frac{\mathbf{1}\{\pi^*(x_i)=a_i\}}{p(a_i|x_i)}}(r(x_i,a_i)-r(x_i,\hat{\pi}^k(x_i)))}\tag{$\E[\epsilon_i|x_i,a_i]=0$ for each $i$}\\
    &= \frac{1}{n_k}\sum_{i\in\II_k}\E_{\{x_i\}_{i\in\II_k}}\bigbrak{\sum_{a}p(a|x_i)\bigsmile{\frac{\mathbf{1}\{\hat{\pi}^k(x_i)=a\}}{p(a|x_i)}-\frac{\mathbf{1}\{\pi^*(x_i)=a\}}{p(a|x_i)}}(r(x_i,a)-r(x_i,\hat{\pi}^k(x_i)))}\\
    &= \frac{1}{n_k}\sum_{i\in\II_k}\E_{\{x_i\}_{i\in\II_k}}\bigbrak{\sum_{a}(\mathbf{1}\{\hat{\pi}^k(x_i)=a\}-\mathbf{1}\{\pi^*(x_i)=a\})(r(x_i,a)-r(x_i,\hat{\pi}^k(x_i)))}\\
    &= -\E_{x}\bigbrak{r(x,\pi^*(x))-r(x,\hat{\pi}^k(x))}=o(n_k^{-1/2})\label{eqn:small_expectation}
\end{align}
by Assumption~\ref{assump:consistent_regret}. Also, 
\begin{align}
    &\Var\bigbrak{\frac{1}{n_k}\sum_{i\in\II_k}\bigsmile{\frac{\mathbf{1}\{\hat{\pi}^k(x_i)=a_i\}}{p(a_i|x_i)}-\frac{\mathbf{1}\{\pi^*(x_i)=a_i\}}{p(a_i|x_i)}}(r(x_i,a_i)+\epsilon_i-r(x_i,\hat{\pi}^k(x_i)))}\\
    &=\frac{1}{n_k^2}\sum_{i\in\II_k}\Var\left[\bigsmile{\frac{\mathbf{1}\{\hat{\pi}^k(x_i)=a_i\}}{p(a_i|x_i)}-\frac{\mathbf{1}\{\pi^*(x_i)=a_i\}}{p(a_i|x_i)}}(r(x_i,a_i)+\epsilon_i-r(x_i,\hat{\pi}^k(x_i)))\right]\\
    &= \frac{1}{n_k^2}\sum_{i\in\II_k}\E\left[\bigsmile{\bigsmile{\frac{\mathbf{1}\{\hat{\pi}^k(x_i)=a_i\}}{p(a_i|x_i)}-\frac{\mathbf{1}\{\pi^*(x_i)=a_i\}}{p(a_i|x_i)}}(r(x_i,a_i)+\epsilon_i-r(x_i,\hat{\pi}^k(x_i)))}^2\right]+o(n_k^{-1})\tag{from equation~\eqref{eqn:small_expectation}}\\
    &= \frac{1}{n_k^2}\sum_{i\in\II_k}\E\left[\bigsmile{\frac{\mathbf{1}\{\hat{\pi}^k(x_i)=a_i\}}{p(a_i|x_i)}-\frac{\mathbf{1}\{\pi^*(x_i)=a_i\}}{p(a_i|x_i)}}^2\bigsmile{(r(x_i,a_i)+\epsilon_i-r(x_i,\hat{\pi}^k(x_i)))}^2\right]+o(n_k^{-1})\\
    &= \frac{1}{n_k^2}\sum_{i\in\II_k}\E\left[\bigsmile{\frac{\mathbf{1}\{\hat{\pi}^k(x_i)=a_i\}}{p(a_i|x_i)}-\frac{\mathbf{1}\{\pi^*(x_i)=a_i\}}{p(a_i|x_i)}}^2\bigsmile{(r(x_i,a_i)-r(x_i,\hat{\pi}^k(x_i)))}^2\right]\tag{cross-term has expectation zero}\\
    &\qquad+\frac{1}{n_k^2}\sum_{i\in\II_k}\E\bigbrak{\bigsmile{\frac{\mathbf{1}\{\hat{\pi}^k(x_i)=a_i\}}{p(a_i|x_i)}-\frac{\mathbf{1}\{\pi^*(x_i)=a_i\}}{p(a_i|x_i)}}^2\epsilon_i^2}+o(n_k^{-1}).\label{eqn：third_term_decomp}
\end{align}
Note that for each $i$, 
\begin{align*}
    &\E\bigbrak{\bigsmile{\frac{\mathbf{1}\{\hat{\pi}^k(x_i)=a_i\}}{p(a_i|x_i)}-\frac{\mathbf{1}\{\pi^*(x_i)=a_i\}}{p(a_i|x_i)}}^2}\\
    &= \E_{x_i}\bigbrak{\E_{a_i}\bigbrak{\bigsmile{\frac{\mathbf{1}\{\hat{\pi}^k(x_i)=a_i\}}{p(a_i|x_i)}-\frac{\mathbf{1}\{\pi^*(x_i)=a_i\}}{p(a_i|x_i)}}^2}\Big|x_i}\\
    &= \E_{x_i}\bigbrak{\sum_{a}p(a|x_i)\bigsmile{\frac{\mathbf{1}\{\hat{\pi}^k(x_i)=a\}}{p(a|x_i)}-\frac{\mathbf{1}\{\pi^*(x_i)=a\}}{p(a|x_i)}}^2}\\
    &= \E_{x_i}\bigbrak{\sum_{a}\frac{1}{p(a|x_i)}\bigsmile{\mathbf{1}\{\hat{\pi}^k(x_i)=a\}-\mathbf{1}\{\pi^*(x_i)=a\}}^2}\\
    &= \E_{x_i}\bigbrak{\sum_{a}\frac{1}{p(a|x_i)}(\mathbf{1}\{\hat{\pi}^k(x_i)=a,\pi^*(x_i)\neq a\}+\mathbf{1}\{\hat{\pi}^k(x_i)\neq a, \pi^*(x_i)=a\})}\\
 %   &= \E_{x_i}\bigbrak{\bigsmile{\frac{1}{p(\hat{\pi}^k(x_i)|x_i)}+\frac{1}{p(\pi^*(x_i)|x_i)}}\mathbf{1}\{\hat{\pi}^k(x_i)\neq \pi^*(x_i)\}}\\
    &\leq \frac{2}{\eta}\E_{x_i}\bigbrak{\mathbf{1}\{\hat{\pi}^k(x_i)\neq \pi^*(x_i)\}}
\end{align*}
By Assumption~\ref{ass:consistency}, this expectation is $o(1)$, so the first term in equation~\eqref{eqn：third_term_decomp} is upper bounded by $\frac{2C_1^2}{\eta n_k}\E_{x_i}\bigbrak{\mathbf{1}\{\hat{\pi}^k(x_i)\neq \pi^*(x_i)\}}$ by bounded reward functions \ref{assump:bounded_r}, the second term is upper bounded by $\frac{2C}{\eta n_k}\E_{x_i}\bigbrak{\mathbf{1}\{\hat{\pi}^k(x_i)\neq \pi^*(x_i)\}}$ by bounded variance \ref{assump:finite_var}, so each term is $o(1)$. Then by Chebyshev inequality, a.k.a. Lemma~\ref{lem:sqrt_conv}, we have the statement. 
\end{proof}

\subsection{Efficiency of $\hat{\psi}^{AIPW}$}\label{sec:efficiency}
It is left to show that $\Var(\Gamma_1^*)$ is the minimum possible variance. To show this, we show that $\hat{\psi}^{AIPW}_*$ is an efficient estimator of $\psi$. In this section, we show efficiency in the nonparametric setting. Following the idea in \cite{luedtke2016statistical}, for any distribution $P$, 
\[\psi(P)=V(\pi_P^*) - \max_a \sum_{x}p(x) r(x,a)=\sum_{x} p(x) \left(\sum_{y} y p(y|a=\pi_P^*(x),x) - \sum_{y} y p(y|a=a_P^*,x)\right).\]

We first show that the $\hat{\psi}_*^{AIPW}$ is a gradient for $\psi(P)$. It is left to show that this is the canonical gradient and pathwise differentiability for $\psi(P)$, which together will give us efficiency. 
\begin{proposition}\label{prop:canon_grad}
If $\psi(P)$ is defined above, then a gradient is defined as 
\begin{align*}
D(P)(x,a,y) &= \mu(x,\pi_P^*(x)) - \mu(x,a_P^*) + \frac{\mathbf{1}\{\pi_P^*(x)=a,a_P^*\ne \pi_P^*\}}{p(a|x)}(y-\mu(x,\pi_P^*(x)))\\
&\qquad- \frac{\mathbf{1}\{a_P^*=a,a_P^*\ne \pi_P^*\}}{p(a|x)}(y-\mu(x,a_P^*)).
\end{align*}
% And an asymptotically efficient estimator is defined as \[\psi_n = \frac{1}{n}\sum_{i=1}^n \frac{\I\{a_i=\pi_P^*(x_i)\}}{p(a_i|x_i)}z_i+\bigsmile{1-\frac{\I\{a_i=\pi_P^*(x_i)\}}{p(a_i|x_i)}}s(a_i,x_i).\] 
\end{proposition}

\begin{proof}
Define $\{P_\epsilon:\epsilon\in \mathcal{N}\subset\R\}$ be a collection of quadratic mean differentiable submodels in $\M$ that passes through $P_0$ at $\epsilon=0$, where $\mathcal{N}$ is a neighborhood of zero. We consider $P_\epsilon=(1+\epsilon h)P$ where $h\in\mathcal{T}_\M(P)$ is some function in the tangent space of $\M$ at $P$. Given that we are in the nonparametric setting, $\mathcal{T}_\M(P)=L_2^0(P)$. Then $D(P_0)$ is a gradient of $\psi$ at $P_0$ if
\[\left.\frac{\partial}{\partial \epsilon} \psi\left(P_{\epsilon}\right)\right|_{\epsilon=0}=\E_{P_0}[D(P_0)(x,a,y) h(x,a,y)].\]
Since $L_2^0(P)=L_2^0(P_{X}) \otimes L_2^0(P_{A \mid X})\otimes L_2^0(P_{Y\mid A,X})$, we break the function $h$ in three components, i.e. 
%we assume $\{P_\epsilon : \epsilon\in \mathcal{N}\}$ has the following score at zero:
\begin{align*}
    h(x,a,y)=S_{Y|A,X}(y|a,x) + S_{A|X}(a|x) + S_{X}(x).
\end{align*}
Then the conditional densities satisfy 
\begin{align*}
    p_\epsilon(y|a,x)&= \{1+\epsilon S_{Y|A,X}(y|a,x)\} p(y|a,x),\\
    p_\epsilon(a|x)&= \{1+\epsilon S_{A|X}(a|x)\} p(a|x),\\
    p_\epsilon(x)&= \{1+\epsilon S_{X}(x)\} p(x).
\end{align*}
Then\footnote{We need to show stability of $\pi_P^*$ and $a_P^*$, i.e. the difference between $\pi_{P_\epsilon}^*$ and $\pi_{P_0}^*$ is negligible, and the same for $a^*$. For now we assume we have this, and same for $a^*$.}
\begin{align*}
    \psi(P_\epsilon) &= \sum_{x} p_\epsilon(x) \left(\sum_{y} y p_\epsilon(y|a=\pi_P^*(x),x) - \sum_{y} y p_\epsilon(y|a=a_P^*,x)\right) \\
    &= \sum_x(1+\epsilon S_X(x)) p(x)\Big(\sum_y y\cdot (1+\epsilon S_{Y|A,X}(y|a,x))p(y|a=\pi^*(x),x)\\
    &\qquad-\sum_y y\cdot (1+\epsilon S_{Y|A,X}(y|a,x))p(y|a=a^*,x)\Big)\\
    &= \sum_x(1+\epsilon S_X(x)) p(x)\Big(\sum_y y\frac{\1\{a=\pi^*(x)\}}{p(a|x)}\cdot (1+\epsilon S_{Y|A,X}(y|a,x))p(y|a,x)p(a|x)\\
    &\qquad-\sum_y y\frac{\1\{a=a^*\}}{p(a|x)}\cdot (1+\epsilon S_{Y|A,X}(y|a,x))p(y|a,x)p(a|x)\Big)\\
    &= \sum_x\Big(\sum_y y\frac{\1\{a=\pi^*(x)\}-\1\{a=a^*\}}{p(a|x)}\cdot (1+\epsilon S_{Y|A,X}(y|a,x))(1+\epsilon S_X(x)) p(y|a,x)p(a|x)p(x)\Big).
\end{align*}
Therefore, 
\begin{align*}
    \left.\frac{\partial\psi(P_\epsilon)}{\partial\epsilon}\right|_{\epsilon=0} &= \sum_{x,y}  y\Big(\frac{\1\{a=\pi^*(x)\}}{p(a|x)} - \frac{\1\{a=a^*\}}{p(a|x)}\Big) (S_{Y|A,X}(y|a,x)+S_X(x))p(y|a,x)p(a|x)p(x)\\
    &= \sum_{x,y} y\Big(\frac{\1\{a=\pi^*(x)\}}{p(a|x)} - \frac{\1\{a=a^*\}}{p(a|x)}\Big)(S_X(x)+S_{A|X}(a|x)+S_{Y|A,X}(y|a,x)) p(y|a,x)p(a|x)p(x)\\
    &\qquad-\sum_{x,y} y\Big(\frac{\1\{a=\pi^*(x)\}}{p(a|x)} - \frac{\1\{a=a^*\}}{p(a|x)}\Big) S_{A|X}(a|x)p(y|a,x)p(a|x)p(x).
\end{align*}
To simplify the second term, we let $g(x,a,y):=y\Big(\frac{\1\{a=\pi^*(x)\}}{p(a|x)} - \frac{\1\{a=a^*\}}{p(a|x)}\Big)$. Then the projection of $g$ onto the space $L_2^0(P_{A|X})$ as follows:
\begin{align*}
    \Pi_{L_2^0(P_{A|X})(P_0)}(g)(x,a) &= \E\bigbrak{\left.g(X,A,Y)\right|A=a,X=x}-\E\bigbrak{\left.g(X,A,Y)\right|X=x}\\
    &= \Big(\frac{\1\{a=\pi^*(x)\}}{p(a|x)} - \frac{\1\{a=a^*\}}{p(a|x)}\Big)\E[Y|A=a,X=x]-(\mu(\pi^*(x),x)-\mu(a^*,x))\\
    &= \frac{\1\{a=\pi^*(x),\pi^*(x)\ne a^*\} - \1\{a=a^*,\pi^*(x)\ne a^*\}}{p(a|x)} \mu(a,x)-(\mu(\pi^*(x),x)-\mu(a^*,x)).
\end{align*}
Note that since \[L_2^0(P)=L_2^0(P_{X}) \otimes L_2^0(P_{A \mid X})\otimes L_2^0(P_{Y\mid A,X}),\]
we have 
\begin{align*}
    &\sum_{x,a,y}\Pi_{L_2^0(P_{A|X})}g(x,a,y) S_X(x) p(x,a,y)=0;\\
    &\sum_{x,a,y}\Pi_{L_2^0(P_{A|X})}g(x,a,y) S_{Y|A,X}(y|a,x) p(x,a,y)=0.
\end{align*}
Therefore, adding them up gives 
\begin{align*}
    &\left.\frac{\partial\psi(P_\epsilon)}{\partial\epsilon}\right|_{\epsilon=0} \\
    &= \sum_{x,a,y} \left(\frac{\1\{a=\pi^*(x),\pi^*(x)\ne a^*\} - \1\{a=a^*,\pi^*(x)\ne a^*\}}{p(a|x)} (y-\mu(a,x)) + \mu(\pi^*(x),x)-\mu(a^*,x)\right)\\
    &\qquad(S_X(x)+S_{A|X}(a|x)+S_{Y|A,X}(y|a,x)) p(x,a,y).
\end{align*}
Therefore, $D(P)(x,a,y)$ is a gradient. 
% Since $P\in\M$ is nonparametric, $D(\pi_P^*,P)(O)$ is the canonical gradient. %It follows that $\psi_n$ is asymptotically efficient by Theorem 1.2. 
\end{proof}
\begin{proposition}
The projection is
\[\Pi_{L_2^0(P_{A|X})(P_0)}g: (x,a)\mapsto \E\bigbrak{\left.g(X,A,Y)\right|A=a,X=x}-\E\bigbrak{\left.g(X,A,Y)\right|X=x}.\]
\end{proposition}
\begin{proof}
To show that it is the projection, we need to show that it is in the tangent space and orthogonality. First, 
\begin{align*}
    \E_{A,Y}[\Pi_{L_2^0(P_{A|X})(P_0)}g | X=x]=0
\end{align*}
so $\Pi_{L_2^0(P_{A|X})(P_0)}g\in L_2^0(P_{A|X})$. Second, for any $h\in L_2^0(P_{Y|A,X})$, we have
\begin{align*}
    &\E_{X,A,Y}[\Pi_{L_2^0(P_{A|X})(P_0)}g h] \\
    &= \E_{X,A}[\E_{Y}[\Pi_{L_2^0(P_{A|X})(P_0)}g(X,A) h(X,A,Y)| A,X]] \\
    &= \E_{X,A}[\Pi_{L_2^0(P_{A|X})(P_0)}g(X,A)\E_{Y}[h(X,A,Y)| A,X]]\tag{$g$ is only a function of $X,A$}\\
    &= 0
\end{align*}
since $\E_{Y}[h(X,A,Y)| A,X]=0$ by definition of the tangent space $L^2(P_{Y|A,X})$. 
\end{proof}
Now we show a proposition about the pathwise differentiability of $\psi$ that ensures the above derived gradient is valid. 
\begin{proposition}\label{prop:path_diff}
Assume the following: 
\begin{enumerate}
    \item the optimal personalized policy $\pi_0^*$ and overall best action $a_0^*$ is unique, and further assume that there exists some $\epsilon_0>0$ such that for any $a\ne a_0^*$ and $\pi\neq \pi_0^*$, $V(\pi_0^*)-V(\pi)\geq \epsilon_0$ and $\E_0[\E_0[Y|A=a_0^*,X]]-\E_0[\E_0[Y|A=a,X]]\geq \epsilon_0$.\label{assump:minimum_gap}
\end{enumerate}
Then, $\psi$ is pathwise differentiable with gradient stated in Proposition~\ref{prop:canon_grad}. 
\end{proposition}
\begin{proof}
Recall that
\begin{align*}
    \psi(P)&=V(\pi_P^*) - \max_a \sum_{x}p(x) r(x,a)=\sum_{x} p(x) \left(\sum_{y} y p(y|a=\pi_P^*(x),x) - \sum_{y} y p(y|a=a_P^*,x)\right)\\
    &= \sum_{x} p(x) \left(\sum_{y} y\1\{a=\pi_P^*(x)\} p(y|a,x) - \sum_{y} y\1\{a=a_P^*\} p(y|a,x)\right).
\end{align*}
Then 
\begin{align*}
    \psi(P_\epsilon)-\psi(P_0) &= \psi(P_\epsilon)-\Psi_{\pi_0^*,a_0^*}(P_\epsilon)+\Psi_{\pi_0^*,a_0^*}(P_\epsilon)-\Psi(P_0)\\
    &= \sum_{x} p_\epsilon(x) \left(\sum_{y} y(\1\{a=\pi_{P_\epsilon}^*(x)\} -\1\{a=a_{P_\epsilon}^*\}) p_\epsilon(y|a,x)\right)\\
    &\quad- \sum_{x} p_\epsilon(x) \left(\sum_{y} y(\1\{a=\pi_0^*(x)\} -\1\{a=a_0^*\}) p_\epsilon(y|a,x)\right)+\Psi_{\pi_0^*,a_0^*}(P_\epsilon)-\Psi(P_0)\\
    &= \sum_{x} p_\epsilon(x) \left(\sum_{y} y(\1\{a=\pi_{P_\epsilon}^*(x)\}-\1\{a=\pi_0^*(x)\}-(\1\{a=a_{P_\epsilon}^*\}-\1\{a=a_0^*\})) p_\epsilon(y|a,x)\right)\\
    &\quad+\Psi_{\pi_0^*,a_0^*}(P_\epsilon)-\Psi(P_0)
\end{align*}
It is known that for a fixed $\pi_0^*$ and $a_0^*$, $\psi(P)$ is pathwise differentiable with gradient $D(\pi,P_0)$. We shall now show that the first term is $o(\epsilon)$. We first decompose the first term as 
\begin{equation}
    \sum_{x,y} y(\1\{a=\pi_{P_\epsilon}^*(x)\}-\1\{a=\pi_0^*(x)\}) p_\epsilon(y|a,x)p_\epsilon(x)-\sum_{x,y} y(\1\{a=a_{P_\epsilon}^*\}-\1\{a=a_0^*\})p_\epsilon(y|a,x)p_\epsilon(x).\label{eqn:path_diff_decomp}
\end{equation}
We would like to show each term is $o(\epsilon)$. For any $a$, 
\begin{align*}
    \E_\epsilon[\E_\epsilon\bigbrak{Y|A=a,X}] - \E_0[\E_0\bigbrak{Y|A=a,X}] &= \int y d(P_\epsilon(y|a,x)P_\epsilon(x)-P_0(y|a,x)P_0(x)).
\end{align*}
If $p_\epsilon(y|a,x)=(1+\epsilon S_Y(y|a,x))p_0(y|a,x)$ and $p_\epsilon(x)=(1+\epsilon S_X(x))p_0(x)$, then $p_\epsilon(y|a,x)p_\epsilon(x)-p_0(y|a,x)p_0(x)=\epsilon (S_Y(y|a,x)+S_0(x))p_0(y|a,x)p_0(x)+o(\epsilon)$, so 
\begin{align*}
    \E_\epsilon[\E_\epsilon\bigbrak{Y|A=a,X}] - \E_0[\E_0\bigbrak{Y|A=a,X}] &= \epsilon\int y (S_Y(y|a,x)+S_X(x))p_0(y|a,x)p_0(x) + o(\epsilon) \leq C\epsilon
\end{align*}
for some constant $C$ given $Y$, $S_Y$, and $S_X$ are bounded. Note that $\pi_P^* = \arg\max_{\pi\in\Pi} \E_P[\E_P\bigbrak{Y|A=\pi_P(X),X}]$ and $a_P^* = \arg\max_{a\in\A} \E_P[\E_P\bigbrak{Y|A=a,X}]$, if Assumption~\ref{assump:minimum_gap} is true, consider any $\epsilon\leq \frac{\epsilon_0}{4C}$, then for any $a\ne a_0^*$, we have 
\begin{align*}
    &\E_\epsilon\E_\epsilon\bigbrak{Y|A=a_0^*,X} - \E_\epsilon\E_\epsilon\bigbrak{Y|A=a,X}\\
    &= \E_\epsilon\E_\epsilon\bigbrak{Y|A=a_0^*,X}-\E_0\E_0\bigbrak{Y|A=a_0^*,X}+\E_0\E_0\bigbrak{Y|A=a_0^*,X}- \E_0\E_0\bigbrak{Y|A=a,X}\\
    &\quad+\E_0\E_0\bigbrak{Y|A=a,X} - \E_\epsilon\E_\epsilon\bigbrak{Y|A=a,X}\\
    &\geq \epsilon_0 - 2\epsilon>0,
\end{align*}
so $a_{P_\epsilon}^*=a_0^*$, and so the second term in equation~\eqref{eqn:path_diff_decomp} is 0. Similarly, for any $\pi$, 
\begin{align*}
    \E_\epsilon[\E_\epsilon\bigbrak{Y|A=\pi(X),X}] - \E_0[\E_0\bigbrak{Y|A=\pi(X),X}] &= \int y\1\{a=\pi(x)\} d(P_\epsilon(y|a,x)P_\epsilon(x)-P_0(y|a,x)P_0(x)).
\end{align*}
By similar reasoning, 
\begin{align*}
    \E_\epsilon[\E_\epsilon\bigbrak{Y|A=\pi(X),X}] - \E_0[\E_0\bigbrak{Y|A=\pi(X),X}] &= \epsilon\int y\1\{a=\pi(x)\} (S_Y(y|a,x)+S_X(x))p_0(y|a,x) + o(\epsilon) \leq C\epsilon
\end{align*}
for some constant $C$ since the indicator is bounded by 1. Then, for any $\pi\neq \pi_0^*$, 
\begin{align*}
    &\E_\epsilon\E_\epsilon\bigbrak{Y|A=\pi_0^*(X),X} - \E_\epsilon\E_\epsilon\bigbrak{Y|A=\pi(X),X}\\
    &= \E_\epsilon\E_\epsilon\bigbrak{Y|A=\pi_0^*(X),X}-\E_0\E_0\bigbrak{Y|A=\pi_0^*(X),X}+\E_0\E_0\bigbrak{Y|A=\pi_0^*(X),X}- \E_0\E_0\bigbrak{Y|A=\pi(X),X}\\
    &\quad+\E_0\E_0\bigbrak{Y|A=\pi(X),X} - \E_\epsilon\E_\epsilon\bigbrak{Y|A=\pi(X),X}\\
    &\geq \epsilon_0 - 2\epsilon>0,
\end{align*}
the first term is also 0, i.e. $\pi_0^*=\pi_{P_\epsilon}^*$. Therefore, $\psi(P_\epsilon)-\psi(P_0)=\psi_{\pi_0^*,a_0^*}(P_\epsilon)-\psi(P_0)$ for $\epsilon\leq \frac{\epsilon_0}{4C}$ and we get pathwise differentiability. 
\end{proof}

\subsection{Analysis of Repeated Re-ordering and Splits}
% Zhaoqi: need to take another look here.
Algorithm~\ref{alg:kpt} that repeatedly performs Algorithm~\ref{alg:single_kpt} to produce a more robust estimate. We now show that this procedure still produces an efficient estimator and thus an efficient test. This procedure is similar to the procedure described in Section 3.4 of \cite{chernozhukov2018double}, where they provide Corollary 3.3 that this approach still ensures their prior theorems hold. 

\begin{corollary}%{Corollary 3.18}}. 
Under assumptions 1--8 and 10--11, the estimator $\tilde{\psi}$ in Algorithm 1  %~\ref{alg:kpt} 
is efficient for $\psi$. 
\end{corollary}
\begin{proof}
From Lemma~\ref{lem:oracle_convergence}, the estimated personalization effect for a particular split $s$ satisfies
\begin{equation}
\sqrt{n}(\hat{\psi}_{s}-\hat{\psi}_*) =o_P(1).
\end{equation}
Then, 
\begin{align*}
    \sqrt{n}(\tilde{\psi}-\psi)&=\frac{1}{S} \sum_{s=1}^S \sqrt{n}(\widehat{\psi}_s-\psi)\\
    &= \frac{1}{S} \sum_{s=1}^S \left(\sqrt{n}(\widehat{\psi}_s-\widehat{\psi}_*)+\sqrt{n}(\widehat{\psi}_*-\psi)\right)\\
    &= o_P(1)+\sqrt{n}(\widehat{\psi}_*-\psi)\Rightarrow N(0,\Var(\Gamma_1^*))
\end{align*}
where the last line follows from Proposition~\ref{prop:asymp_normal_oracle} and Slutsky's Theorem. 

We next show that the multi-split variance estimator $\tilde{\sigma}^{2}=\frac{1}{S}\sum_{s=1}^S\left(\hat{\sigma}_s^2+(\hat{\psi}_s-\hat{\psi})^2\right)$ is consistent for the true unknown variance $\sigma^2$. 
By Lemma~\ref{lem:oracle_convergence} and Proposition~\ref{prop:asymp_normal_oracle}, 
we know the variance for the estimated personalization effect from a single split $\hat{\psi}^s$ converges to the true variance $\sigma^2$. Earlier in this proof we showed that for fixed $S$, 
$||\tilde{\psi} - \hat{\psi}^s|| = o_p(n^{-1/2})$. 
This implies that (as $S$ is fixed) the additional across-split term in $\tilde{\sigma}$ is $o_p(n^{-1})$ and therefore $\tilde{\sigma}$ is consistent for $\sigma^2$. 

Therefore the multi-split estimate $\tilde{\psi}$ is asymptotically equivalent to a single split estimate, and has the same normal limit and variance.

\end{proof}
\subsection{Technical lemmas}
In this section, we present a few technical lemmas used to show convergence. 
\begin{proposition}\label{prop:bound_hat_psi} Let $\hat{\sigma}^2$ be defined in Algorithm~\ref{alg:single_kpt} and $\hat{\sigma}_k^2$ denote the variance for the $k$th fold, then
\[\hat{\sigma}^2\geq \frac{1}{n} \sum_{k=1}^K n_k \hat{\sigma}^2_k.\]
\end{proposition}
\begin{proof}
\begin{eqnarray}
\bar{\Psi}_k & \equiv& \frac{1}{n_k} \sum_{j=1}^{n_k} h_{kj} \\
\bar{\Psi} & \equiv & \frac{1}{n} \sum_{i=1}^{n} h_{i} \\
\hat{\sigma}^2_k & \equiv & \frac{1}{n_k} \sum_{j=1}^{n_k} (h_{jk} - \bar{\Psi}_k)^2 \\
\hat{\sigma}^2 & \equiv & \frac{1}{n} \sum_{i=1}^{n} (h_{i} - \bar{\Psi})^2 = \frac{1}{n} \sum_{k=1}^K \sum_{j=1}^{n_k} (h_{kj} - \bar{\Psi})^2 = \frac{1}{n} \sum_{k=1}^K \sum_{j=1}^{n_k} (h_{kj} - \bar{\Psi}_k + \bar{\Psi}_k - \bar{\Psi})^2 \\
&=& \frac{1}{n} \sum_{k=1}^K \sum_{j=1}^{n_k} (h_{kj} - \bar{\Psi}_k)^2  + (\bar{\Psi}_k - \bar{\Psi})^2 + 2 (h_{kj} - \bar{\Psi}_k)(\bar{\Psi}_k - \bar{\Psi})\\
&=& \frac{1}{n} \sum_{k=1}^K \sum_{j=1}^{n_k} (h_{kj} - \bar{\Psi}_k)^2  + \frac{1}{n} \sum_{k=1}^K n_k (\bar{\Psi}_k - \bar{\Psi})^2 + \frac{1}{n} \sum_{k=1}^K \sum_{j=1}^{n_k} 2 (h_{kj} - \bar{\Psi}_k)(\bar{\Psi}_k - \bar{\Psi}) \nonumber \\
&=& \frac{1}{n} \sum_{k=1}^K \sum_{j=1}^{n_k} (h_{kj} - \bar{\Psi}_k)^2  + \frac{1}{n} \sum_{k=1}^K n_k (\bar{\Psi}_k - \bar{\Psi})^2 + \frac{1}{n} \sum_{k=1}^K 2 (\bar{\Psi}_k - \bar{\Psi}) \sum_{j=1}^{n_k} (h_{kj} - \bar{\Psi}_k) \nonumber \\
&=& \frac{1}{n} \sum_{k=1}^K \sum_{j=1}^{n_k} (h_{kj} - \bar{\Psi}_k)^2  + \frac{1}{n} \sum_{k=1}^K n_k (\bar{\Psi}_k - \bar{\Psi})^2 \label{eqn:mean_cancel}\\
&=& \frac{1}{n} \sum_{k=1}^K n_k \frac{1}{n_k} \sum_{j=1}^{n_k} (h_{kj} - \bar{\Psi}_k)^2  + \frac{1}{n} \sum_{k=1}^K n_k (\bar{\Psi}_k - \bar{\Psi})^2 \\
&=& \frac{1}{n} \sum_{k=1}^K n_k \hat{\sigma}^2_k  + \frac{1}{n} \sum_{k=1}^K n_k (\bar{\Psi}_k - \bar{\Psi})^2 \\
&\geq &  \frac{1}{n} \sum_{k=1}^K n_k \hat{\sigma}^2_k
\end{eqnarray}
where Equation~\ref{eqn:mean_cancel} holds because for each fold, $n_k \bar{\Psi}_k = \sum_{j=1}^{n_k} h_{kj}.$ 

\end{proof}

\begin{lemma}\label{lem:sqrt_conv}
For some random variable $S_n$, if $\E[S_n]=o(n^{-1/2})$ and $\Var(S_n)=o(n^{-1})$, then $\sqrt{n}S_n\overset{p}{\to}0$. 
\end{lemma}
\begin{proof}
For some $\epsilon>0$, since $\{\sqrt{n}S_n>\epsilon\} \subseteq \{\sqrt{n}(S_n-\E[S_n])>\frac{\epsilon}{2}\}\cup \{\sqrt{n}\E[S_n]>\epsilon/2\}$, we have 
\begin{align*}
    \P(\sqrt{n}S_n>\epsilon)&\leq \P(\sqrt{n}(S_n-\E[S_n])>\epsilon/2)+\P(\sqrt{n}\E[S_n]>\epsilon/2).
\end{align*}
By Chebyshev inequality,
\[\P(\sqrt{n}(S_n-\E[S_n])>\epsilon/2)\leq \frac{4n\Var(S_n)}{\epsilon^2}\to 0,\]
and $\P(\sqrt{n}\E[S_n]>\epsilon/2)\to 0$ since $\E[S_n]=o(n^{-1/2})$. Therefore we have $\sqrt{n}S_n\overset{p}{\to}0$. 
\end{proof}

\begin{algorithm}[!htb]
\caption{PAPD-Norm-for-Personalization-Baseline}
\begin{algorithmic}[1]\label{alg:papd_norm_personalization}
\REQUIRE Number of folds $K$, Dataset $D$
\STATE Partition $D$ into $D_1,D_2,\cdots,D_K$ folds.
\FOR{$k = 1,2,\cdots,K$}
    \STATE $\hat{\pi}^k \leftarrow$ PersonalizedPolicyLearning$(D_{-k})$ 
    \STATE $\hat{a}^k \leftarrow $ BestInterventionLearning$(D_{-k})$  
%    \STATE $\mu_y  =\frac{1}{|\mathcal{D}_k|}\sum_{i \in \mathcal{D}_k} y_i$
    \STATE Compute PAPD estimate, centering $y$
    \begin{equation}
\hat{\psi}_{PAPD}^{k} =\frac{1}{|\mathcal{D}_k|}\sum_{i \in \mathcal{D}_k}  \frac{(\mathbf{1}\{\hat{\pi}^k(x_i)=a_i\}-\mathbf{1}\{\hat{a}^{k}=a_i\})}{p(a_i|x_i)}(y_i-\mu_y)
\end{equation}
\ENDFOR
\STATE Compute $\hat{\psi}_{PAPD}=\frac{1}{K}\sum_{k=1}^K \hat{\psi}^{k}$, variance $\hat{\sigma}_{PAPD}^2:=\Var(\hat{\psi})$
\ENSURE $\hat{\psi}_{PAPD}$, $\hat{\sigma}_{PAPD}^2$
\end{algorithmic}
\end{algorithm}

% EB: I don't think we can have this here-- it is post the Science revision
%\input{two_fold_polished}

\section{Assessing Reliability, Type I and Power Using Simulations}
\label{sec:simulation_reliability}

In real or semi-synthetic datasets, the true personalization effect is unknown. Therefore, to assess our estimator's performance in a setting where we can precisely calculate the true personalization effect,  we conducted some experiments in fully simulated settings. 

We consider two simulations, both with binary treatments (2 interventions, a treatment and control), and 20 discrete contexts. Contexts are sampled from a uniform distribution, and actions are also sampled uniformly at random. We now define the outcome model and personalization effects for the two separate simulations.
In all cases reward outcomes $r(x,a)$ are generated with additive Gaussian noise (\(\sigma=0.2\)).
\begin{enumerate}
\item Threshold-Large (Figure~\ref{fig:first}). In this simulation context $x_0$ has a treatment effect of (\(r(x_0, 1) - r(x_0, 0) = \Delta\)), and in all other contexts $x \neq x_0$,  the treatment effect is (\(r(x, 1) - r(x, 0) = -\alpha \Delta\)) (implying the control yields a higher outcome than treatment) where $\alpha=0.1$. $\Delta$ is a  constant that we will vary in our experiments. As there is a uniform probability over contexts, the value of the best personalized policy, providing all with treatment, and providing all with control is:
\begin{itemize}
\item Expected utility of providing all with control: $\frac{19}{20}\alpha \Delta = \frac{19}{200}\Delta$
\item Expected utility of providing all with treatment: $\frac{1}{20} \Delta$
\item Expected utility of providing each context with their best treatment/control: $\frac{10}{200}\Delta +\frac{19}{200}\Delta$
\end{itemize}
The personalization effect is $0.05\Delta$. In this scenario we have one context which highly benefits from personalization, and many contexts for which both actions yield similar, but not identical,  reward outcomes. 
\item Threshold-small (Figure~\ref{fig:second}). This case is designed to consider when there are a large number of contexts that benefit from personalization, but with a smaller magnitude. Here 30\% of the contexts (6 contexts) have \(r(x, 1) - r(x, 0) = \Delta\) and the remaining 70\% of contexts have \(r(x, 1) - r(x, 0) = -\alpha\Delta\) with $\alpha=0.5$.  Given that there is a uniform probability over contexts, the value of the best personalized policy, providing all with treatment, and providing all with control is:
\begin{itemize}
\item Expected utility of providing all with control: $\frac{7}{10}\alpha \Delta = \frac{7}{20}\Delta$
\item Expected utility of providing all with treatment: $\frac{3}{10} \Delta$
\item Expected utility of providing each context with their best treatment/control: $\frac{3}{10}\Delta +\frac{7}{20}\Delta$
\end{itemize}
The best single intervention is the control intervention. The personalization effect is $0.3 \Delta$.

%We also consider large number of features with a smaller reward gap. We assume a set of features $\X^{treat}$ favoring treatment and another set of features $\X^{control}$ favoring control. In our simulation, $\X^{treat}$ contains 30\% of the features  and $\X^{control}$ contains the remaining 70\% of the features. For $x\in\X^{treat}$, the reward gap is \(r(x, 1) - r(x, 0) = \Delta\), and for the rest $x'\in\X^{control}$, the reward gap \(r(x', 0) - r(x', 1) = \alpha\Delta\) with $\alpha=0.5$. Given that there is a uniform probability over contexts, the best single intervention is to provide the control action, which yields an expected utility of $0.12\Delta$ over providing all individuals with treatment. The personalized policy can do even better, and the true personalization effects is $0.72\Delta$, so the true utility of personalization is $0.6\Delta$. We choose $\alpha$ in this way so the overall best choice is still giving control to everyone, but there is a personalization effect for giving treatment to $\X^{treat}$. We call this \textsf{threshold-small} instance given there are multiple features with small gaps favoring treatment. 

\end{enumerate}

We consider these two cases because we think they reflect interesting different realistic settings. For example, a precision oncology drug may greatly help only mutation‐positive patients but harm the rest over a control, whereas a broad wellness program might yield modest gains for most participants yet slightly worsen outcomes for a few. Studying both extremes examines how well our test identifies personalization when benefits are either concentrated or widely distributed.

%We create a binary treatment simulation with 20 discrete contexts and outcomes that range from $(0,1)$, we vary the true personalization effect $\psi$ from 0 to 0.05. Rewards are generated with Gaussian noise (\(\sigma=0.2\)).

\begin{figure}[htbp]
  \centering
  % First sub‑figure
  \begin{minipage}[b]{0.48\textwidth}
    \centering
    \includegraphics[width=\textwidth]{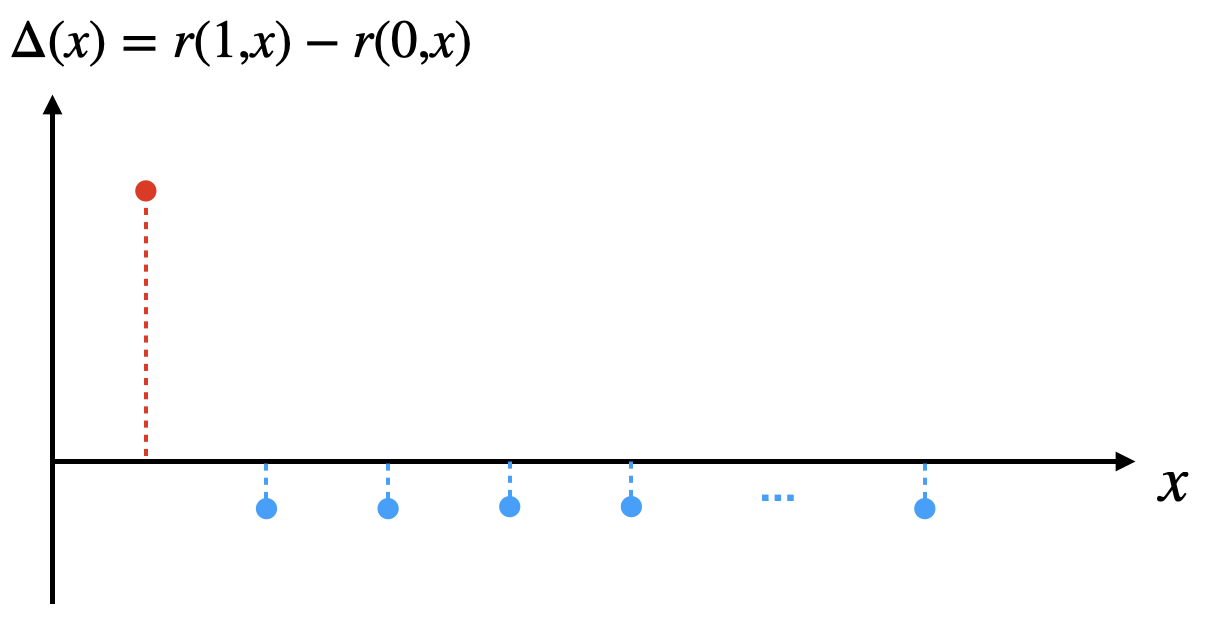}%threshold_large.png}
    \caption{Threshold-large instance}
    \label{fig:first}
  \end{minipage}
  \quad % horizontal separation
  % Second sub‑figure
  \begin{minipage}[b]{0.48\textwidth}
    \centering
    \includegraphics[width=\textwidth]{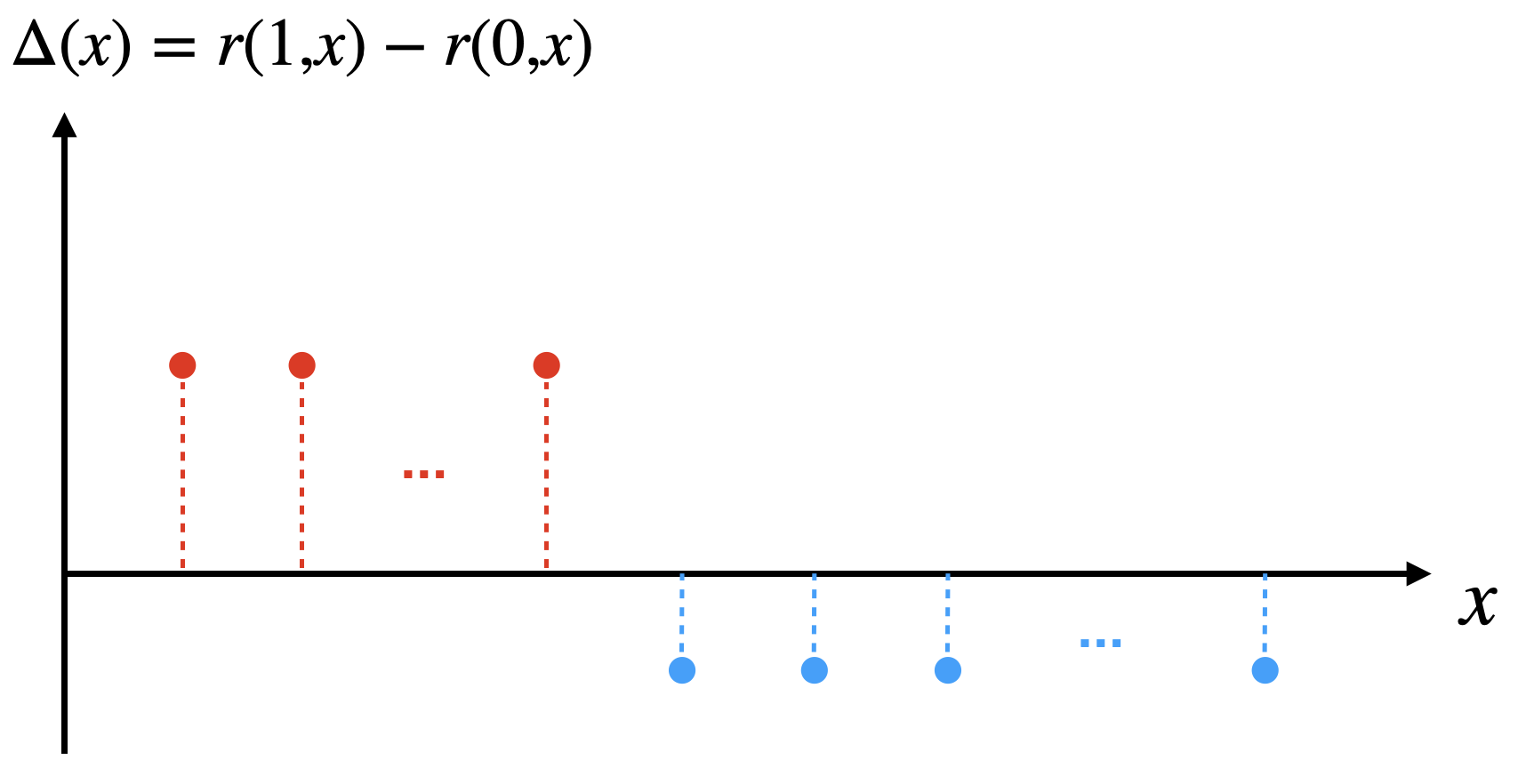}%threshold_small.png}
    \caption{Threshold-small instance}
    \label{fig:second}
  \end{minipage}
\end{figure}

We consider a range of $\Delta$ values, which allow us to effectively vary the size of the personalization effect. For each instance (fixed $\Delta$),  we sample 100 datasets of size $N=3000$, and report a $p$ value at the end using Algorithm~\ref{alg:kpt}. We use a random forest regressor for reward learning, a doubly robust policy forest learner from the econml package for policy learning, and learn the single-best action by selecting the action with the maximum reward estimate (aka we use a random forest regressor to build a reward model only on the data used for best action/policy learning, and then take the action with the maximum expected value over the contexts in this set).  We compare our results with the baselines TrainEval, SRP~\cite{shi2020sparse}\footnote{Note because the context is a scalar/single dimension, we use their test.QTI function since the sparse projection aspect is not needed.}, and PAPD~\cite{imai2023experimental}. We adapt PAPD to be estimate the personalization effect, by using the best intervention learner as one policy learner. We also center the outcome variable to reduce estimator variance, as suggested in the authors' empirical section. Algorithm~\ref{alg:papd_norm_personalization} shows the pseudocode for PAPD for the personalization context. We use 3 folds for PAPD. We choose this to enable PAPD to benefit from multiple folds and using more data for policy learning. As we discuss in the main text, it is an interesting direction for future work to explore how best to tune the size of the partitions used for policy learning, effect estimation, and (in the case of our KPT) the reward and propensity modeling. 

%Figure~\ref{fig:s3_new} and \ref{fig:test_efficient_large} display the $p$-values, the percent of times the test returns significance among 100 datasets (i.e. significance rate), and the percentage of contexts for which the estimated personalized policy matches the true best policy (averaged over the 100 datasets, splits and folds) over true personalization effect. %This last part is informative because it helps to highlight that our method can often reliably detect a personalization effect even before the best personalized policy is identified for all contexts (for example, when the personalization effect is $0.03$ in the small threshold domain, 80\% of contexts ..***. 

Figure~\ref{fig:s3_new} and \ref{fig:test_efficient_large} show the $p$-values versus the true personalization effect, and the percent of times the test returns significance among 100 datasets (i.e. significance rate) vs the true personalization effect. Note that the baselines-- SRP, PAPD and TrainEval-- only use a single split of the data, and so we report results for a single split of the data. It is possible results for these methods could change slightly under a different seed (we show their variance to the seed in the main text).

Both SRP and our KPT have $\leq 5\% $  type I error, adequately rejecting the test of evidence of personalization benefit when none is present, but TrainEval and PAPD falsely reject the null hypothesis in over 10\% of the datasets when the personalization effect is zero. Our K-fold Personalization Test performs similarly compared to SRP and PAPD when there is a non-zero personalization effect, though SRP slightly outperforms (has higher detection rates) for two  personalization effect sizes, likely due to its specialized policy learning procedure for binary treatments, which could be incorporated into KPT. Overall these results help demonstrate the reliability and efficiency of KPT. 

Our method relies on learning a number of nuisance parameters, including learning the personalized policy, which is a challenging and active area of research.   We define policy accuracy as the proportion of contexts where the learned policy matches the optimal personalized policy for that context, averaged over the datasets, splits and folds: $\frac{1}{SKN_D}\sum_{N_D} \sum_{S}\sum_{K} \frac{1}{|D_{dsk}|} \sum_{D_{dsk}} I(\hat{\pi}_k(x_i) = \pi^*(x_i))$, where $N_D$ is the number of datasets (here, 3000). Results are shown in the third column of Figure~\ref{fig:s3_new} and Figure~\ref{fig:test_efficient_large}. Although our theoretical efficiency results require fast policy learning, these results show our K-fold Personalization Test can still work well even in the regime when policy learning is not perfect. For example, Figure~\ref{fig:s3_new} shows at $\psi=0.03$, policy learning accuracy 
is under $90\%$ but KPT still detects a personalized effect in almost all of the datasets. PAPD has a slightly higher policy learning accuracy than KPT, but this is likely due to the fact that they use more data for policy learning, yet KPT still has a similar performance in terms of $p$-values and significance rates. 

% shows how often the best overall intervention changes across folds and splits, and what percent of the contexts have at least two different personalized policy interventions assigned to them across folds and splits, in our synthetic example. 

%%%%%%%%%%%%%%%%%%%%%%%%%
% practical efficiency (though don't always expect to hold) 
% sensititivity to policy learning? 
% show if learn worse reward / 
% look at MSE of propensity and reward models? and still do quite well?
% 
% did our CI cover as desired? 
% reliability
% stability? 
% type 1 error
% could look at for these
% discuss policy learnining in conclusion. once have it, what to do
% stability 
% real-world

\begin{figure}[!htb]
    \centering
    \includegraphics[width=0.8\linewidth]{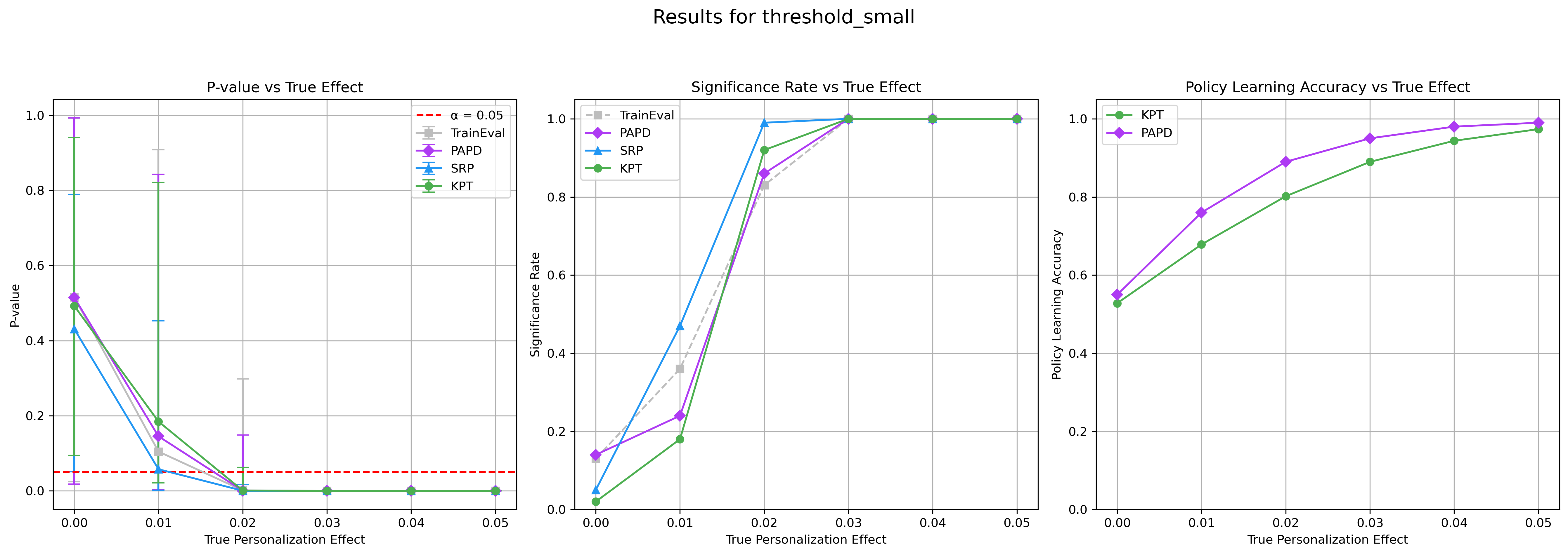}
    \caption{$p$-value, significance rate, and policy accuracy results for Threshold-small instance. Results averaged over 100 sampled datasets, each of size $N=3000$. Our K-Fold Personalization Test uses $S=100$. Intervals in the ``P-value vs True effect" (left-most subplot) are shown with 5\% and 95\% empirical percentiles. Our test is efficient and offers a theoretical benefit. In addition, when the treatment effect is zero, TrainEval often falsely reject the null hypothesis but our approach does not.}
    \label{fig:s3_new}
\end{figure}

\begin{figure}[!htb]
    \centering
\includegraphics[width=0.8\linewidth]{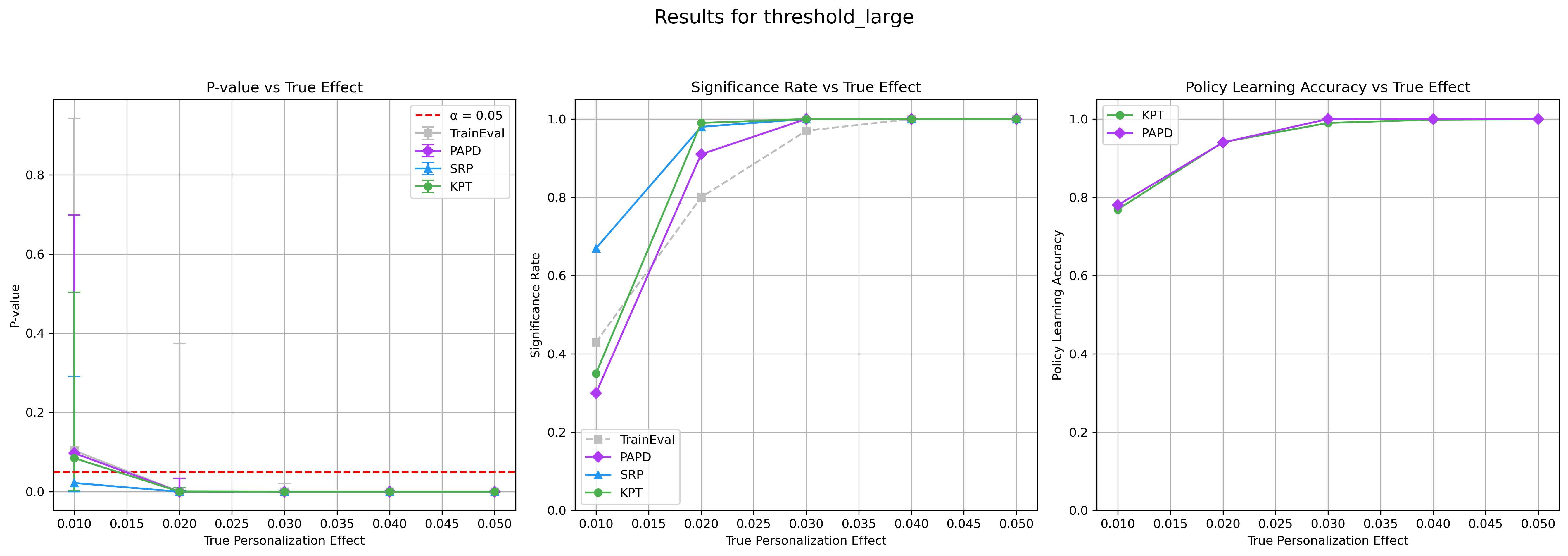}
\caption{$p$-value, significance rate, and policy accuracy results for Threshold-large instance.  Results averaged over 100 sampled datasets, each of size $N=3000$. Intervals in the ``P-value vs True effect" (left-most subplot) are shown with 5\% and 95\% empirical percentiles. Our K-Fold Personalization Test uses $S=100$.}
    \label{fig:test_efficient_large}
\end{figure}

\paragraph{CI coverage.} We plot the coverage of confidence interval, that is, the percentage of confidence intervals covering the true value over 100 different datasets, for our approach KPT on both simulation instances. Here we do not have confidence intervals for the CI coverage since KPT only outputs one confidence interval per dataset. We show the results in Figure~\ref{fig:CI_coverage_sim}. We can see that for both Threshold-small and Threshold-large instances our approach achieves valid coverages over all sizes of true effect ranging from 0 to 0.1. 
\begin{figure}[!htb]
    \centering
    \includegraphics[width=0.5\linewidth]{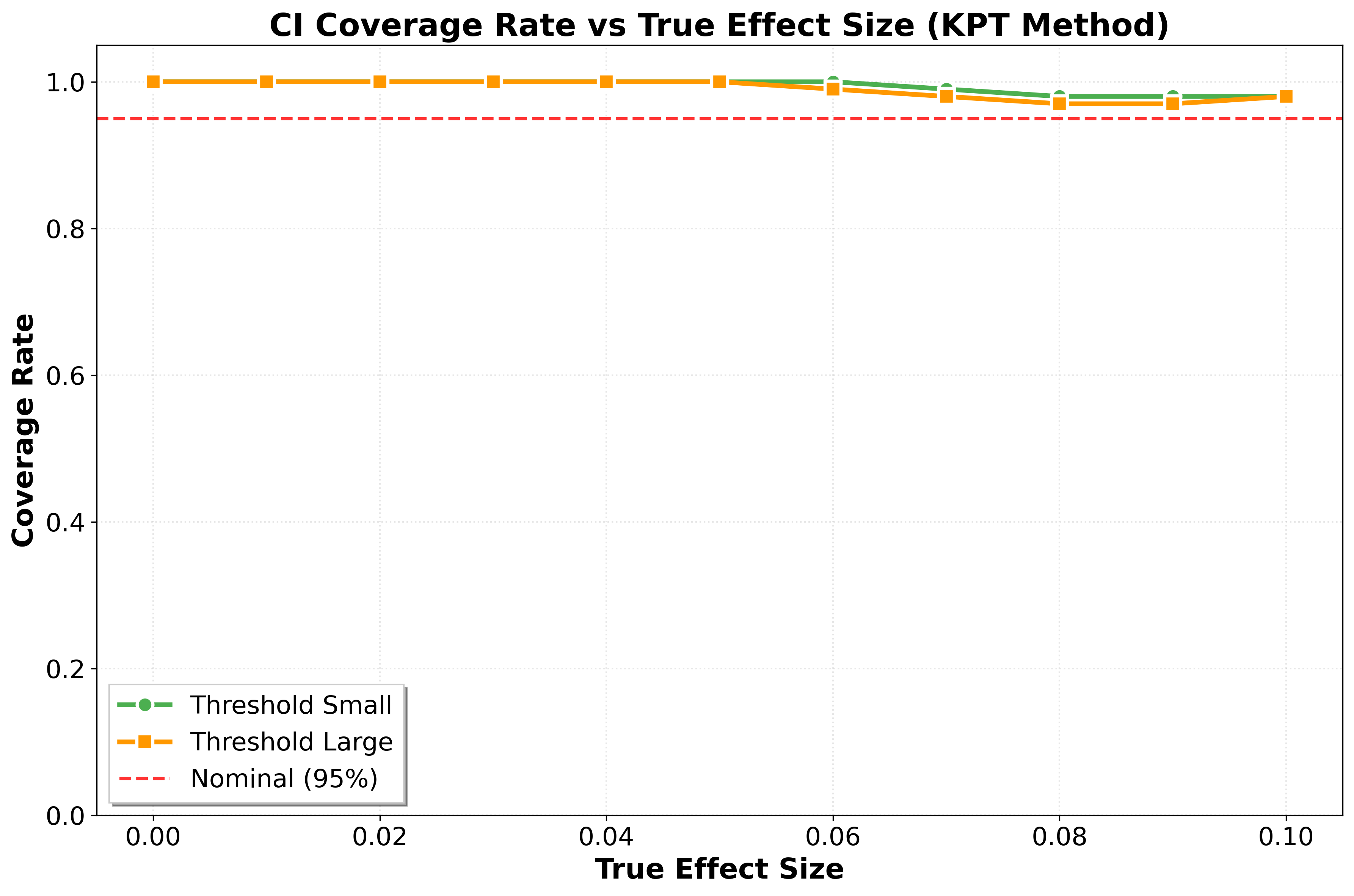}
    \caption{CI coverage for simulation instances}
    \label{fig:CI_coverage_sim}
\end{figure}

\paragraph{Policy stability.}
Here we compute policy stability as the percentage of the contexts $X$ whose learned optimal policies is always the same, across repeated splits and folds: $1- \frac{1}{|X|}\sum_x I(\exists s,s' \; \exists k,k' \hat{\pi}_{sk}(x) \neq \hat{\pi}_{s'k'}(x))$. We also compute the single-best action stability, which is the percentage of runs that it returns the same single best intervention, where for each run, the algorithm returns an overall-best action via the majority vote. 

Figure~\ref{fig:policy_stability_sim} shows the results. Here the error bar indicates the standard deviation across different runs. We can see that as true effect sizes increases, both policy learning and single-best action learning are more stable. Unsurprisingly, in many cases it is much easier to learn the best-arm than the personalized policy.
\begin{figure}[!htb]
    \centering
    \includegraphics[width=0.48\linewidth]{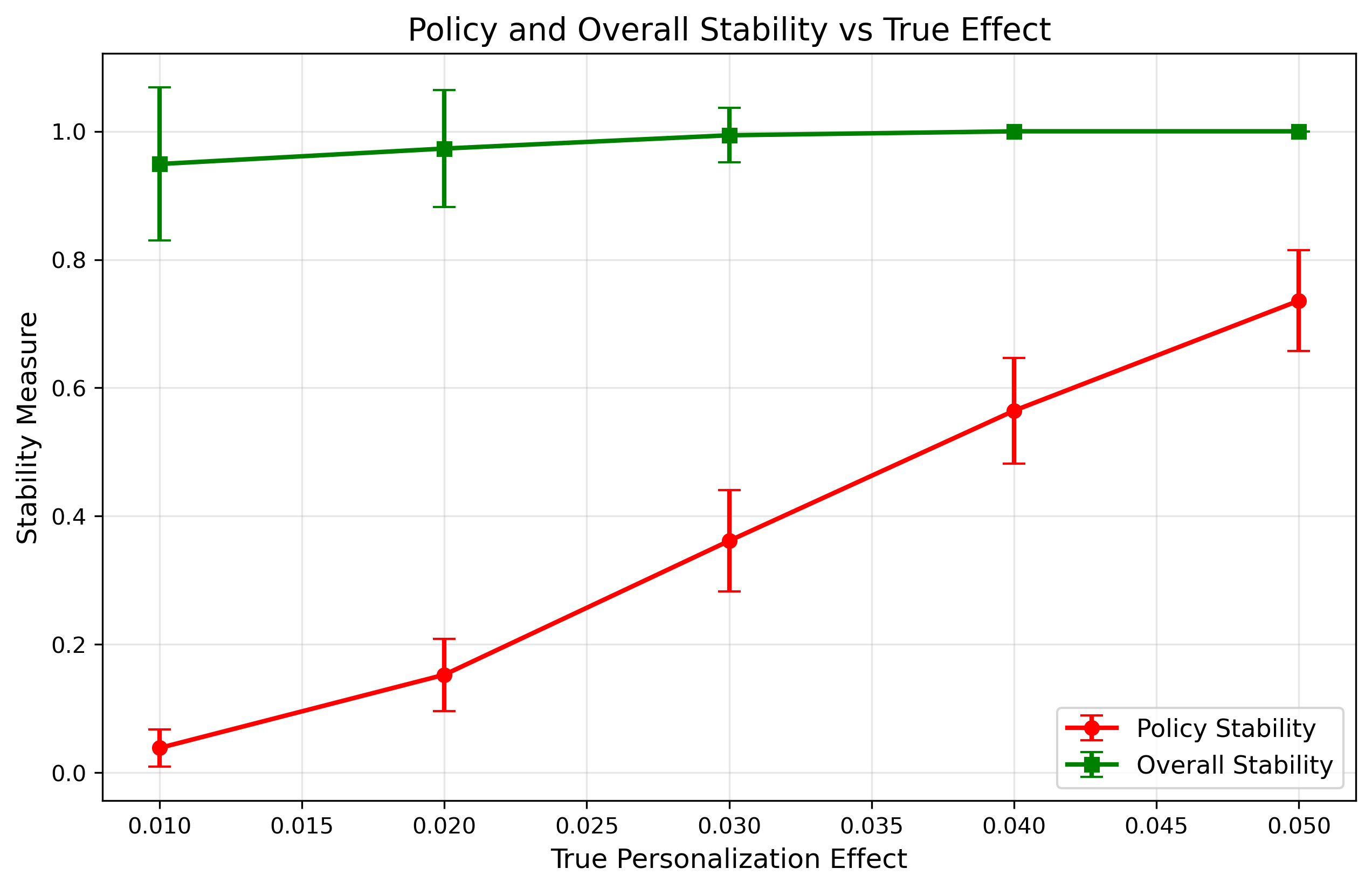}
    \includegraphics[width=0.48\linewidth]{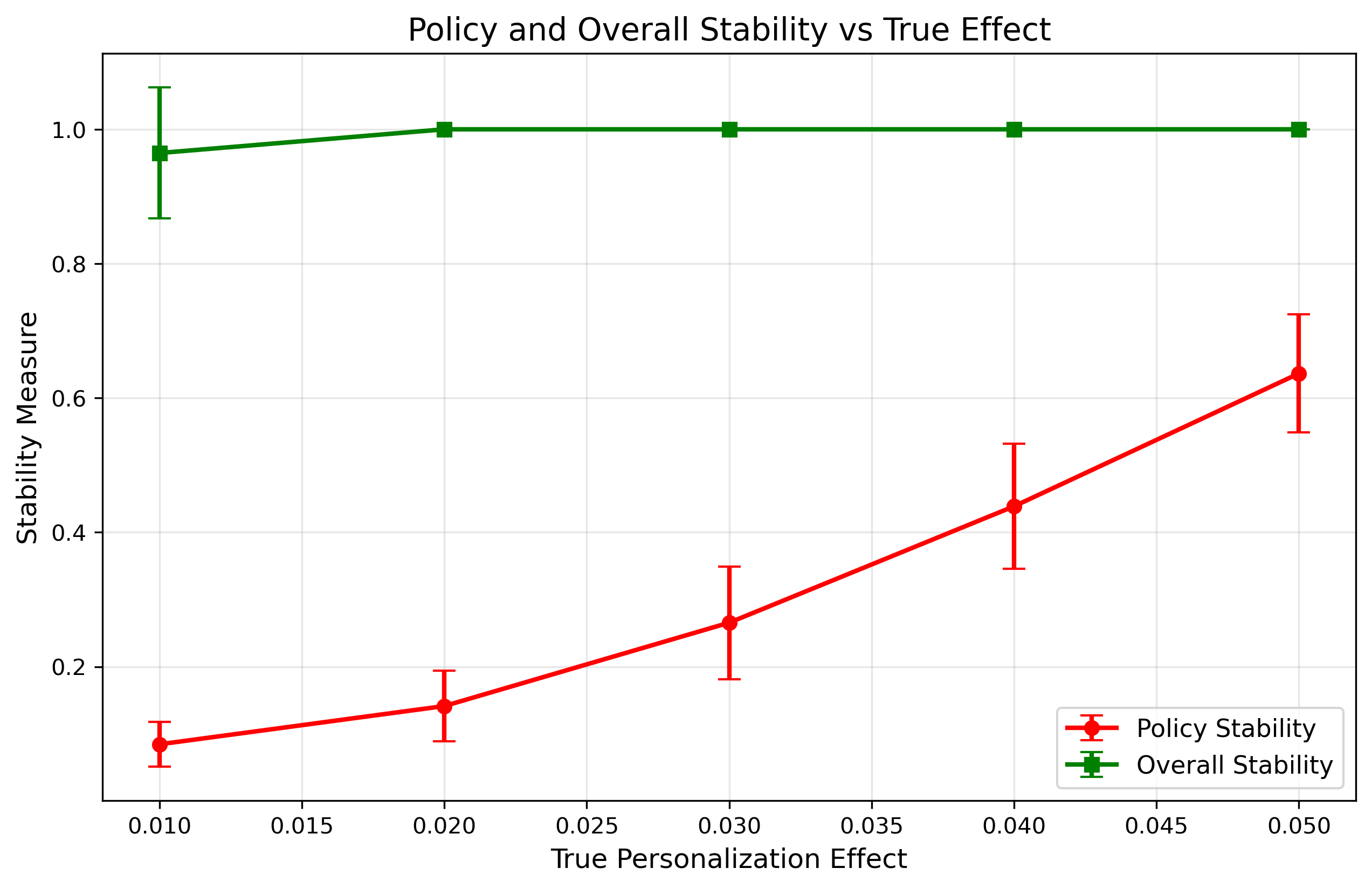}
    \caption{Policy stability for simulation instances. (left) Threshold small (right) Threshold large.}
    \label{fig:policy_stability_sim}
\end{figure}
\paragraph{Single-best action stability.} Here we plot the single-best action stability versus the action gap, where the action gap is computed as the utility of the best action versus the second best action. Figure~\ref{fig:action_gap} shows the results for both simulation instances. We can see that the single-best action is stable even if the action gap is as small as 0.007.

\begin{figure}[!htb]
    \centering
    \includegraphics[width=0.48\linewidth]{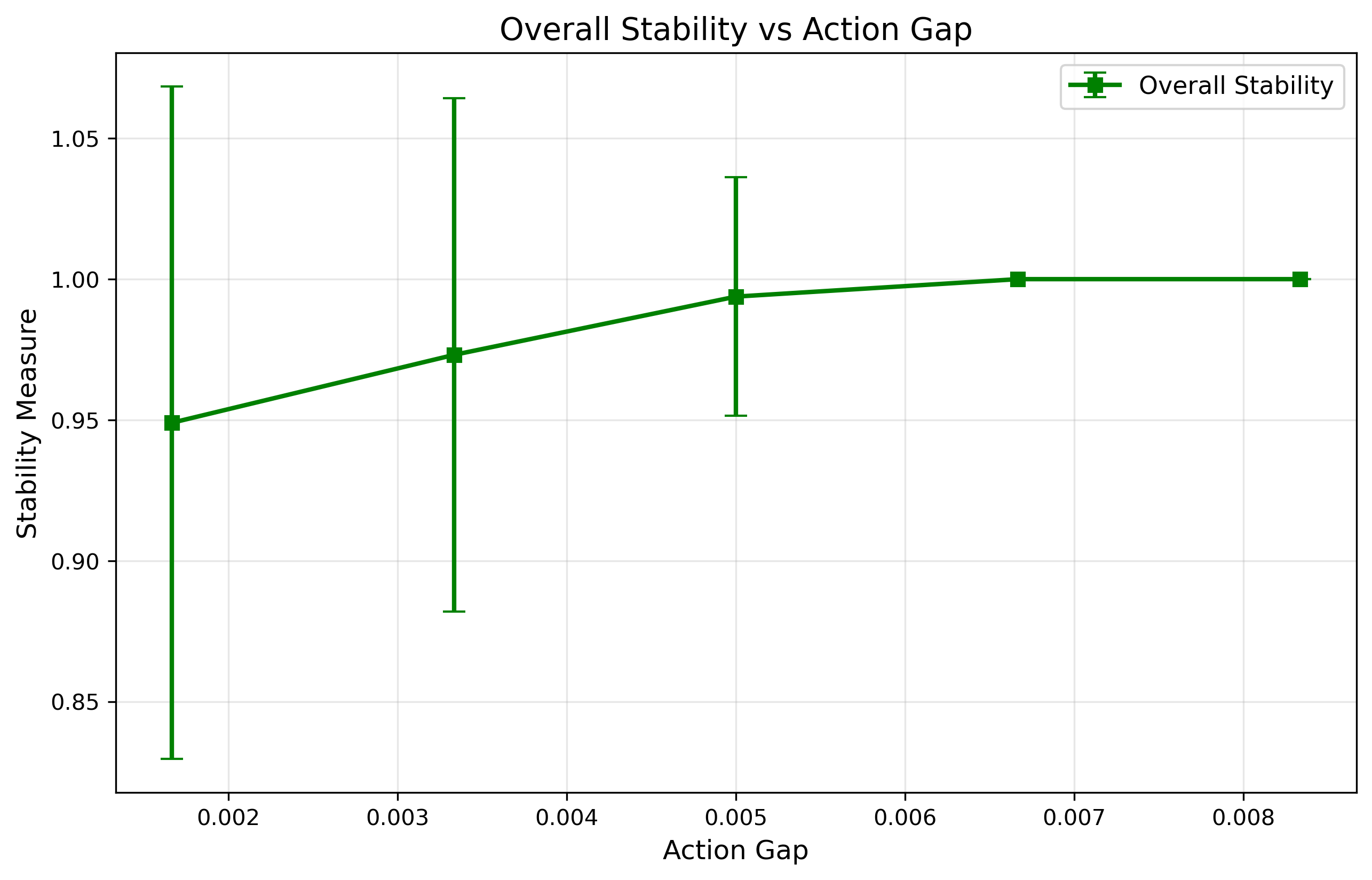}
    \includegraphics[width=0.48\linewidth]{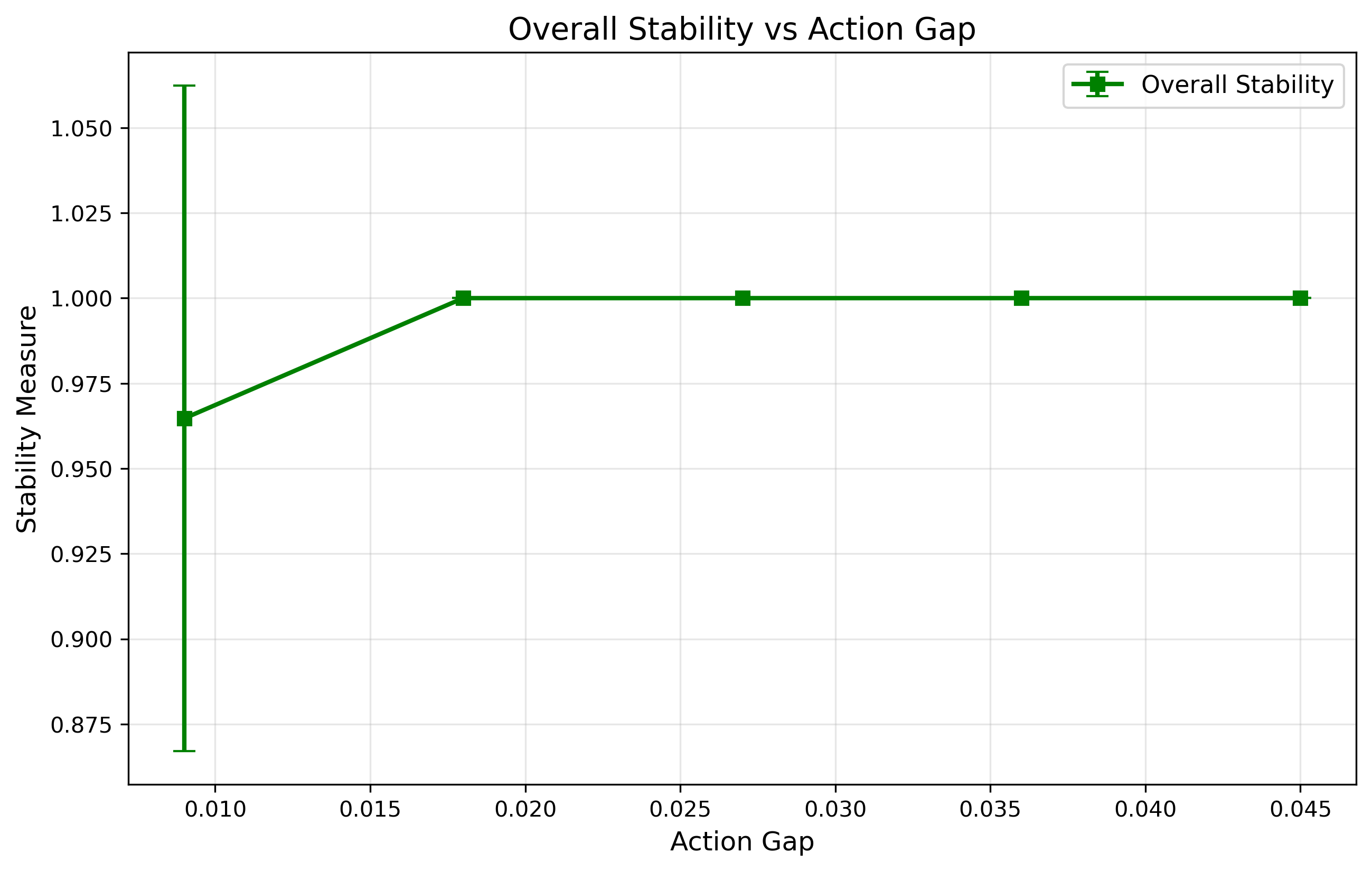}
    \caption{Action stability for simulation instance (left) Threshold small (right) Threshold large}
    \label{fig:action_gap}
\end{figure}

\paragraph{Cumulative Distribution Function for $p$-values.}
We plot the CDF for $p$-values for the simulation instance across different true effect sizes. Figure~\ref{fig:cdf_p_sim} shows the result. We can see that both TrainEval and SRP have very wide range of $p$-values. Also, when the true effect is zero, some $p$-values for the baselines fall below $\alpha=0.05$, which indicates false rejection. 
\begin{figure}[!htb]
    \centering
    \includegraphics[width=0.48\linewidth]{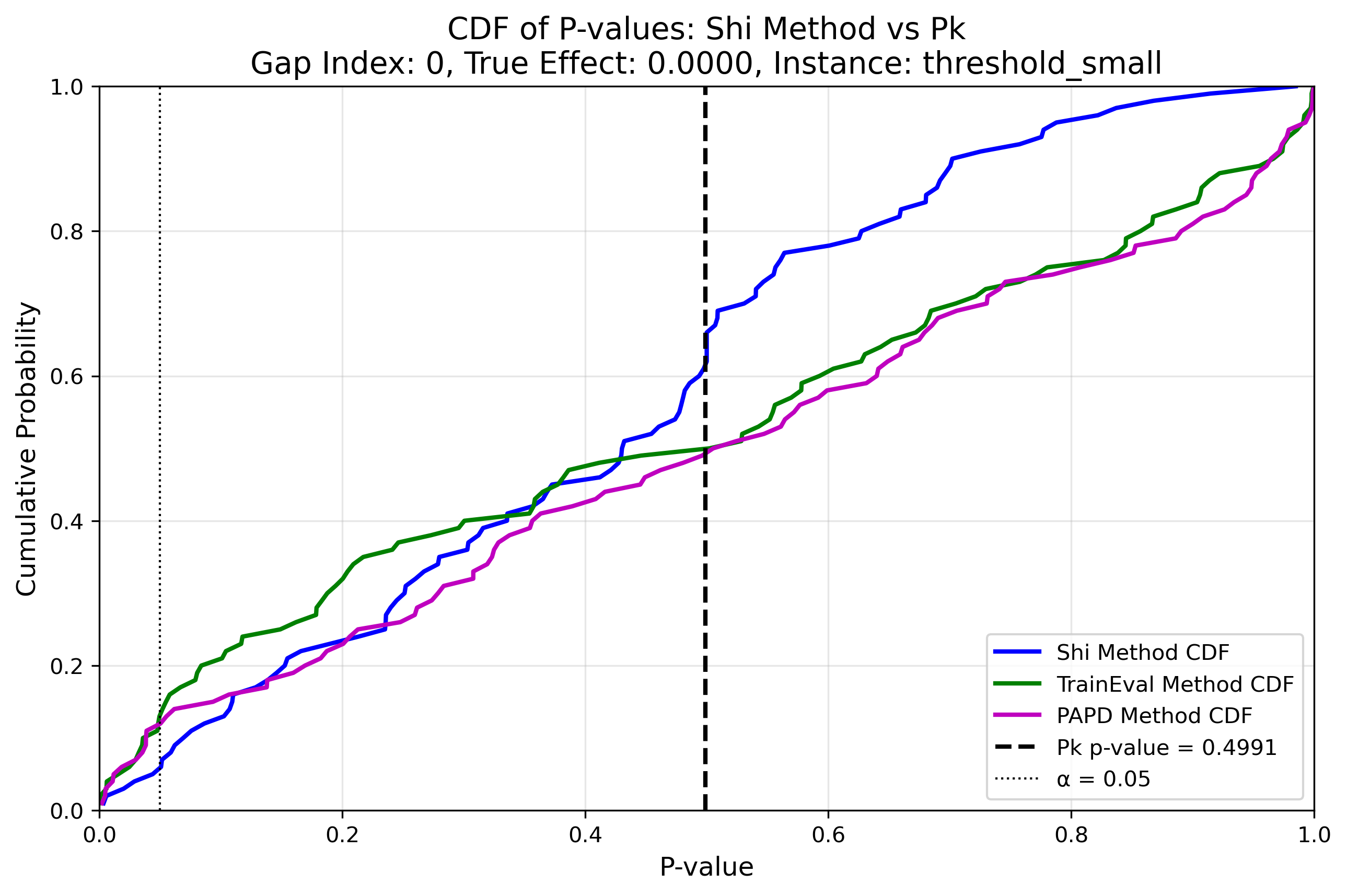}
    \includegraphics[width=0.48\linewidth]{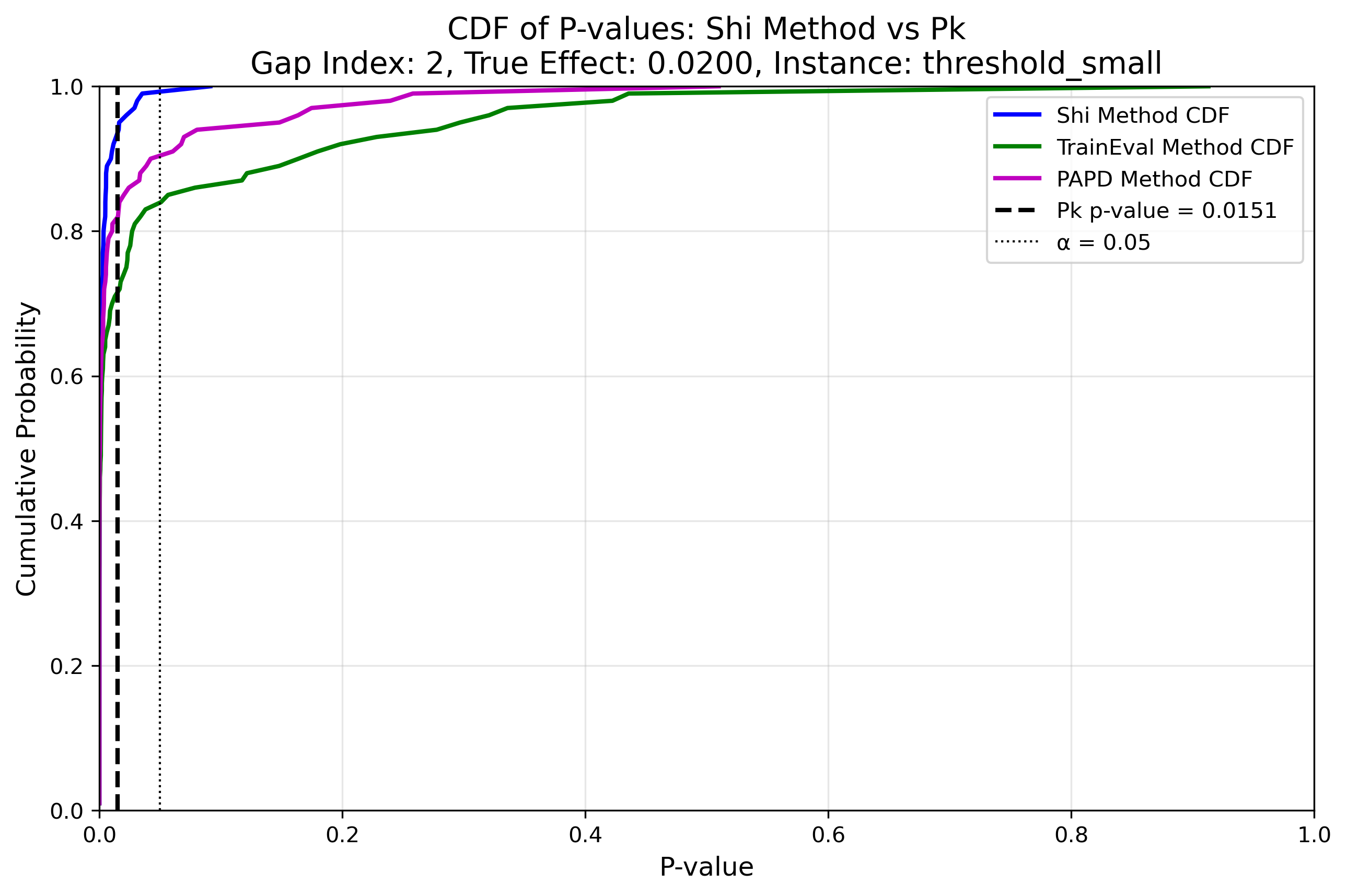}
    \caption{CDF for $p$-values for the simulation instance}
    \label{fig:cdf_p_sim}
\end{figure}

\subsection{Experimental Methods for Real-world Datasets}

For all datasets, we run 100 repeated splits for all methods. We use common policy classes-- trees, forest and linear models with thresholding. There are many methods for learning personalized decision policies, and these policy classes are flexible and common methods available in standard packages, like sklearn and econml. Our focus in the paper was in demonstrating a general methodology for estimating the size of personalization effects, rather than definitively assessing whether personalization effects are every present in a particular domain. For this reason, we did not exhaustively optimize over policy classes and hyperparameters in each domain, and instead aimed to pick a common method using for personalized policy learning. We explore the impact of the chosen policy class in a subsequent section. The learner and data processing scripts we implement in this section use standard packages in Python, in particular, scikit-learn (\url{https://scikit-learn.org/}), NumPy (\url{https://numpy.org/}), and pandas (\url{https://pandas.pydata.org/}). All packages are publicly available. %We include the hyperparameters used for each setting inside of the provided code. 

While we use tree/forest based methods, which are quite interpretable and flexible, in three domains, in the joke (recommendation) domain we used the argmax of a linear outcome model class. This was done for two reasons. First, prior work~\cite{kong2020sublinear} had used such a policy class for this domain. Second, we found in initial experiments that, likely due to the fairly large covariate dimension $(D=90)$, learning with decision trees was quite computationally slow. As our K-fold personalization estimate requires learning $K$-personalized decision policies across $S$ random partitions of the data, for computational reasons we focused on linear disjoint classes. 

% For three of the datasets (job, health and education) we use the econml doubly robust policy tree method ("DRPolicyTree") for learning a personalized policy, with the following hyperparameters:  
%                 max\_depth=2, 
%                 min\_impurity\_decrease=0.01, 
%                 honest=True, 
%                 min\_samples\_leaf=5,
%                 and a known propensity model, using the known probabilities of each action in the respective domain. 

For the comparison with PAPD~\cite{imai2023experimental} we use the specialization of PAPD to compute the personalization effect as shown in Algorithm~\ref{alg:papd_norm_personalization}-- see the prior section for additional details.

Below we describe additional processing and modeling details for each specific dataset. 

\paragraph{Semi-Synthetic Job Corps dataset.} We use the source data from a large randomized trial evaluating the impact of the Job Corps program~\cite{schochet2008does} on participants age 18-24. The data are available publicly online at \url{https://www.openicpsr.org/openicpsr/project/113269/version/V1/view}. The original data are in sas, and we converted the data to csv. We merged the impact and key\_vars sas datafiles to create a single dataframe of covariates, treatment (or control) status and outcomes.  We used 12 covariate features from the dataset: `FEMALE',  `HASCHLD', `AFTER\_ZT', `ARRST\_GR', `EDUC\_GR', `AGEGROUP', `INPERS', `IN57', `NONRES', `TYPEAREA', `RACE\_ALL',  `APP\_QTR'. As outcome variable we use 'EARNY4' which is the self reported weekly wages 4 years after program enrollment. We restricted the data to individuals that had all covariates, treatment status and outcome variables present (no missing data). This yielded a final dataset of 10214 participants of age 18-24 randomly assigned to treatment or control. All covariates and outcome values for the 10214 individuals are unmodified except for the outcomes for participants aged 18-19 at enrollment, where we 
\begin{itemize}
\item Add a synthetic positive wage offset of +$\$30/week$ in the control group
\item Add -$\$5/week$ (a negative offset) to the wages reported for those in the Job Corps group. 
\end{itemize}
We consider a personalized intervention policy class that can make different decisions for the three age brackets analyzed in prior work~\cite{schochet2008does}. For our KPT and the baseline TrainEval, we therefore use `AGEGROUP' as the input feature for the personalized policy class. We use the public econml package (\url{https://github.com/py-why/EconML}) DRPolicyTree learner for the personalized policy class, with parameters                 max\_depth=2, 
                min\_impurity\_decrease=0.01, 
                honest =True, 
                min\_samples\_leaf=5,
 and using known propensity model weights (since the original data were an RCT). A policy tree is a flexible and popular policy class similar to a decision tree. We did not seek to carefully tune these parameters and instead chose them because honest trees have been noted to have beneficial properties, and a minimum impurity decrease threshold, a limited depth and a minimum number of data points per leaf is likely to reduce the potential for learning a decision policy that is overfit to the training data with poor generalization performance.  Both our KPT and baseline TrainEval use a reward outcome model and we use the above 12 features and intervention (JobCorps or control) to predict reported wages. We use the random forest regressor from the sklearn package for reward outcome modeling.  For learning the single best intervention ("best-arm" learning),  we estimate the empirical outcome of each intervention using the known propensity weights (which here come from a RCT) and select the action with the highest estimated outcome. We use an identical personalized policy class, best single intervention policy learner, and reward outcome model for Nefazodone. 
 
 %use a doubly-robust learner from the econml package to estimate the ATE across all observations, and output treatment if the estimated ATE is positive, and control otherwise. %For learning a personalized decision policy, we use the econml package doubly robust "DRPolicyTree" with the following hyperparameters:  
 %max\_depth=2,   min\_impurity\_decrease=0.01,   honest=True,     min\_samples\_leaf=5,    and a known propensity model, using the known probabilities of each action. 
For PAPD~\cite{imai2023experimental} we use the same policy learning and best intervention learning algorithms as used in KPT and TrainEval. 
SRP\cite{shi2020sparse} does not support using separate features for policy learning and reward learning, so we use `AGEGROUP' for both policy learning and reward learning. Since there is only a single feature available, the sparse random projection aspect of their method is not relevant, and we use their test.QTI function (which implements the other aspects of their method, but not the sparse random projection designed to handle high dimensional dimension) from the publicly available package from \cite{shi2020sparse} for their method. We slightly extend their code to do 100 runs across 100 random partitions of the original input dataset, and include this in our released code. 

\paragraph{Nafazodone.}
The Nefazodone-CBASP clinical trial\cite{keller2000comparison} randomized 681 patients with nonpsychotic chronic major depressive to either receive the medication Nefazodone (a serotonin antagonist and reuptake inhibitor), the Cognitive Behavioral-Analysis System of Psychotherapy (CBASP), or both, for 12 weeks. A large number of potential baseline covariates were available, including alcohol and drug abuse status, and items in the Structured Clinical Interview for DSM Disorders. The Hamilton Rating Scale for Depression (HRS-D) was used to assess baseline status and status at week 12 and final follow up. Following previous analyses of individualized policy learning with this data\cite{shi2020sparse,zhao2012estimating}, we included 647 patients (216 Nefazodone, 220 CBASP, 211 combined) with complete records on 50 baseline covariates. 

The personalized policy class, best single intervention policy learner, and reward outcome model for Nefazodone are the same as those used in the semi-synthetic Job Corps domain. Specifically, for the reward outcome model ($\hat{r}(x,a)$) used in our KPT and baseline TrainEval, we use all covariates and use the random forest regressor from the public sklearn package.  For the personalized policy class and training, we use the public econml package (\url{https://github.com/py-why/EconML}) DRPolicyTree learner  with parameters                 max\_depth=2, 
                min\_impurity\_decrease=0.01, 
                honest=True, 
                min\_samples\_leaf=5, 
 and using known propensity model weights (since the original data were an RCT). For learning the best single intervention (best arm learning) on a particular data subset $\tilde{D}$, we estimate the empirical outcome of each intervention using the known propensity weights (which here come from a RCT) and select the action with the highest estimated outcome.  We use the same policy learning and best intervention learning algorithms for PAPD~\cite{imai2023experimental} as used in KPT and TrainEval. 

To implement the SRP method\cite{shi2020sparse}, we use the RPCV (randomized projection) and other functions from their public software released as part of their paper. This code assumes that the dataset size is divisible by 3 in order to perform the dataset splitting. We therefore use all 216 Nefazodone patients, and randomly sample 210 (of 211) patients that received the combined therapy, and 216 (of the 220) patients that received the CBASP therapy.
This means that our version of the SRP baseline uses 5 less data points (642 instead of 647 patients) compared to the cohort used in the original cited SRP results\cite{shi2020sparse}; however, we expect these to have a minor effect, and the original paper\cite{shi2020sparse} also does not find a positive personalization effect in this dataset with the "full" 647 patients.
We also modified their code to use different dataset re-ordering and random partitions and ran it 100 times, to look at sensitivity to how the data are split (as their method does rely on splitting the data).  We include our modified code for their method for this dataset in our released code.

\paragraph{MOOC Course Completion with Behavioral Interventions.}
We use a publicly available dataset (\url{https://osf.io/9bacu/overview}) from a large randomized controlled trial over 247 online courses across 2.5 years \cite{kizilcec2020scaling} . The outcome is course completion rate. We focus on the wave 1 and 2 data, which contains 199,517 participants with 20 covariates. We select 6 covariates `intent\_assess', `hours', `crs\_finish', `educ', `is\_fluent', and `HDI4'. There are 6 possible interventions in this dataset. For our KPT and the baseline TrainEval, we use the same reward outcome model class and single best intervention learning procedure. Specifically, we use all covariates and use the random forest regressor from the sklearn package for reward outcome modeling. For best single intervention (best-arm) learning, 
 we estimate the empirical outcome of each intervention using the known propensity weights (known from the RCT data) and select the action with the highest estimated outcome. For policy learning we use all covariates and the doubly robust policy tree learner from the econml package "DRPolicyForest" with hyperparameters n\_estimators=100 and a known propensity model, using the known probabilities of each action.  We explore sensitivity to the personalized policy class in a later section.

As in the other domains, for PAPD~\cite{imai2023experimental} we use the same policy learning and best intervention learning algorithms as used in KPT and TrainEval. We did not run SRP since there are 6 possible interventions and their approach is designed for binary interventions\footnote{Shi et al.\cite{shi2020sparse} did apply their approach to a scenario with three interventions, the Nefazodone dataset, but had to make multiple assumptions: they assumed the best single intervention was known, and create two different personalized policy class: one with the best single intervention plus one action, and one with the best single intervention plus the other action. If the best action is unknown, one could construct 15 pairs of policies and try to look for effects across these and use union bounding, but even here it would not address if there is a personalized policy that requires 3 or more interventions. }.

\paragraph{Joke Recommendation dataset.} We use the public and well studied Jester datasets (\url{https://eigentaste.berkeley.edu/dataset/}) from \cite{goldberg2001eigentaste}. Specifically we use Jester Dataset 1 
which contains 73,421 user ratings (from -10.00 to +10.00) for 100 jokes collected between April 1999 - May 2003. We follow the preprocessing methods in \cite{kong2020sublinear}. We first select 10 highest-rated jokes and keep the 48447 users who have ratings for all these jokes. We then keep the 10 highest-rated jokes as potential interventions, and transform user ratings for the remaining 90 jokes into a 100-dimensional covariate feature vector, to mimic a typical recommendation system where a system would leverage prior user ratings to inform new content recommendations. We then select a subset of these ratings by randomly assigning a joke for each user and only observing the rating for that joke, to mimic a randomized control trial with bandit feedback. We normalize scores to be between 0 and 5.0.

 For our KPT and the baseline TrainEval, for the reward outcome modeling we use the sklearn package lasso  model (with $\alpha=0.001$) to learn one outcome model per action ($\tilde{r}_a$). For policy learning, we first learn a lasso outcome model (with $\alpha=0.001$) per action ($\tilde{r}_a$) and the personalized policy is constructed by choosing the action that maximizes the predicted reward  (for the personalized policy). The best overall intervention is constructed using the same policy lasso outcome models using   $\tilde{a} = \arg\max \mu(a)$, $\mu(a) = \frac{1}{|D_s|}\sum_{i=1}^{D_s} \tilde{r}_a(x_i)$. SRP\cite{shi2020sparse} is designed for binary interventions and is not applicable to this setting for the same reasons detailed about for the MOOC Course Completion with Behavior Interventions dataset. 

For PAPD~\cite{imai2023experimental} we continue to use the same policy learning and best intervention learning algorithms as used in KPT and TrainEval. We note that in this domain in early explorations we did not center the outcome variable $Y$ in the PAPD estimator-- doing so significantly reduced the confidence interval size, and here, and for all settings, we center the outcome variable.

\subsection{Policy and Best Intervention Learning Stability}\label{sec:policy_stab}

We also compute the policy stability and single-best action stability for real data. Recall policy stability is the proportion of contexts for which the personalized policy learned is always the same, across all splits and folds: $1- \frac{1}{|X|}\sum_x I(\exists s,s' \; \exists k,k' \hat{\pi}_{sk}(x) \neq \hat{\pi}_{s'k'}(x))$.
%Since here the true best intervention is unknown, the 
Single-best action stability is computed as the percentage of the occurrence of the most frequently output best intervention across splits and folds. The personalized policy learning displays the proportion of contexts' that have at least two different interventions across different personalized policies (across folds and splits). The result is shown in Table~\ref{tab:policy_stability}. The number inside the parentheses denotes the number of actions for the dataset. We can see that for the joke dataset, policy learning is quite unstable, likely since there are a large number of covariates. 
%both policy learning and single-best action learning is highly stable, likely due to the large datasize and the likely presence of significant difference in what jokes are found funny for different individuals, and that some jokes may be more funny on average across the population.
For education, policy learning is very unstable, and the most frequent single-best action only occurs 25.8\% of the time. As there are four different actions, this indicates that the predicted single-best action is very close to uniform across repeated splits (compared to a perfectly uniform of 25\%). For jobcorps, the policy learning is again highly unstable, and the predicted single-best action is also quite unstable (58.2\% with two actions). We observe that though our theoretical guarantees for Type I error assume that best interventions can be easily learned, in our semi-synthetic Job Corps example, our method is still able to detect the benefit of a personalized policy. In the Nefazodone dataset  policy learning is unstable, but the same single-best action is learned 91\% of the time. This represents a setting where learning the best personalized policy is hard, while learning the single-best action is relatively easy. % and there is not a significant benefit of personalization as shown before in Figure~\ref{fig:ci_cdf}. 
\begin{table}[!htb]
    \centering
    \begin{tabular}{|c|c|c|}
        \hline Dataset (Number of Possible Interventions) & Policy & Best-arm \\
       \hline  Education (4)  & $6.5\times 10^{-5}$ & 0.258 \\
\hline Jobcorps (2)   & 0.009 & 0.582 \\
\hline Nefazodone (3) & 0.020 & 0.907 \\
\hline Joke (10)     & 0.0279 & 0.668 \\ % used to be 1 for policy
        \hline
    \end{tabular}
    \caption{Personalized policy and best intervention (arm) stability across domains. Personalized policy stability is the percent of contexts (in our dataset) for which the learned personalized policies always select the same intervention: $1- \frac{1}{|X|}\sum_x I(\exists s,s' \; \exists k,k' \hat{\pi}_{sk}(x) \neq \hat{\pi}_{s'k'}(x))$. Best intervention/arm stability is the intervention that is most frequently learned as the single best intervention. Nefazodone (which has 3 possible interventions) and Joke datasets (which has 10 possible jokes) have a high consistency of the best overall action, but Jobcorps (with 2 interventions) and Education (with 4 possible actions/interventions) do not. }
    \label{tab:policy_stability}
\end{table}

\subsection{The Impact of Personalized Policy Class}\label{sec:policy_class_choice}

The personalization estimate depends on an input choice of a policy class, and as we discuss, lack of a significant effect for one policy class does not suggest that there is never an impact of personalization in this domain. For example, it could be that using a simple policy class is insufficient to identify important interactions that, when taken into account, allow the policy class to meaningfully identify that some subgroups benefit significantly from different interventions. 

Here we explore how our K-fold Personalization estimates vary across different choices of the personalized policy class for the four domains: semi-synthetic Jobcorps, Nefazodone, MOOC Education and Joke Recommendation. We consider multiple possible policy classes, including a thresholded linear policy, a DR Forest (from the econml package), a DR Policy Tree (from the econml package). For the Joke Recommendation dataset, we consider using ridge, linear and lasso regression. Figures~\ref{fig:jobcorps_multiclass},~\ref{fig:nefazodone_multiclass},~\ref{fig:education_multiclass}, and ~\ref{fig:joke_multiclass} display the impact of the choice of policy class on our PKT method's estimation of the personalization effect. 

%Nefazodone dataset and the semi-synthetic Jobcorps dataset. 
%We plot the confidence intervals with $z$-statistics  for Jobcorps and Nefazodone datasets, and compare them with other baselines in Figure~\ref{fig:alltypes_jobcorps} and \ref{fig:alltypes_nefazodone}. 

Unsurprizingly, we observe the choice of policy class does have an impact on the confidence intervals and $z$-statistics. Often this effect happens to be fairly minor in the domains and policy classes explored. Hoewver it is worth noting that in the case of Nefazodone, the DR forest policy class does find evidence of a significant impact of personalization, estimating a very small effect in contrast to the DR tree policy class that we used for Job and  Nefazodone in our main results. We note that the actual personalization effects estimated by the two policy classes are quite similar: ~0.1 for DR Policy Tree (but not significant) and ~0.5 for DR Forest Policy. This highlights that beyond simply assessing if the significance test passes or fails, it will likely be of use to consider the size of the personalization effect for stakeholders to decide whether or not personalization will be useful in their setting given other tradeoffs. Indeed, this was one of our key motivations behind showing our results both with confidence intervals and $z$-statistics.  

Note that policy class choice also will impact baselines like TrainEval. In Figure~\ref{fig:education_multiclass_traineval} we see that empirical cdf of statistics from TrainEval varies based on the underlying policy class in the MOOC Education domain. In the case of using the DR PolicyTree policy learning from econml, in a number of cases the best personalized policy learned did assign all individuals the same intervention as the one learned as the best overall intervention, which results in the jump at 0 on the x-axis.

\begin{figure}[!htb]
    \centering
    \includegraphics[width=0.8\linewidth]{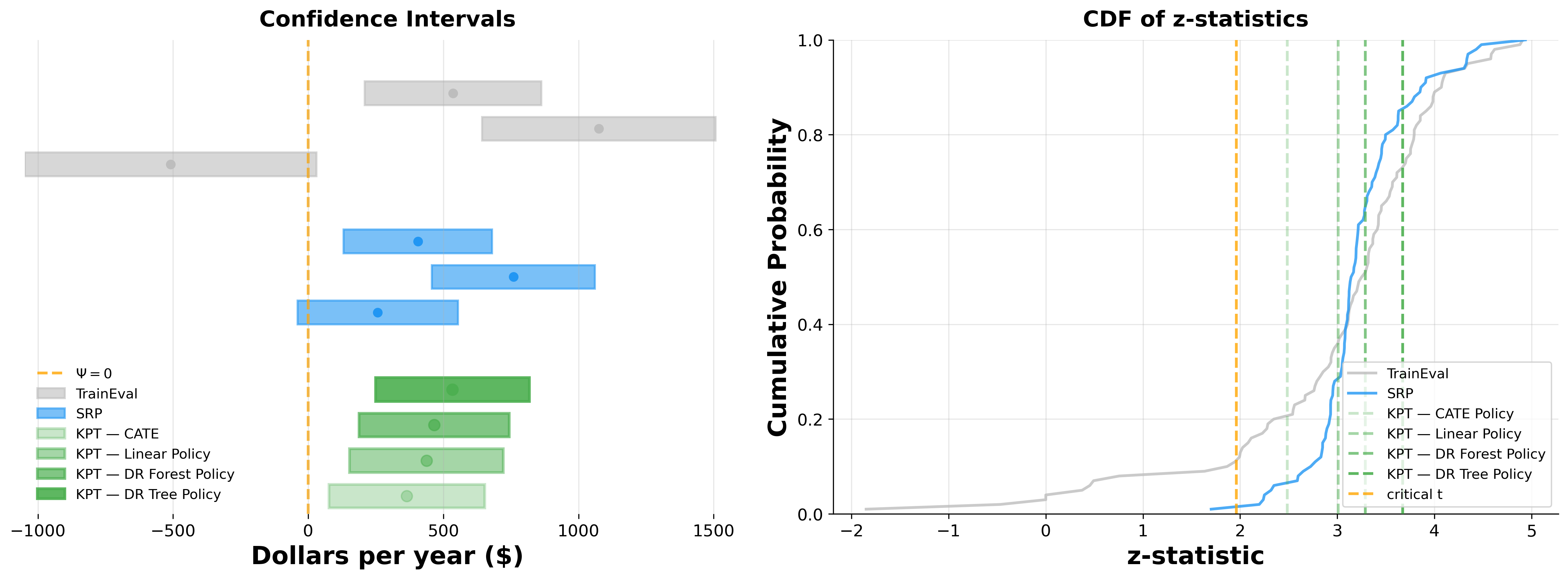}
    \caption{Jobcorps.   Examining the Impact of Policy Class by Trying Multiple Policy Classes.}
    \label{fig:jobcorps_multiclass}
\end{figure}

\begin{figure}[!htb]
    \centering
    \includegraphics[width=0.8\linewidth]{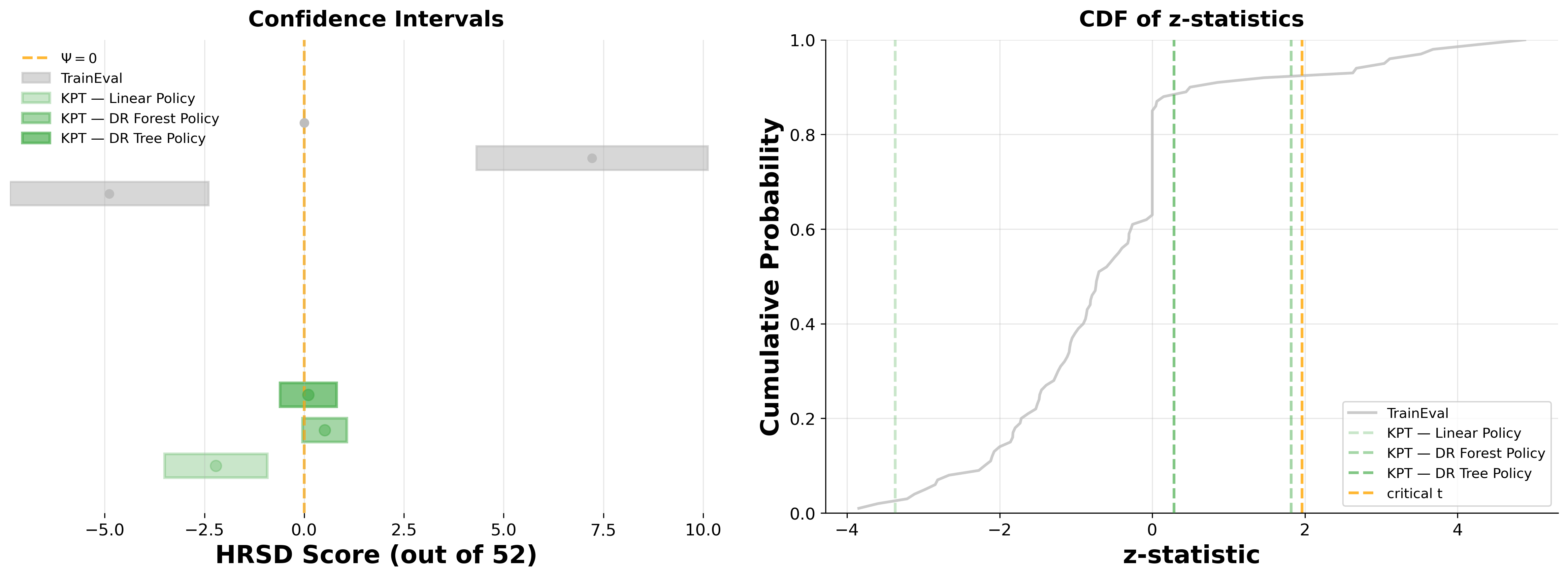}
    \caption{Nefazodone.   Examining the Impact of Policy Class by Trying Multiple Policy Classes.}
    \label{fig:nefazodone_multiclass}
\end{figure}

\begin{figure}[!htb]
    \centering
    \includegraphics[width=0.8\linewidth]{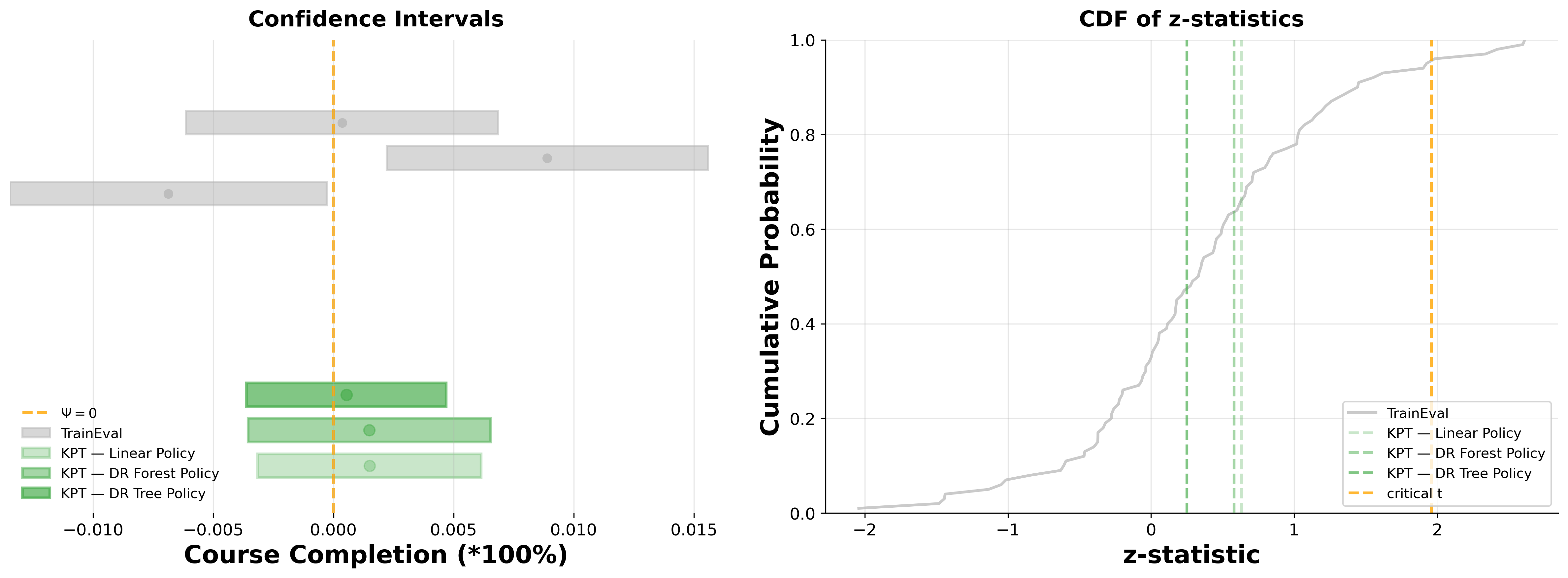}
    \caption{MOOC Education. Examining the Impact of Policy Class by Trying Multiple Policy Classes.}
    \label{fig:education_multiclass}
\end{figure}

\begin{figure}[!htb]
    \centering
    \includegraphics[width=0.8\linewidth]{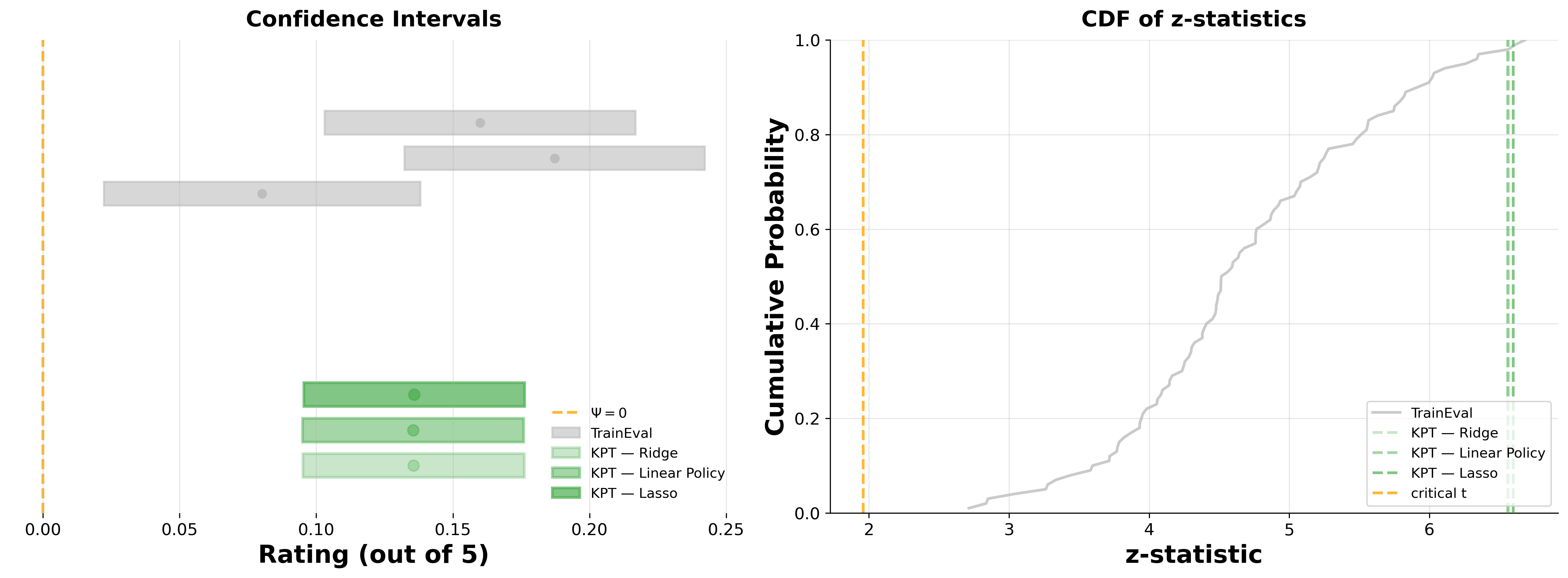}
    \caption{Joke Recommendation.  Examining the Impact of Policy Class by Trying Multiple Policy Classes.}
    \label{fig:joke_multiclass}
\end{figure}

\begin{figure}[!htb]
    \centering
    \includegraphics[width=0.8\linewidth]{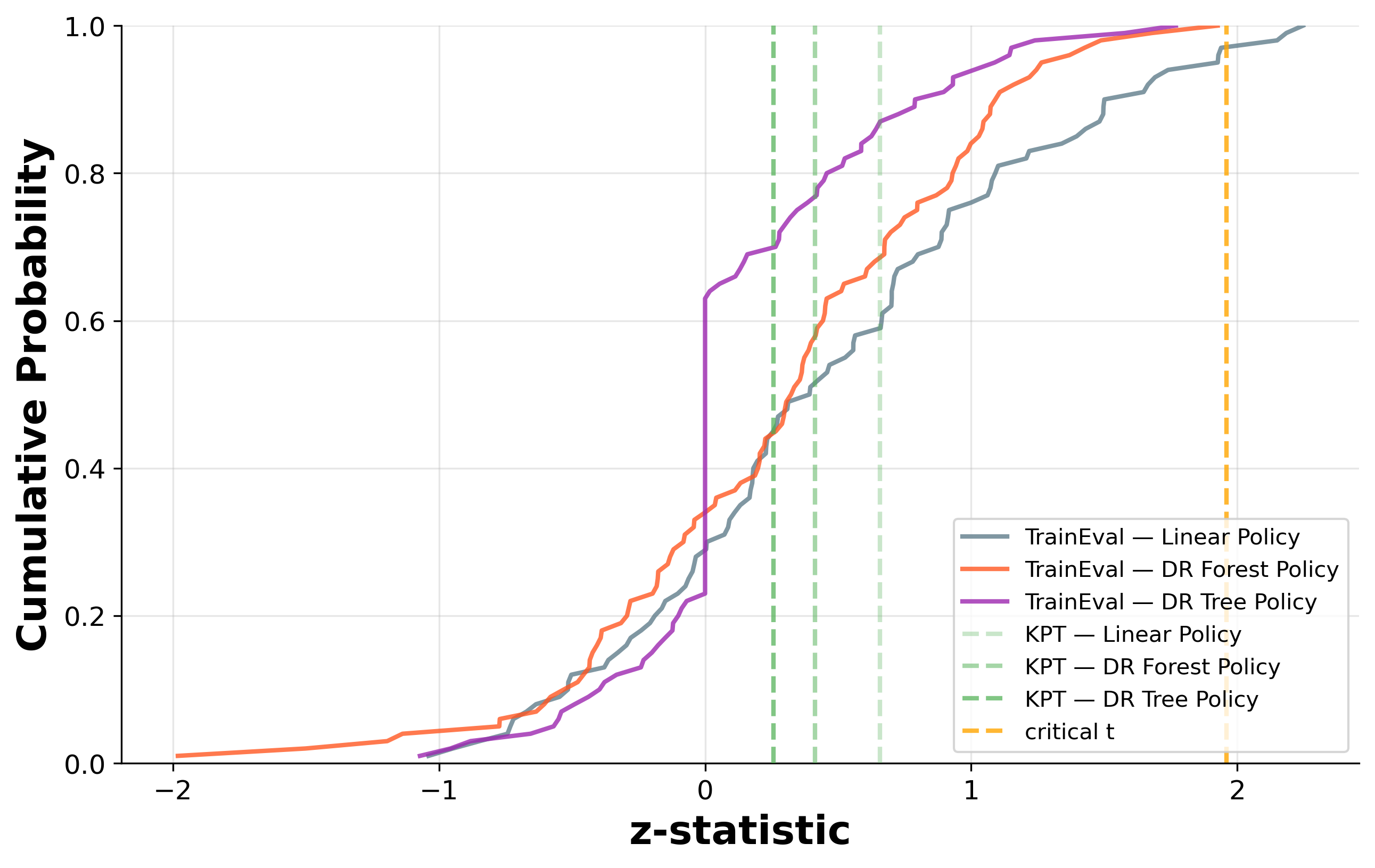}
    \caption{MOOC Education. Examining the Impact of Policy Class by Trying Multiple Policy Classes  for TrainEval}
\label{fig:education_multiclass_traineval}
\end{figure}

We also consider how the choice of the class for the personalized policy may impact the learned policy stability, and explore this for the  Job Corps dataset. Specifically, we consider the CATE policy, which learns the CATE using a doubly robust learner, and then allocates  treatment to an individual with a particular context if their estimated CATE is greater than zero, and selecting control otherwise. We also consider personalized policy classes using a linear policy, DR forest policy and DR tree policy. We define policy stability as the variability of the learned policy across repeated random data partitions, with higher values corresponding to greater sensitivity to sampling fluctuations. 

Figure~\ref{fig:stability_type} shows the results. The results indicate that DR Forest attains higher stability (0.268) relative to the other approaches. Note that the DR Tree stability is slightly different (0.014 vs 0.009 reported in Table~\ref{tab:policy_stability}) because we used a different random seed for these experiments. The overall personalization effect in both cases with our KPT method with DR Tree is nearly identical (\$10.1 vs \$10.26).  
\begin{figure}[!htb]
    \centering
    \includegraphics[width=0.5\linewidth]{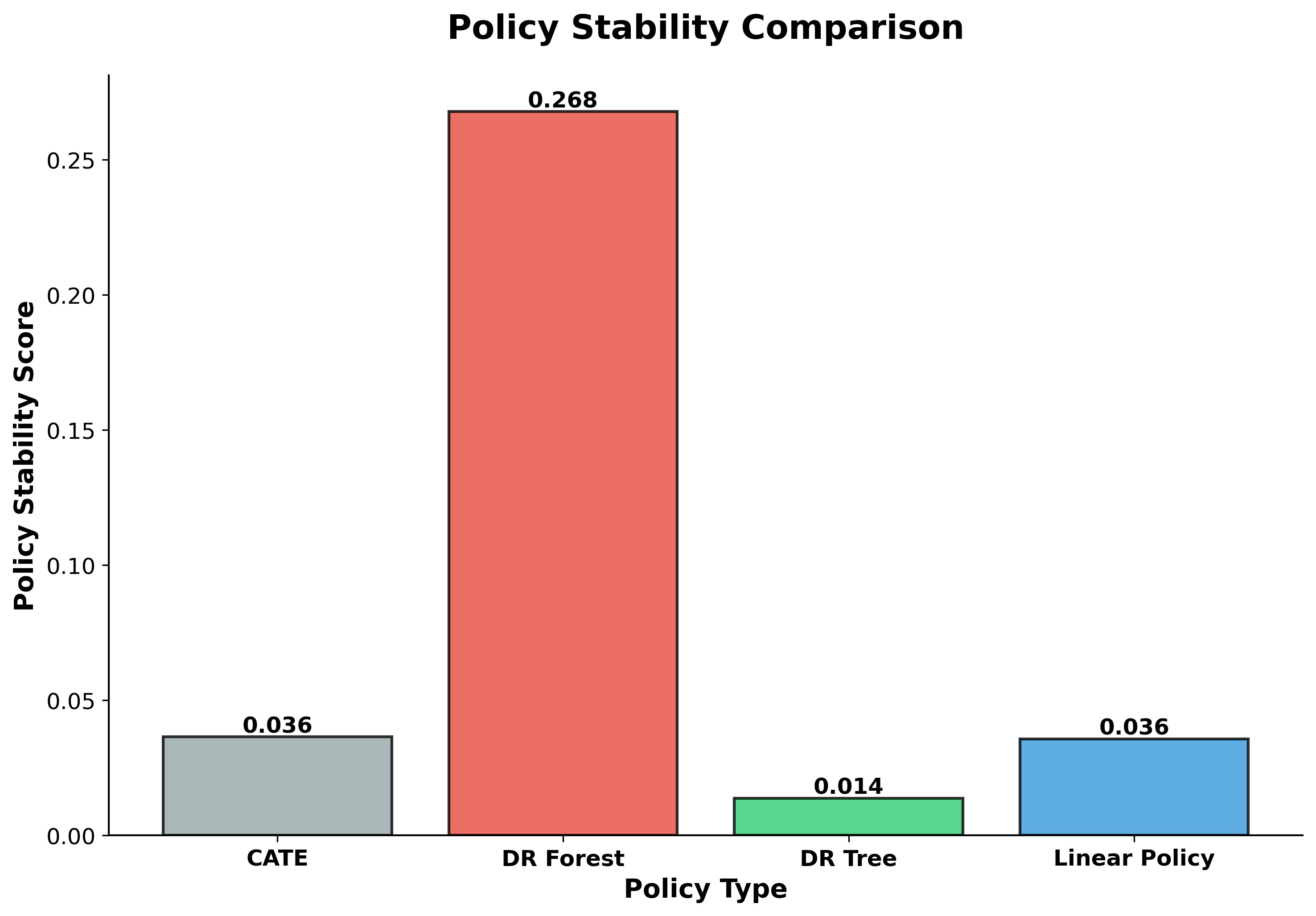}
    \caption{Policy stability for different policy types on Jobcorps dataset}
    \label{fig:stability_type}
\end{figure}

It is an interesting future direction to consider how to automatically select the policy class that results in the best policy learner given a domain: a simple first step could be to do model selection over the policy class, though this would also increase computational complexity. A stakeholder may also want to consider how stable policy learning is across splits and folds for a particular policy class in deciding, if personalization has a positive effect, which policy class to use.

\subsection{Comparison of Different $p$-value Aggregation Strategies}

In our K-fold personalization estimate, we compute a single statistic, confidence interval and $p$-value across many random partitions, and prove the resulting statistic is also asymptotically efficient. This relates to some prior work on algorithm comparison, where, for example, machine learning researchers may wish to compare the expected future performance of different classification algorithms by using a fixed existing dataset to both train each algorithm and assess its performance~\cite{dietterich1998approximate}. In such settings a prior proposal has been to compute a p-value for each of a set of $J$ splits (such as $J=100$), and then aggregate across the $J$ p-values, and we might employ a similar strategy in our setting for estimating the personalization effect. 

However, we note that aggregating across these $J$ values is not trivial, because the values are not independent -- each arises from the same underlying dataset, but randomly partitioned in different sets. This violates the standard assumption of independence that is often used in classic approaches for aggregating $p$-values like Fisher's method. Indeed, prior work on algorithm comparison has pointed out that repeated train-test splits yield dependence between splits and, without careful consideration, can inflate Type I error~\cite{dietterich1998approximate}. In contrast, our approach asymptotically achieves zero Type I error (See Corollary 3.2), so we do not suffer from inflated Type I error.

To aggregate across algorithm outputs across multiple train-test splits, Dietterich's 5x2cv Paired t-test~\cite{dietterich1998approximate} and Alpaydin's Combined 5x2cv F-test construct a single statistic using multiple splits and folds \cite{raschka2018model}. Dietterich’s 5x2cv Paired t-test constructs a ratio where the numerator uses an estimate (of the gain of one algorithm over another) from only one split in the numerator, and the denominator is the square root of the average of the within fold variances, across splits. In contrast, our estimate uses results from all splits in both our numerator and denominator of the test statistic. Perhaps unsurprisingly, this means empirically Dietterich’s numerator can be quite unstable depending on which split it is computed on, as shown empirically by Alpaydin. In contrast, Alpaydin’s Combined 5x2cv F-test statistic does average across all folds and splits, similar to our’s approach, but his estimator also assumes independence across splits and folds.

Neither estimator is proven to have a sampling distribution that coincides with any standard parametric distribution. Indeed, 
Dietterich and Alpaydin both state that their statistic is only approximately distributed as the relevant distribution, since the quantities violate the independence assumptions needed for their stated direct argument to hold. In contrast, our estimator accounts for both the within-fold variability and between-fold variability, and our test does not rely on treating fold estimates as independent units. In particular, we prove that, under some assumptions, our test statistic has an asymptotically normal statistic, in contrast with the Alpaydin’s combined test which so far lacks such guarantees.

The Cauchy combination test~\cite{liu2020cauchy} is a more recent method explicitly designed for aggregating many p-values where the p-values may have dependency structure, using the following statistic:
\begin{equation}
t = \sum_i w_i \tan(\pi( 0.5-p_i)),
\end{equation}
where $w_i$ are weights, bounded between $(0,1)$ that sum to 1, and $p_i$ is the $i$-th p-value. Those authors show a Cauchy distribution well approximates the tail distribution of their proposed statistic. However the Cauchy combination test is not designed for aggregating p-values that arise from multiple partitions of the same data, which is our setting. For this setting, where each p-value is computed using a variable split from the same data, 
others~\cite{meinshausen2009p} proposed an alternate aggregation method that computes a capped quantile:
\begin{equation}
p_{agg} = \min(1,quantile_{\gamma}(\frac{p_1}{\gamma},
\frac{p_2}{\gamma},\ldots)),
\end{equation}
where $0<\gamma<1$ and can be selected adaptively. Very recent work~\cite{gasparin2025combining} has further analyzed and advanced this setting of exchangeable P-values. Note that none of the methods guarantee efficiency in our setting. 

To explore the empirical use of such methods, we consider a variant of our methods and the baselines, in which we compute one p-value per split, and then aggregate the p-values using either the Cauchy combination test (which can handle dependent p-values)\cite{liu2020cauchy} or the capped quantile (designed specifically for repeated data splitting)\cite{meinshausen2009p}. The results are in Table~\ref{tab:test_results}.

\begin{table}[ht]
\centering
\begin{tabular}{l l l l l}
\hline
Dataset & Method & p-value & Cauchy p-value & Multisplit p-value \\
\hline
JobCorps & KPT & $1.33 \times 10^{-4}$ & $2.57 \times 10^{-10}$ & $1.66 \times 10^{-4}$ \\
JobCorps & PAPD &  & $7.68 \times 10^{-11}$ & $1.37 \times 10^{-6}$ \\
JobCorps & TrainEval &  & 1.0000 & $3.76 \times 10^{-6}$ \\
JobCorps & SRP &  & $2.61 \times 10^{-5}$ & 0.0016 \\
Nefazodone & KPT & 0.4156 & 1.0000 & 1.0000 \\
Nefazodone & PAPD &  & 0.0039 & 1.0000 \\
Nefazodone & TrainEval &  & 1.0000 & 1.0000 \\
Nefazodone & SRP &  & 0.9266 & 1.0000 \\
Education & KPT & 0.2473 & 0.1225 & 0.4852 \\
Education & PAPD &  & 0.0266 & 0.3680 \\
Education & TrainEval &  & 0.0085 & 0.6365 \\
Joke & KPT & $3.77 \times 10^{-11}$ & $1.72 \times 10^{-15}$ & $6.69 \times 10^{-11}$ \\
Joke & PAPD &  & $7.22 \times 10^{-16}$ & $1.95 \times 10^{-12}$ \\
Joke & TrainEval &  & $1.67 \times 10^{-16}$ & $1.42 \times 10^{-10}$  \\
\hline
\end{tabular}
\caption{Hypothesis testing results across datasets and methods}
\label{tab:test_results}
\end{table}

We can see in Table~\ref{tab:test_results} that using these p-value aggregation methods across splits generally yield similar findings to our KPT approach. However, the Cauchy p-value aggregation method sometimes gives surprising results and is prone to outliers. 
For example, in the Nefazodone dataset, the Cauchy aggregate p-value is low when using the 100 PAPD estimates, even though our KPT method, the multi-split $p$-value aggregation methods, TrainEval and SRP are all suggesting there is no significant effect. When looking at the $p$-values for these cases, we observe that in Nefazodone for PAPD, p-values per split, a few $p$-values are very tiny ($<10^{-3}$). A similar phenomenon occurs in Education, where our KPT estimate and the multi-split $p$-value aggregation methods suggest no significant benefit of personalization, but for both PAPD and TrainEval a very small percentage of splits had a small p-value, and the resulting Cauchy combination p-value is small. In our Job Corps domain, the Cauchy $p$-value for TrainEval is 1 even though there is evidence of significant effect from other results. In Job Corps, there exists a split where the test reports a really large $p$-value while other $p$-values are small: for example, this can occur if on this particular set of folds, the personalized policy matches the single best action for all contexts. The Cauchy combination sums $\tan(\pi(0.5- p_i))$ which is $ \approx \tan(\pi/2)$ if $p_i \approx 0$ and $ \approx \tan(-\pi/2)$ if $p_i \approx 1$. In both cases, tan explodes to infinity. Therefore the Cauchy aggregated $p$-value will tend to report significance if any of the $p$-values are very small and report insignificance if any test fails with a $p$-value of approximately 1, which are not desirable when the variation is due to repeated splitting. The capped quantile method for repeated splitting~\cite{meinshausen2009p} and our KPT aggregation method instead are more suitable for the setting when the multiple values are caused by variation in the splits.
%is very sensitive to these small outliers, which causes these small $p$-values. } 

Perhaps even more importantly, the Cauchy combination test and capped quantile method are techniques for aggregating p-values, and they do not provide a confidence interval over an estimated effect size. In contrast, our interest is in both creating a statistical test (and p-value) as well as a tight confidence set over the proposed personalized effect. 

\subsection{Heterogeneous Treatment Effect Testing and Personalization Effect Testing}
\label{sec:hte_and_personalization}

As noted above, heterogeneous treatment effects are a necessary, but not sufficient, condition for personalization effects to exist. Heterogeneous treatment effect estimation is itself an open area of research, and most methods focus on settings where  there are binary intervention spaces (aka treatment/control). In such settings, or more broadly, if there are applicable heterogeneous treatment effects for the problem setting, a stakeholder could follow a two stage process, of (1) first evaluating if a test for HTE returns positive evidence of HTE, and then (2) if positive, use our proposed KPT to estimate the effects of personalization. We now show illustrate a setting where this two stage process could be used to learn no personalization effect is likely to occur, as well as to illustrate where HTE is insufficient to predict personalization effect.

Specifically, we tested for the presence of HTE in the Nefazodone depression dataset using two state-of-the-art approaches for HTE testing~\cite{imai2025statistical,yadlowsky2025evaluating} which can be used with binary treatment intervention spaces (We used the public grf package version 2.4.0 and evalITR package version 1.0.0, and R version 4.3.3). Our results are shown in Table \ref{tab:HTE_nefazodone}. Both tests do not show evidence of heterogeneous treatment effects, which also immediately implies there would be no personalization benefits. This matches our personalization results, prior work where a personalized policy always assigned all individuals to the combination therapy~\cite{zhao2012estimating},  and prior related work which found no evidence of a personalization effect~\cite{shi2020sparse}. 

\begin{table}

\caption{Heterogeneous Treatment Effect (HTE) Tests on the Nefazodone Dataset. RATE~\cite{yadlowsky2025evaluating}
\label{tab:HTE_nefazodone} and Imai-Li HTE~\cite{imai2025statistical} are two methods for testing for heterogeneous treatment effects.}
\centering
\begin{tabular}[t]{l|l|c|c|c|c|c}
\hline
Comparison & Method & Statistic & Estimate & Std. Err. & z & p-value\\
\hline
Combo vs Nefazodone & Imai--Li HTE & 1.106 & --- & --- & --- & 0.954\\
 & RATE & --- & -0.299 & 0.898 & -0.333 & 0.630\\
Combo vs Therapy & Imai--Li HTE & 3.376 & --- & --- & --- & 0.642\\
 & RATE & --- & -0.477 & 0.862 & -0.553 & 0.710\\
\hline
\end{tabular}
\end{table}

Next we provide an illustrative simulated setting where there exists  treatment effect heterogeneity but no personalization benefit, as well as a setting where there is no HTE.  The simulation settings are described as follows. We have dataset $(X,A,Y)$ where $X$ is in a set of $20$ discrete states and $A\in \{0,1\}$  with probability $1/2$ for each. In the first setting we have a setting with no hetereogenous treatment effects, SimNoHTE, where $Y=0.1 X +A+ \epsilon$ with $\epsilon\sim N(0,\sigma^2)$ with $\sigma=0.2$. In a second setting we induce HTE but no personalization, SimHTE-NP, $Y=0.1 X +(1+0.5X)\cdot A+ \epsilon$ with $\epsilon\sim N(0,0.2^2)$, so the best optimal action is always $1$. 
We use \cite{imai2025statistical} and \cite{yadlowsky2025evaluating} as state-of-the-art baselines for HTE tests. We run both of the baselines for 20 runs. The results are shown in Table \ref{tab:HTE_sim}. 

% \begin{table}[tbh]
% \centering
% \caption{Heterogeneous Treatment Effect (HTE) Tests on the Illustrative Simulations. In SimHTE-NP there are HTE but no personalization effect. In SimNoHTE there are no hetereogenous treatment effects. See text for simulation details.}
% \begin{tabular}{|l|c|c|} \hline
% Method & SimHTE-NP (\%) & SimNoHTE (\%) \\ \hline 
% Imai--Li HTE \cite{imai2025statistical} & 31\% & 5\% \\ \hline 
% RATE \cite{yadlowsky2025evaluating}  & 100\% & 10\% \\ \hline
% \end{tabular}
% \label{tab:HTE_sim}
% \end{table}

\begin{table}[ht]
\centering
\caption{ Heterogeneous Treatment Effect (HTE) Tests on the Illustrative Simulations. In SimHTE-NP there are HTE but no personalization effect. In SimNoHTE there are no hetereogenous treatment effects. See text for simulation details.}
\label{tab:HTE_sim}
\begin{tabular}{|l|c|c|}
\hline
Method & SimHTE-NP (\%) & SimNoHTE (\%) \\
\hline
Imai--Li HTE \cite{imai2025statistical} & 25\% & 5\% \\
RATE \cite{yadlowsky2025evaluating} & 100\% & 10\% \\
\hline
\end{tabular}
\end{table}

We can see that the RATE method~\cite{yadlowsky2025evaluating} correctly  detects HTE 100\% of the time for the HTE but no personalization simulation, and no HTE 90\% of the time for the no HTE simulation setting. The Imai-Li HTE test \cite{imai2025statistical} is more conservative and only reports significance 25\% of the 20 runs, and reports significance 5\% of the time when there is no HTE. However, reporting significance for HTE is not sufficient for having personalization, as shown in these simulation instances above. In both of the simulated data generation process settings, we run our KPT method for 100 repeated splits and our KPT method correctly does not detect any personalization effect, returning a reported test statistic of 0.\footnote{We did this for 20 runs, and the data generation process was the same, but the selected seed was accidentally not identical to the seed used for dataset generation for the results reported for RATE and Imai and Li's HTE method. We expect this to have minimal impact-- the primary goal of this section is to illustrate that HTE is necessary but not sufficient for personalization effects to be non-zero, and that methods to detect HTE are not sufficient for this latter task.}

\subsection{Impact of Choice of Number of Splits $S$}
\label{sec:vary_number_splits}
Our theory holds for any fixed number of splits $S$, and the repeated permutation approach is similar to the one in double/debiased machine learning (DML)~\cite{chernozhukov2018double}, which shows results for a fixed single choice of the number of splits $S$. However, in practice there may be a tradeoff between computational cost and statistical stability. To explore this, here we consider the impact on the choice of $S$ on the semi-synthetic Jobcorps and Joke settings.

Specifically, we run the full KPT procedure with $S=200$, and then compute the aggregated personalization estimates and confidence intervals using only the first 5, 10, 25, 50, 75, 100, 150 or all 200 splits. The results are shown in Figures~\ref{fig:jobcorps_varyS} and \ref{fig:joke_varyS}. 
\begin{figure}[!htb]
    \centering
    \includegraphics[width=\linewidth]{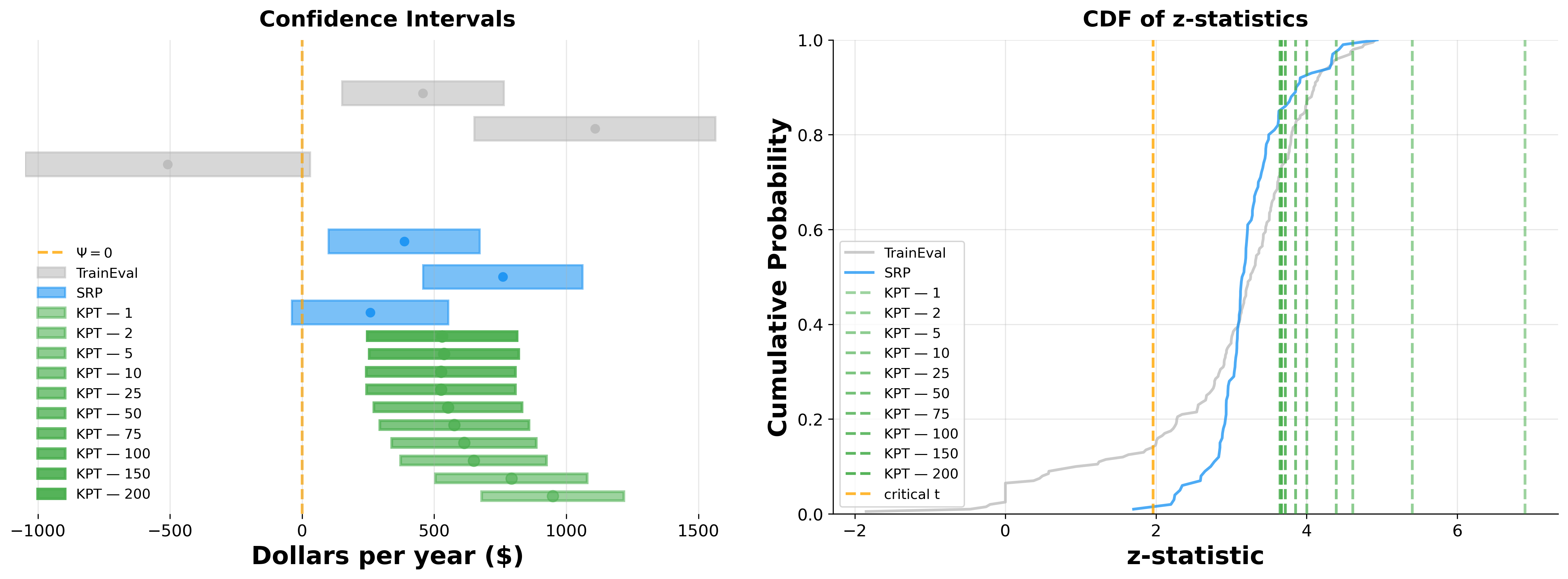}
    \caption{Job Corps dataset results for different choices of $S$}
    \label{fig:jobcorps_varyS}
\end{figure}
\begin{figure}[!htb]
    \centering
    \includegraphics[width=\linewidth]{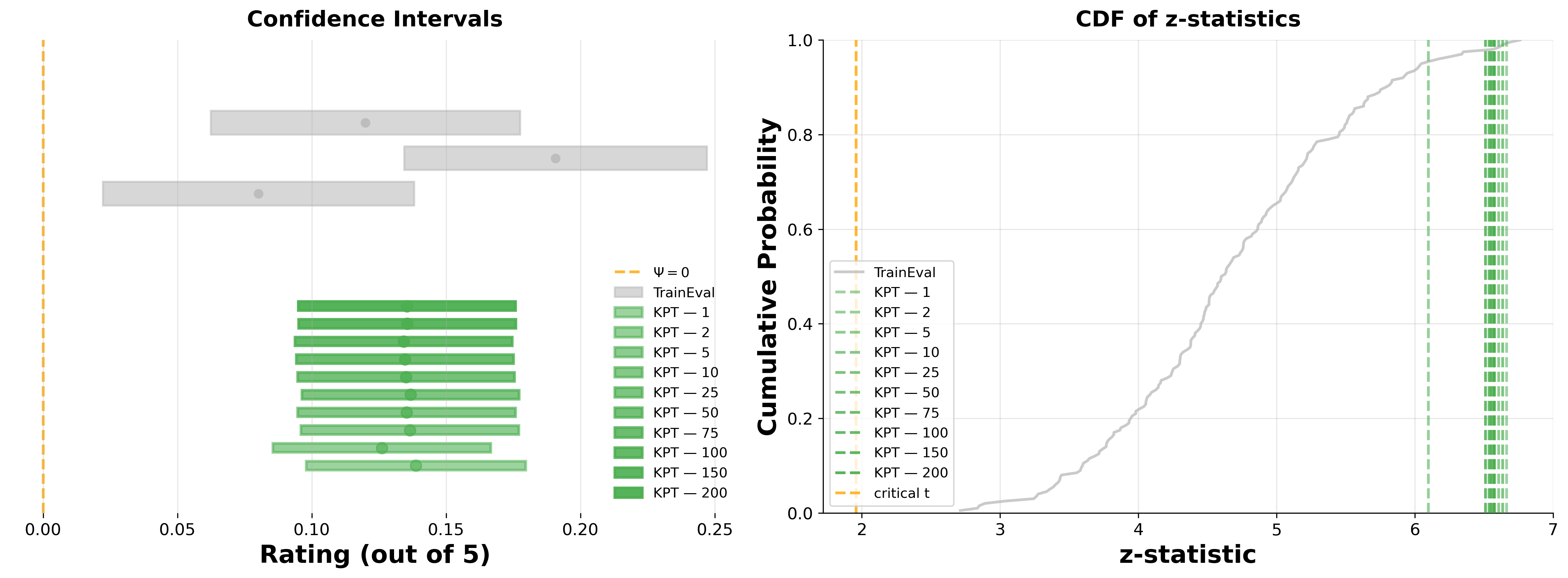}
    \caption{Joke dataset results for different choices of $S$}
    \label{fig:joke_varyS}
\end{figure}

This result does not show how much our estimate varies across repeated runs of a fixed number of splits (e.g. 5 runs vs 5 runs), but it does show that for $S=75$ or larger for Job Corps and $S=5$ or larger for Joke,  the estimates and confidence intervals stabilize. We do not view this as practical guidance for selecting $S$, as this analysis does suggest that the selected value may depend on the domain. However, it does suggest that relatively small choices for $S$ (like $S=10$) may already yield a CI that largely overlaps with the one obtained by much larger choices (e.g. $S=200$) of $S$.

\subsection{Exploration of Impact of Size of Dataset $N$}
\label{sec:vary_size_dataset}
Finite-sample behavior and the relationship between sample size and detection power is an interesting area. In Figures~\ref{fig:s3_new} and ~\ref{fig:test_efficient_large} we vary the personalization effect size while holding sample size fixed-- in doing so, these figures help characterize the regimes in which the test has sufficient power. 

Here we also sample subsets of the Joke dataset at varying sizes from $N=500$ to the full dataset $N=48445$ and recompute the personalization estimate and associated confidence intervals and test statistics. The results are shown in Figure \ref{fig:joke_varyN}. 
\begin{figure}[!htb]
    \centering
    \includegraphics[width=\linewidth]{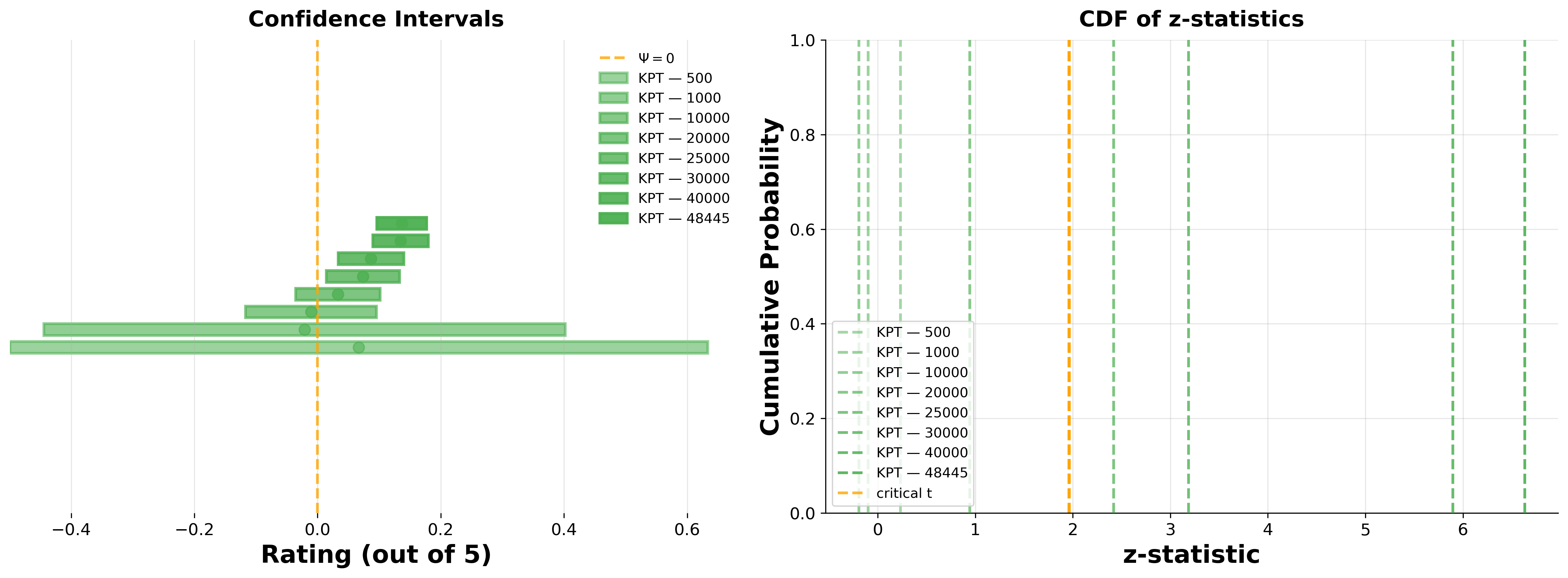}
    \caption{Joke results for varying sample sizes}
    \label{fig:joke_varyN}
\end{figure}
We find for the particular subsets selected, our approach shows a significant personalization effect for $N=25000$ and above. 

In general the needed sample size is impacted by the data generation process and the choice of policy class and other nuisance modeling classes and representational choices (e.g. reward outcome model class and the representation of features). This is consistent with insights from the PAC learning and bandit literature, where sample complexity depends on both the reward gap and the richness of the function class \cite{fiez2019sequential,li2022instance,zanette2021design}.
With function approximation or in observational data, the needed dataset size is even more complex to compute, as the effective sample size is influenced by nuisance parameters, such as the propensities, and their expressive capacity. Deriving a formal procedure for and analysis of experimental design including sample size and power calculations is a very interesting direction for future work, beyond the scope of the current paper.

% **Add figures

\subsection{Non-Unique Optimal Actions and Decision Policies}
\label{sec:null}
Our theory for Type I error relies on Assumption~\ref{assump:fast_best_arm_learner} (fast best arm learning) and Assumption~\ref{ass:consistency} and our semi-parametric efficiency result relies on Assumption~\ref{assump:fast_best_arm_learner},~\ref{ass:consistency} and~\ref{assump:consistent_regret}. These assumptions can fail when the best overall intervention is not unique. In particular, Assumption~\ref{assump:fast_best_arm_learner} may not hold in settings where there is no unique best intervention that maximizes the average utility in the population, such as in the binary intervention case where there is zero average treatment effect. Similarly, if there is a substantial proportion of the population whose expected outcomes are identical for multiple interventions, there will not be a unique optimal personalized policy, and consistent policy learning and fast policy learning will generally not be possible. Such a scenario could arise if, for example, in the binary treatment case if there is a significant part of the population that are non-responders to the treatment (and their conditional average treatment effect is zero).

We first consider the binary intervention setting where there is zero average treatment effect (and Assumption~\ref{assump:fast_best_arm_learner} is violated). There are two main cases of interest here: (Case 1) there is also zero personalization effect, and we wish to ensure that we do not erroneously detect a personalization effect, and (Case 2) there is a positive personalization effect, and we care about power of our procedure to detect this when the average treatment effect is zero and there is not a unique best overall arm. We begin with pragmatic procedures and empirical findings, followed by a discussion on theoretical analysis.
\\
\\
\noindent{\textbf{Case 1: Zero ATE, Zero Personalization Effect}}\\
We refer to this setting as the ``strong null" setting, where there is no unique best overall action and no personalization effect, so $H_0$ holds. First we note that in the binary action case, this implies that $r(x,a_1) = r(x,a_2)$ $\forall$ $x \in \X$, since if there existed a $x$ with positive probability $p(x) > 0$ where one intervention had a higher reward, there would be a positive personalization effect. Therefore, in this setting, there is also zero heterogeneous treatment effects. Therefore to guard against false personalization effects in the binary treatment case under this strong null setting, a sufficient process would be to first use an HTE method with Type I guarantee, such as the results for the RATE approach~\cite{yadlowsky2025evaluating}. 
% provides valid asymptotic Type I guarantees for detecting HTE. 
In the strong-null binary intervention setting, such a test should find no evidence of HTE with desired Type I guarantees, and hence no evidence of personalization.
% Such an HTE test will return no evidence of HTE (with desired Type I guarantees) in this strong null case binary intervention case, and therefore no personalization effect can be present.

In addition, empirically, we find our KPT estimator still seems to perform well in the strong null regime. In the Threshold-small simulation with true effect size equal to 0, the test has a valid Type-I error (less than 5\%), as shown in Figure~\ref{fig:s3_new}. In addition, in the real dataset case of MOOC Education where there was previously reported a null effect of any individual intervention being dominant~\cite{kizilcec2020scaling}, we also did not reject the null hypothesis of zero personalization. This result is consistent with prior related findings~\cite{kizilcec2020scaling} rather than producing spurious detections. Therefore, we expect that our approach is still likely to be valid in practice.
\\
\\
\noindent{\textbf{Case 2: Zero ATE, Positive Personalization Effect}}\\
We examine the performance of our KPT experimentally in this setting through a simulation, where the data generation process is defined for the average treatment effect to be zero but there is a positive personalization effect. The setup is similar to the Threshold-small simulation case, but where
\begin{itemize}
\item 50\% of contexts (10 contexts) have $r(x,1)-r(x,0) = \Delta$
\item 50\% of contexts (10 contexts) have $r(x,1)-r(x,0) = -\Delta$
\item  The ATE is 0
\item  The personalization effect is $\Delta$
\end{itemize}
As in the other simulation settings, contexts are sampled from a uniform distribution, actions are also sampled uniformly at random, and reward outcomes are generated with additive Gaussian noise ($\sigma=0.2$). 

For the data and training we use the same setup as our other simulation results: 
\begin{itemize}
\item For a fixed instance (fixed $\Delta$) we sample 100 datasets of size $N=3000$
\item We use a random forest regressor for reward learning, a doubly robust policy forest learner from the econml package for policy learning, and learn the single-best action by selecting the action with the maximum reward estimate (aka we use a random forest regressor to build a reward model only on the data used for best action/policy learning, and then take the action with the maximum expected value over the contexts in this set).
\end{itemize}

We consider two values of $\Delta=0.05$ and $\Delta= 0.2$. Our KPT finds a significant positive impact of personalization 95\% of the datasets for the smaller gap, and in 100\% of the cases for the larger gap (see Table~\ref{tab:zero_ate_pos_personalization}). This suggests that though our the theoretical analysis we conduct in this paper requires fast best arm learning, which is violated in some cases such as when the ATE is zero, that empirically, as we might expect, our method can still perform well and detect personalization effects in such settings. 

% \begin{table}[h!]
% \centering
% \begin{tabular}{|c|c|c|}
% \hline
% $\Delta$ & KPT &  TrainEval \\
% \hline
% 0.05 & 95/100 & 93/100 \\
% 0.2 & 100/100 & 100/100\\
% \hline
% \end{tabular}
% \caption{Number of datasets (out of 100) where $p < 0.05$ for different values of $\Delta$ for the simulation where ATE$=0$ and the personalization effect is $\Delta$.}
% \label{tab:zero_ate_pos_personalization}
% \end{table}

\begin{table}[ht]
\centering
\begin{tabular}{| c | c | c |}
\hline
Delta & KPT \# Datasets with $p < 0.05$ & TrainEval \# Datasets with $p < 0.05$ \\
\hline
0.05 & 95/100 & 93/100 \\ \hline 
0.2 & 100/100 & 100/100 \\ \hline 
\hline
\end{tabular}
\caption{Number of datasets (out of 100) where $p < 0.05$ for different values of $\Delta$ for the simulation where ATE=0 and the personalization effect is $\Delta$.}
\label{tab:zero_ate_pos_personalization}
\end{table}

\noindent{\textbf{Theoretical analysis in the case of non-unique overall best interventions and non-unique optimal personalized policy (Zero Gap)}}

When the best overall action or best personalized policy is not unique, we are faced with a nonregular inference problem.  This regime is known to be difficult in asymptotic and semiparametric statistics
% : when the optimal action is not uniquely separated, the target functional becomes nonregular, plug-in learners can be unstable, 
and standard efficiency arguments may fail due to the instability of plug-in estimators. This type of phenomenon is discussed in the literature on semiparametric efficiency and irregular problems; see, for example, Chapters 25–26 of ~\cite{van2000asymptotic}  and~\cite{bickel1993efficient}. The gap assumption ensures regularity and enables sharp asymptotic guarantees.

%% Start here, discuss smoothing
One way to address this challenge is through smoothing, such as by using softmax to define a probabilistic decision policy. Very recent  work~\cite{zhang2026replicable} shows that such approaches can, for example, stabilize the learning procedure of bandit algorithms to ensure replicable asymptotics. 

Another relevant line of work is subagging (subsample aggregation) in which many random subsamples of the data are used and learned on, and an estimate is computed by averaging across the results. Prior work~\cite{shi2020breaking} considers subagging where a personalized policy is learned on one random subsample of the data, and its value is evaluated on the remainder, the whole process is repeated many times, and those estimates are combined to compute an estimate and confidence intervals over the optimal value function. They prove this approach can handle when the optimal action is not unique for many contexts, and yields valid confidence intervals. 

Subagging improves stability by averaging learners trained on subsamples of size strictly smaller than $n$, which can mitigate the instability that arises when large subsets of the context space have non-unique best actions. In this sense, it plays a role similar to smoothing: it regularizes the learner and can restore well-behaved asymptotics even when the optimal action is not uniquely identified.

However, a key limitation of both smoothing and subagging approaches is that they typically rely on effective sample sizes strictly smaller than $n$ (either through explicit subsampling or through implicit smoothing/regularization). As a result, we are not aware of methods that would preserve semiparametric efficiency when a positive personalization effect exists. Intuitively, by stabilizing the decision boundary, these methods trade off variance for bias or reduced sensitivity, which slows down the rate at which true personalization benefits can be detected. 

Our method takes a complementary stance. By imposing the gap condition (MG), we operate in a regime where the problem is regular, which allows us to achieve fast rates and semiparametric efficiency when personalization benefits are present. In contrast, smoothing and subagging methods prioritize robustness when there exists no unique best overall intervention or a substantial part of the population with no unique best intervention, but may lose power. It is an interesting direction for future work to explore alternate assumptions or approaches that may offer similarly strong guarantees to our approach under weaker assumptions like lack of outcome gaps.

\end{document}